\documentclass[draftclsnofoot,onecolumn,12pt]{IEEEtran}
\usepackage{algorithm,amsbsy,amsmath,amssymb,epsfig,bbm,mathrsfs,multirow,amsthm}
\usepackage{array,multirow,graphicx}
\usepackage{graphicx}
\usepackage{graphics}
\usepackage{epstopdf}
\usepackage{color}
\usepackage{bbm}
\usepackage{comment}
\usepackage{cite}
\usepackage{algpseudocode}
\usepackage{mathtools,xparse}
\usepackage{subcaption}
\usepackage{mathtools}
\definecolor{mygray}{gray}{0.6}
\definecolor{myblue}{rgb}{0.8,0.85,1} 
\usepackage{array}
\newcolumntype{L}[1]{>{\raggedright\let\newline\\\arraybackslash\hspace{0pt}}m{#1}}
\newcolumntype{C}[1]{>{\centering\let\newline\\\arraybackslash\hspace{0pt}}m{#1}}
\newcolumntype{R}[1]{>{\raggedleft\let\newline\\\arraybackslash\hspace{0pt}}m{#1}}

\usepackage{tabularx, boldline}
\usepackage{cellspace}

\usepackage{grffile}
\DeclarePairedDelimiter{\ceil}{\lceil}{\rceil}

\DeclarePairedDelimiter{\norm}{\lVert}{\rVert}

\newcommand{\beq}{\begin{equation}}
\newcommand{\eeq}{\end{equation}}
\newcommand{\bitm}{\begin{itemize}}
\newcommand{\ba}{\begin{array}}
\newcommand{\ea}{\end{array}}
\newcommand{\eitm}{\end{itemize}}
\newcommand{\beqn}{\begin{eqnarray}}
\newcommand{\eeqn}{\end{eqnarray}}
\newcommand{\beqno}{\begin{eqnarray*}}
\newcommand{\eeqno}{\end{eqnarray*}}
\newcommand{\bma}{\begin{displaymath}}
\newcommand{\ema}{\end{displaymath}}
\newcommand{\bnu}{\begin{enumerate}}
\newcommand{\enu}{\end{enumerate}}
\newcommand{\bce}{\begin{center}}
\newcommand{\ece}{\end{center}}
\newcommand{\btb}{\begin{tabular}}
\newcommand{\etb}{\end{tabular}}

\newtheorem{theorem}{\textbf{\textsc{Theorem}}}

\begin{document}
\title{\huge Dynamic Network Service Selection in Intelligent Reflecting
Surface-Enabled Wireless Systems: Game Theory Approaches}

\author{ 
\IEEEauthorblockN{Nguyen Thi Thanh Van, Nguyen Cong Luong,
Feng Shaohan,
Huy T. Nguyen,
Kun Zhu, Thien Van Luong, and Dusit Niyato
}

}

\maketitle
\begin{abstract}
In this paper, we address dynamic network selection problems of mobile users in an Intelligent Reflecting Surface (IRS)-enabled wireless network. In particular, the users dynamically select different Service Providers (SPs) and network services over time. The network services are composed of IRS resources and transmit power resources. To formulate the SP and network service selection, we adopt an evolutionary game in which the users are able to adapt their network selections depending on the utilities that they achieve. For this, the replicator dynamics is used to model the service selection adaptation of the users. To allow the users to take their past service experiences into account their decisions, we further adopt an enhanced version of the evolutionary game, namely fractional evolutionary game, to study the SP and network service selection. The fractional evolutionary game incorporates the memory effect that captures the users' memory on their decisions. We theoretically prove that both the game approaches have a unique equilibrium. Finally, we provide numerical results to demonstrate the effectiveness of our proposed game approaches. In particular, we have reveal some important finding, for instance, with the memory effect, the users can achieve the utility higher than that without the memory effect.
\end{abstract}

\section{Introduction}

Intelligent Reflecting Surface (IRS) is an emerging technology for the development of the next-generation wireless networks~\cite{wu2019towards}, \cite{di2019smart}. IRS consists of passive elements that can reflect incident signals by intelligently adjusting their phase-shifts corresponding to wireless channels. The signals reflected by the IRS are added constructively with non-reflected signals, i.e., the Line-of-Sight (LoS) signals, at the user receivers to boost the received signal power and enhance the data rate at the users. As a result, IRS has recently been proposed to be integrated with next-generation wireless
technologies such as terahertz (THz) communications. The reason is that THz communication is able to provide data transmission rate up to terabit per second (Tbps), but this technology is limited in distance due to the fact that the THz waves are vulnerability to blockage and have severe path attenuation. For this, network service providers (SPs) deploy multiple IRSs in the THz networks to extend the network coverage and enhance the Quality of Service (QoS) of the mobile users. This results in a high density of the IRSs and base stations (BSs) in the THz networks, and the mobile users will more frequently handover among IRSs, BSs, and even SPs to achieve their desired QoS with low cost. In this case, the dynamic network selection of the mobile users in the THz networks becomes critical. 

Although there are some works, i.e.,~\cite{chen2019sum},~\cite{ma2020intelligent}, \cite{ma2020joint}, \cite{hao2020robust}, and~\cite{Ma2020archive}, that have recently investigated the IRS-enabled THz networks, they do not focus on the network selection of the mobile users. In particular, the work in \cite{ma2020joint} is proposed to determine phase shifts of IRS to maximize the data rate. Extending the work in \cite{ma2020joint}, the work in~\cite{chen2019sum} aims to jointly optimize the IRS phase shifts, beamforming at the BS, and spectrum allocation to maximize the data rate. Similar to~\cite{chen2019sum}, the work in \cite{hao2020robust} aims to maximize the overall network throughput by jointly optimizing the phase-shits of the IRS and beamforming at the BS.  


Also, there are some works that have recently investigated the network service selection. However, these works are not considered under the THz networks, and they do not account fir the dynamics of the network service selection. In particular, the authors in~\cite{gao2020stackelberg} and \cite {gao2020resource} consider an IRS-enabled network in which the BS is owned by an SP and the IRS belongs to a different SP. The Stackelberg game is then adopted to maximize the individual utilities of the SPs. Accordingly, the SP of IRS as the leader offers reflection modules as network resources and decides their prices. Note that the authors consider the allocation of reflection modules, i.e. instead of all the reflection elements, to the users since triggering all the reflection elements frequently results in an increased latency of adjusting phase-shift as well as the implementation complexity. Given the price, the SP of BS as the follower selects the best trigger reflection modules and determines their phase shifts and the transmit beamforming at the BS. Although the proposed game approach is demonstrated by the simulation results to be effective, the dynamics of the network service selections are not modeled in the work. Therefore, a more effective approach needs to be adopted to study the dynamic network service selection.

Evolutionary game~\cite{hofbauer2003evolutionary} as an effective tool can be adopted to study the dynamic selection and adaptation decision of a population of agents or players. This game has significant advantages~\cite{han2012game},\cite{quijano2017role} compared with the traditional games. In particular, the traditional games allow players, e.g. the mobile users in this work, to choose the desired solution immediately, while in the evolutionary game, the players are able to gradually adjust their strategies until they achieve a refined equilibrium solution. Especially, the evolutionary game is able to track and capture the strategy dynamics of the players as well as the strategy trends and behaviors of the players over time. Therefore, the evolutionary game has been widely used to study the dynamics of selection behaviors of users. The authors in \cite{liu2018evolutionary} investigate the mining pool selection of miners in a blockchain system. Accordingly, the miners compete to solve a crypto-puzzle to win the reward of mining new blocks. Due to the difficulty of the crypto-puzzle, the miners are willing to select mining pools for their secure stable profits. To study the dynamic selection of mining pool of the miners, the evolutionary game with the replicator dynamics is adopted for modeling the strategy evolution of the miners. It is demonstrated by both theory analysis and simulation results that in the case of two mining pools, the evolutionary game approach exists a unique Nash equilibrium at which no miner has an incentive to switch its pool selection since this will undermine some other miner's utility.

 Next-generation wireless networks are expected to deploy different wireless access technologies, and the wireless network technology selection of the mobile users is crucial that impacts their QoS. For this, the evolutionary game is adopted to effectively study the dynamic network selection of the mobile users as proposed in~\cite{niyato2008dynamics}. The evolutionary game is also adopted to model the network service selection of secondary transmitters in a backscatter-based cognitive network~\cite{gao2019dynamic}.~In particular, the system model includes multiple access points serving multiple secondary transmitters. Each access point provides three network services, namely harvest-then-transmit (HTT), backscatter, and HTT-backscatter, to the mobile users. The secondary transmitters receive different utilities when choosing network services from different access points. To model the access point and service adaptation of the secondary transmitters, a series of ordinary differential equations is used to formulate the replicator dynamic process. Both the theory and numerical results show that the evolutionary game approach exists a unique equilibrium at which the secondary transmitters achieve the same utility even if they select different access points and network services. This demonstrates that with the evolutionary game, the users can adapt their selections gradually to reach the equilibrium. Especially, the complexity of algorithm to implement the evolutionary game is low. In particular, as analyzed in~~\cite{gao2019dynamic}, the complexity of strategy adaptation at each secondary transmitter is $O(1)$, that is suitable for the dynamic strategy selections of the users. 
 

Given the aforementioned advantages, in this paper, we adopt the evolutionary game theory to study the dynamic service selection strategies of mobile users in the IRS-enabled terahertz network. The considered network consists of multiple SPs that deploy BSs along with multiple IRSs to provide network services to multiple mobile users. In particular, the SPs offer combinations of IRS and transmit power resources as network services that the users can select for their data transmissions. The network is thus considered to be a user-centric network. To satisfy different QoS requirements of the mobile users, similar to \cite{gao2020reflection},~\cite{gao2020stackelberg}, we assume that the SP divides its IRSs into reflection modules. Furthermore, the SPs have different transmit power levels that the mobile users can select. To model the SP and service adaptation of the users in the network, we leverage the replicator dynamic process that is expressed as a series of ordinary differential equations. Note that with the classical evolutionary game, the users only consider the instantaneous utility for their decision-making, i.e., the SP and service adaptation. This is not natural and practical due to the fact that the user is typically aware of its past network service experience when making the network selection. In other words, the awareness of the users' memory needs to be accounted. To address the limitation of the classical evolutionary game, we further adopt the fractional evolutionary game as a memory-aware economic process. The fractional evolutionary game enables the users to incorporate the instantaneous and past experiences of the users for their decisions. Both theoretical analysis and simulation results show the effectiveness of the proposed game approaches.







The main contributions of the paper include the followings:
\begin{itemize}
\item We consider the IRS-enabled terahertz network in which multiple SPs deploy IRSs to serve the mobile users. The SP offers IRS and transmit power resources as network resources to the mobile users. Different combinations of the network resources constitute different network services provided by the SPs. The network is a user-centric network in which the users can select and adapt the network services provided by the SPs over time to achieve their desired utility. The IRS-enabled terahertz network introduces new transmission scheme that makes the utility function of the users more complicated, and a new solution is required to model the network service adaptation of the users. To model the network service adaptation, we adopt the evolutionary game. 
\item We consider the scenario in which the users use the delayed information for their decisions. Such a delay can cause instability in the decision making process. In this scenario, the delayed replicator dynamics is adopted to model the SP and network service adaptation. We analyze the equilibrium region of the delayed replicator dynamics and show in the simulation results that the evolutionary game approach still reach an equilibrium with a small delay. 
\item To capture the users' memory on their decision-making, we incorporate the users' memory effect to reformulate the SP and network service selection problem into a fractional evolutionary game. We then compare the network selection strategies of the users between the classical game and fractional game. 
\item We theoretically prove that the fractional evolutionary game processes a unique equilibrium. Then, the simulation results with the direction field of the replicator dynamics are provided to verify the stability of the equilibrium. 
\item We provide performance evaluation to demonstrate the consistency with the analytical results and to validate both the proposed game approaches. The performance comparison between the two game approaches are also discussed and analyzed. 
\end{itemize}

The rest of the paper is organized as follows. In Section~\ref{system_model}, we present the IRS-enabled terahertz system and utility functions of the mobile users. In Section~\ref{classical_evol_game}, we formulate the dynamic SP and network service selection problem as the classical evolutionary game and analytically derive the stability region of the delayed replicator dynamics. In Section~\ref{frac_game}, we reformulate the dynamic SP and network service selection problem as a fractional evolutionary game, followed by the proofs of the existence, uniqueness, the stability of the equilibrium of the game. The simulation results and discussions are presented in Section~\ref{perform_eval}, and the conclusions are given in Section~\ref{conclu}.

\section{System Model}
\label{system_model}
This section presents the system model, channel models, and utility functions of the users in the network. Typical notations used in this paper are summarized in Table~\ref{table:major_notation}. 
\subsection{Network Model}

\begin{table}[h!]
\caption{List of frequency symbols used in this paper.}
\label{table:major_notation}
\footnotesize
\centering
\begin{tabular}{l|l}
\hline\hline
\textbf{Notation} & \textbf{Description}   \\ [1ex]
\hline
 $M,N$ & Number of SPs, number of users           \\ 
  \hline
 $I_m, J_{m,j}$ & Number of IRS of SP $m$, power level $j$ offered by SP $m$        \\ 
   \hline
 $K_m$ & Number of reflection elements of each IRS of SP $m$        \\ 
\hline
 $Q_m$ & Number of modules of each IRS of SP $m$     \\ 
 \hline
 $\theta_{m,l,k,e}$ & Phase shift of element $e$ of subset $k$ in IRS $l$ of SP $m$      \\ 
  \hline
 $N_{m,l,k,j}$ & Number of users selecting power level $J_{m,j}$ and subset $k$ in IRS $l$ of SP $m$    \\ 
\hline
 $p_{m,l,k,j,i}$ & Probability that user $i$ selects power level $J_{m,j}$ and subset $k$ in IRS $l$ of SP $m$     \\ 
  \hline
 $\mathbf{h}_{m,i}$ & Channel from BS $m$ to user $i$  \\ 
   \hline
 $\mathbf{h}_{m,l,k,i}^{\rm{IU}}$ & Channel from subset $k$ in IRS $l$ of SP $m$ to user $i$   \\ 
    \hline
 $\mathbf{G}_{m,l,k}$ & Channel from BS $m$ to subset $k$ in IRS $l$ of SP $m$   \\ 
 \hline
 $\mathbf{w}_{m,j,i}$ & Beamforming vector associated with user $i$ that selects power level $J_{m,j}$ of SP $m$   \\ 
 \hline
 $f, \mu,\beta$ & Carrier frequency, learning rate, order of the Caputo fractional derivative     \\ 
  \hline
 $\gamma_m^I, \gamma_m^P$ & Price per IRS element, price per power unit   \\ 
\hline
\end{tabular}
\label{table:parameters}
\end{table}

\begin{figure}[h]
 \centering
\includegraphics[width=0.6\linewidth]{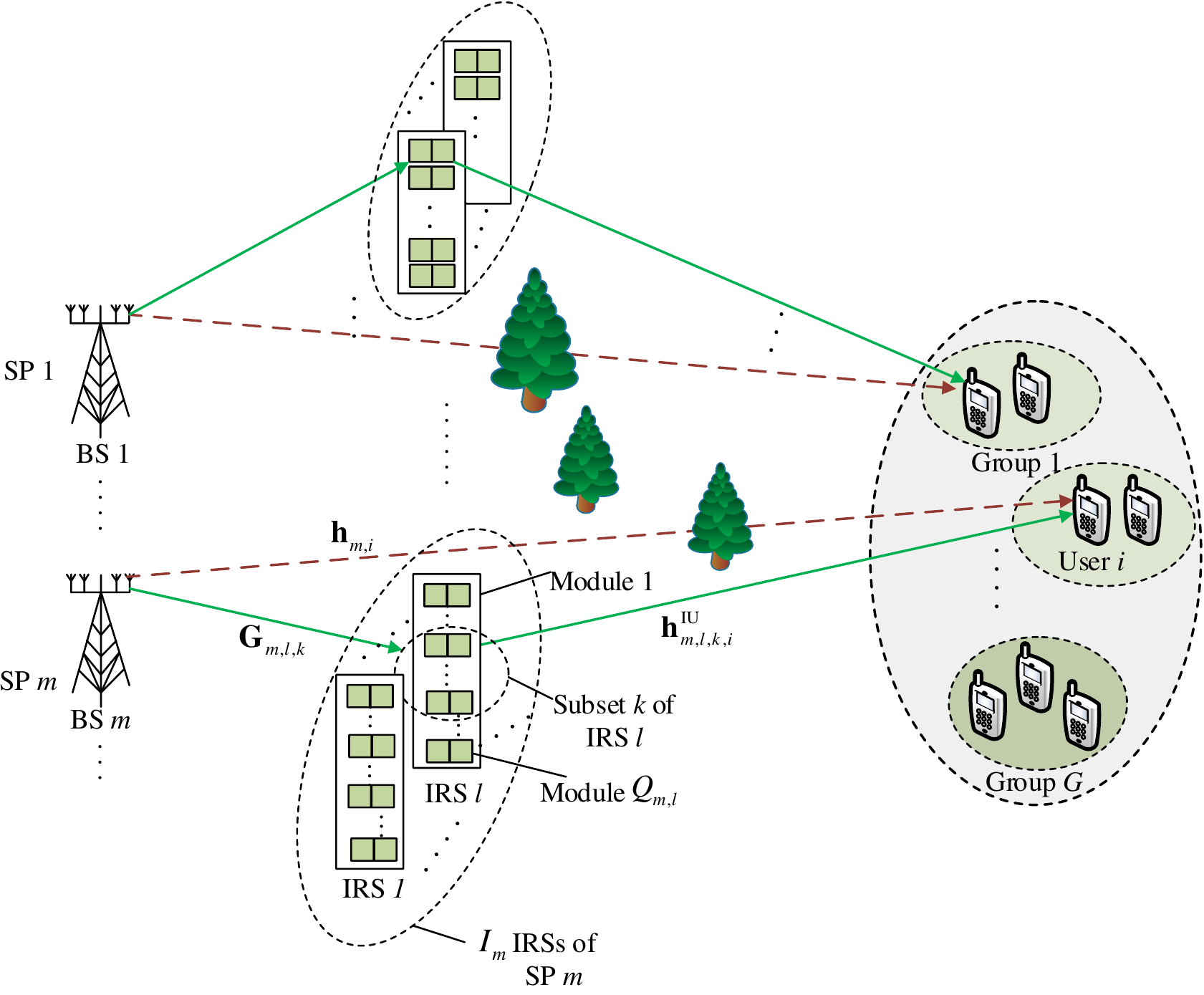}
 \caption{A system model with multiple SPs with multiple IRSs serving multiple users.}
  \label{IRS_model_selection}
\end{figure}
The system model is an IRS-enabled terahertz MIMO system as shown in~Fig.~\ref{IRS_model_selection}. The system model consists of a set $\mathcal{M}$ of $M$ SPs and a set $\mathcal{N}$ of $N$ single-antenna users. Without loss of generality, each SP $m \in \mathcal{M}$ deploys a BS, i.e., BS $m$, that is equipped with $L_m$ antennas. Denote $B_m$ as the bandwidth allocated to SP $m$, i.e., BS $m$. To provide flexible services to the users, BS $m$ has a set $\mathcal{P}_m$ of $P_m$ power levels, denoted by $\{J_{m,1},\ldots,J_{m,P_m}\}$, that the users can select for their transmissions. Note that the assumption of the discrete power levels is reasonable since in real networks, transmit power control algorithms choose steps of power increment/decrement~\cite{wu2001distributed}. Also, the number of power levels can be increased to match with the real system implementation without causing much more complexity. We also assume that $J_{m,1} <J_{m,2}< \dots < J_{m,P_m}$, where $ J_{m,P_m}$ is the maximum power of BS $m$.~To improve the QoS for the users, SP $m$ deploys a set $\mathcal{I}_m$ of $I_m$ IRSs. Let $l \in \mathbb{Z} $ denote the index of IRS, and $1\leq l \leq \max_{m} \{I_m\}$. We assume that IRSs belonging to the same SP have the same size, and IRS $l$ of SP $m$ has $K_{m,l}= K_{m}$ reflection elements. IRS $l$ of SP $m$ is divided into $Q_{m,l}=Q_m$ modules that are controlled by parallel switches. Each module in IRS $l$ consists of $E_{m,l}=E_m$ elements. Note that during a time slot, one BS-IRS pair of the corresponding SP can serve multiple users, but the user is associated with one BS-IRS pair. Moreover, the user can select one or multiple modules, i.e., a subset of modules, of the selected IRS. A network service is defined as a combination of a power level and a subset of modules. In general, the data throughput achieved by the user, say user $i$, depends on 1) the power level that the user selects, 2) the bandwidth allocated to the BS that the user selects, 3) the location of the selected IRS, 4) the number of modules of the selected IRS, and 5) the interference caused by other users selecting the same BS with user $i$. Note that the data throughput does not depend on indexes of the modules of the selected IRS. As such, each IRS $l$ of SP $m$ has a set $\mathcal{Q}_{m,l}$ including $Q_{m,l}$ of potential subsets of modules that the user can select, and subset $k, 1 \leq k \leq Q_{m,l},$ of the IRS has $k$ modules. Denote $\boldsymbol{\Theta}_{m,l,k}$ as the phase-shift matrix corresponding to the subset that the user selects, i.e., subset $k$ of IRS $l$ of SP $m$. Then, $\boldsymbol{\Theta}_{m,l,k}$ is a diagonal matrix in which its main diagonal consists of phase-shifts of $kE_{m,l}$ reflection elements of IRS $l$ of SP $m$. In particular, we have $\boldsymbol{\Theta}_{m,l,k}=\text{diag}(\theta_{m,l,k,1},\ldots,\theta_{m,l,k, kE_{m,l}})$, where $\theta_{m,l,k,e}$ is the phase-shift of reflection element $e$ of subset $k$ in IRS $l$ of SP $m$, $\theta_{m,l,k,e}=e^{j\varphi_{m,l,k,e}}, \varphi_{m,l,k,e} \in [0, 2\pi), e=\{1,\ldots,kE_{m,l}\}$. With the assistance of subset $k$ of IRS $l$ of SP $m$, the signal received at each user $i$ is the sum of 1) the received signal via the direct link, 2) the received signal via the IRS-assisted link, and 3) the intra-interference caused by other users that select the same BS with user $i$.~Let $\mathcal{N}_m$ denote the set of $N_m$ users selecting SP $m$, i.e., and also BS $m$. Then, the received signal at user $i$ when selecting subset $k$ of IRS $l$ and power level $J_{m,j}$ offered by SP $m$ is determined as follows:
\begin{align}
\notag
y_{i}=&\big{(}\mathbf{h}^{\text{H}}_{m,i} +   (\mathbf{h}^{\rm{IU}}_{m,l,k,i})^{\text{H}}\boldsymbol{\Theta}^{\text{H}}_{m,l,k}\mathbf{G}_{m,l,k}\big{)}\mathbf{w}_{m,j,i}s_{i}  + \\ 
&+  \sum_{n\in \mathcal{N}_m, n \neq i} \big{(}\mathbf{h}^{\text{H}}_{m,i} +   (\mathbf{h}^{\rm{IU}}_{m,l,k,i})^{\text{H}}\boldsymbol{\Theta}^{\text{H}}_{m,l,k}\mathbf{G}_{m,l,k}\big{)}\mathbf{w}_{m,j,n}s_{n}      +  \omega_i, 
\label{received_signal_user}
\end{align}
where $s_{i}$ is the data symbol intended to user $i$, $\mathbf{w}_{m,j,i}  \in \mathbb{C}^{L_m\times1}$ is the beamforming vector associated with $s_{i}$ containing power level $J_{m,j}$ that the user selects, $\mathbf{h}_{m,i} \in \mathbb{C}^{L_m\times1}$ is the channel from BS $m$ to user $i$, $\mathbf{h}^{\rm{IU}}_{m,l,k,i} \in \mathbb{C}^{K_{m,l}\times1}$ is the channel from subset $k$ of IRS $l$ of SP $m$ to user $i$, $\mathbf{G}_{m,l,k} \in \mathbb{C}^{K_{m,l} \times L_m}$ is the vector of channels from BS $m$ to subset $k$ of IRS $l$, and $\omega_i$ is the Gaussian noise at user $i$, $\omega_i \sim \mathcal{CN}(0,\sigma_0^2)$, where $\sigma_0^2$ is the variance. Similar to \cite{wu2019intelligent} and~\cite{zhou2020intelligent}, we assume that each BS $m$ has a perfect knowledge of channel state information (CSI) of the channels. Note that we aim to model network selection strategies of the users which is not influenced by this assumption. In fact, CSI estimation algorithms such as \cite{liaskos2019joint} can be used to obtain the full CSI at all the BSs with low training overhead. Since IRSs are typically deployed in static environments, we can also assume that the quasi-static flat-fading model or even static flat-fading model is applied for all channels~\cite{wu2019intelligent}. The signal-to-interference-plus-noise ratio (SINR) over bandwidth $B_m$ of user $i$ is defined as follows:
\begin{align}
\eta_{m,l,k,j,i} = \frac{\left |(\mathbf{h}^{\text{H}}_{m,i} + (\mathbf{h}^{\rm{IU}}_{m,l,k,i})^{\text{H}}\boldsymbol{\Theta}^{\text{H}}_{m,l,k}\mathbf{G}_{m,l,k})\mathbf{w}_{m,j,i}\right|^2} {\left |\sum_{n\in \mathcal{N}_m, n \neq i} \big{(}\mathbf{h}^{\text{H}}_{m,i} +   (\mathbf{h}^{\rm{IU}}_{m,l,k,i})^{\text{H}}\boldsymbol{\Theta}^{\text{H}}_{m,l,k}\mathbf{G}_{m,l,k}\big{)}\mathbf{w}_{m,j,n}\right|^2+  B_m\sigma_0^2}.
\label{SINR_user_i}
\end{align}

To remove the intra-interference among the users, the BSs can use time-division multiple access for their users. In this case, $\eta_{m,l,k,i}$ can be expressed by
\begin{equation}
\eta_{m,l,k,j,i}=\frac{\left|(\mathbf{h}^{\text{H}}_{m,i} +(\mathbf{h}^{\rm{IU}}_{m,l,k,i})^{\text{H}}\boldsymbol{\Theta}^{\text{H}}_{m,l,k}\mathbf{G}_{m,l,k})\mathbf{w}_{m,j,i}\right |^2} {B_m\sigma_0^2}.
\label{SINR_user_ii}
\end{equation}
\subsection{Channel Model}
In this section, we describe the channel models for the IRS-enabled THz network. The channel models for each user include the channel between the BS that the user selects and the user, and the cascaded channel of the IRS-aided link. Here, the cascaded channel of the IRS-aided link includes (1) the channel between the BS and the subset of IRS modules that the user selects and (2) the channel between the subset of IRS modules and the user. To model the channels in the THz network, we adopt the Saleh-Valenzuela channel model~\cite{lin2015indoor}. Without loss of generality, we model the channels between user $i$ when it selects BS $m$, subset $k$ of IRS $l$ of SP $m$. In particular, we determine models of channels $\mathbf{h}_{m,i}$, $\mathbf{h}^{\rm{IU}}_{m,l,k,i}$, and $\mathbf{G}_{m,l,k}$ that are given in (\ref{SINR_user_i}). For an ease of presentation, we remove the indices $m, l, k, i$ from the channels, and thus $\mathbf{h}_{m,i}$, $\mathbf{h}^{\rm{IU}}_{m,l,k,i}$, and $\mathbf{G}_{m,l,k}$ can be expressed by $\mathbf{h}$, $\mathbf{h}^{\rm{IU}}$, and $\mathbf{G}$, respectively. Also, we assume that BS $m$ has $L$ antennas, subset $k$ of IRS $l$ that the user selects has $K$ reflection elements. 

\subsubsection{BS-user channel}
The channel between the BS and the user is expressed by

\begin{equation}
\mathbf{h} = \kappa^{(0)}\mathbf{a}(\Phi^{(0)})+ \sum_{l=1}^{\mathcal{L}} \kappa^{(l)}\mathbf{a}(\Phi^{(l)}),
\label{Saleh_channel}
\end{equation}
where $\kappa^{(0)}\mathbf{a}(\Phi^{(0)})$ is the LoS element in which $\kappa^{(0)}$ is the gain and $\mathbf{a}(\Phi^{(0)})$ is the spatial direction, and $\kappa^{(l)}\mathbf{a}(\Phi^{(l)}), 1 \leq l \leq \mathcal{L}$, is one of $\mathcal{L}$ non-LoS (NLoS) elements. $\mathbf{a}(\Phi^{(l)})$ is the $L \times 1$ array steering vector corresponding to th-$l$ element that is defined as follows:
 \begin{equation}
\mathbf{a}(\Phi^{(l)})= \frac{1}{\sqrt{L}} \big{[} e^{-j2\pi\Phi^{(l)}(-\frac{N-1}{2})},\ldots, e^{-j2\pi\Phi^{(l)}(\frac{N-1}{2})}  \big{]},
 \end{equation}
where $\Phi^{(l)}$ is the spatial direction of the signal corresponding to component $l$, that is defined as $\Phi^{(l)}= \frac{d}{\lambda}\sin\xi^{(l)} $, where $\xi^{(l)} \in [ -\pi/2, \pi/2]$ is the angle-of-departure (AoD) of path $l$ corresponding to the BS and user, $\lambda$ is the signal wavelength, and $d$ is the distance between adjacent antennas of the BS or the distance between adjacent IRS elements of the IRS that is typically defined as $d=\lambda/2$. 

\subsubsection{BS-IRS-user channel}
The channel between the BS and the subset of modules of IRS that the user selects can be modeled as 
 \begin{equation}
\mathbf{G}= \sqrt{  \frac{LK}{\mathcal{L}_{\text{BS-I}}} }\sum_{l=1}^{\mathcal{L}_{\text{BS-I}}}\kappa_{\text{BS-I}}^{(l)}\mathbf{a}_{\text{I}}(\Phi^{(l)}_{\text{AOA}})\mathbf{a}_{\text{BS}}^{\text{H}}(\Phi^{(l)}_{\text{AOD}}),
 \end{equation}  
 where $\mathcal{L}_{\text{BS-I}}$ denotes the scattering paths between the BS and the subset of IRS that the user selects, $\kappa^{(l)}_{\text{BS-I}}$ is the complex gain of path $l$, and $\Phi^{(l)}_{\text{AOD}}$ and $\Phi^{(l)}_{\text{AOA}}$ are the spatial directions of path $l$ corresponding to the BS and the subset of IRS, respectively. We consider the BS's antennas and the IRS's reflection elements as uniform linear arrays (ULAs), and thus we can determine $\mathbf{a}_{\text{BS}}(\Phi^{(l)}_{\text{AoD}}) \in \mathcal{C}^{L}$ and $\mathbf{a}_\text{I}(\Phi^{(l)}_{\text{AOA}}) \in \mathcal{C}^{K}$ as follows: $\mathbf{a}_{\text{BS}}(\Phi^{(l)}_{\text{AOD}}) = \frac{1}{\sqrt{L}}\big{[} e^{-j2\pi\Phi^{(l)}_{\text{BS}}  (-\frac{L-1}{2}) },\ldots, e^{-j2\pi\Phi^{(l)}_{\text{BS}}(\frac{L-1}{2})}  \big{]}$, and $\mathbf{a}_{\text{I}}(\Phi^{(l)}_{\text{AOA}}) = \frac{1}{\sqrt{K}}\big{[} e^{-j2\pi\Phi^{(l)}_{\text{I}}  (-\frac{K-1}{2}) },\ldots, e^{-j2\pi\Phi^{(l)}_{\text{I}}(\frac{K-1}{2})}  \big{]}$. Here, $\Phi^{(l)}_{\text{I}}= \frac{d}{\lambda}\sin\xi^{(l)}_{\text{I}}$ and $\Phi^{(l)}_{\text{BS}}= \frac{d}{\lambda}\sin\xi^{(l)}_{\text{BS}} $, where $\xi^{(l)}_{\text{BS}} \in [ -\pi/2, \pi/2]$ and $\xi^{(l)}_{\text{I}} \in [ -\pi/2, \pi/2]$ are the AoD and the angle-of-arrival (AoA) of path $l$ corresponding to the BS and the subset of IRS, respectively.
 
 Similarly, we can determine the channel between the subset of IRS and the user as follows:
 
 \begin{equation}
\mathbf{h}^\text{IU} = \sum_{l=0}^{\mathcal{L}_\text{I-U}} \kappa^{(l)}_{\text{I-U}}\mathbf{a}_{\text{I-U}}(\Phi^{(l)}_{\text{I-U}}),
\label{Saleh_channel}
\end{equation}
where $\mathcal{L}_{\text{I-U}}$ denotes the scattering paths between the subset of IRS and the user, $\kappa^{(l)}_{\text{I-U}}$ is the complex gain of path $l$, and $\mathbf{a}_{\text{I}}(\Phi^{(l)}_{\text{I-U}}) = \frac{1}{\sqrt{K}}\big{[} e^{-j2\pi\Phi^{(l)}_{\text{I-U}}  (-\frac{K-1}{2}) },\ldots, e^{-j2\pi\Phi^{(l)}_{\text{I-U}}(\frac{K-1}{2})}  \big{]}$, where $\Phi^{(l)}_{\text{I-U}}= \frac{d}{\lambda}\sin\xi^{(l)}_{\text{I-U}} $ with $\xi^{(l)}_{\text{I-U}} \in [ -\pi/2, \pi/2]$ being the AoD of path $l$ corresponding to the subset of IRS.

\subsubsection{Path loss}
In THz communication systems, the non-LoS elements is proved to be much weaker than the LoS element, i.e., lower than $20$ dB~\cite{han2014multi}. Therefore, in the IRS-enabled THz network, we consider the LoS elements of the involved channels. Without loss of generality, we calculate the path loss of the LoS element between the BS and the user. This can be applied to calculating the path loss of the LoS elements between the BS and the subset of IRS as well as that between the subset of IRS and the user. The path loss of the LoS element, denoted by $\chi_\text{LoS}$, is a function of spreading loss and molecular absorption loss, denoted by $\chi_\text{abs}$. Then, the path loss is determined as~\cite{han2014multi}
\begin{equation}
\chi_\text{LoS}(f)= \chi_\text{spr}(f)\chi_\text{abs}(f)e^{-j2\pi f\tau_{\text{LoS}}},
\end{equation}
where $f$ is the carrier frequency, and $\tau_{\text{LoS}}=r/c$ is the LoS propagation time, $r$ is the distance between the BS and the user, and $c$ is the speed of light. $\chi_\text{spr}$ is the spreading loss that is determined by
\begin{equation}
\chi_\text{spr}(f)= \frac{c}{4\pi f r}.
\end{equation}

While, $\chi_\text{abs}(f)$ is the molecular absorption loss that is determined as follows:
\begin{equation}
\chi_\text{abs}(f)= e^{-0.5\zeta(f)r},
\end{equation}
where $\zeta(f)$ is the medium absorption coefficient that depends on carrier frequency and the
composition of the transmission medium at a molecular level. For example, given $f=2$ THz,  the molecular absorption coefficient $\zeta=2.3\times10^{-5}$ m$^{-1}$ for oxygen (O$_{2}$)~\cite{jornet2011channel}.

\subsection{Utility Functions}

This section presents the utility functions of the users when they select different SP and network services. There are totally $N$ users, $M$ BSs, $P_m$ power levels, and $\sum_{m=1}^MI_m$ IRSs in the network. Users selecting the same SP, the same subset of the IRS, and the same power level are grouped into a group. Thus, there are totally $G$ groups in the network, where $G=\sum_{m=1}^M\sum_{l=1}^{I_m}P_mI_mQ_{m,l}$.~Without loss of generality, we can assume that group $g,1 \leq g \leq G$, consists of users that select SP $m$, subset $k$ of IRS $l$, i.e., the corresponding phase-shift matrix $\boldsymbol{\Theta}_{m,l,k}$, and power level $J_{m,j}$. For this, we can denote $g$ as the combination of indexes $(m,l,k,j)$ for the expression simplification. Let $\mathcal{N}_{m,l,k,j}$, i.e., $\mathcal{N}_{g}$, be a set of $N_{m,l,k,j}$, i.e., $N_{g}$, of users in group $g$. We have $\sum_{m=1}^M\sum_{l=1}^{I_m}\sum_{k=1}^{Q_{m,l}}\sum_{j=1}^{P_m}N_{m, l,k,j}=N$, and user $i \in \mathcal{N}_{g}$ selects SP $m$, IRS $l$, $\boldsymbol{\Theta}_{m,l,k}$, and $J_{m,j}$ at a probability of $p_{m,l,k,j,i}=N_{m,l,k,j}/N$. As the expected number of the users selecting SP $m$, IRS $l$, $\boldsymbol{\Theta}_{m,l,k}$, and $J_{m,j}$ is $p_{m,l,k,j,i}N$ and BS $m$ adopts the time-division multiple access, each user in the group will access the channel with a probability of
$\frac{1}{p_{m,l,k,j,i}N}$ for every time slot. Therefore, the expected data rate that the user in the group can achieve is 
 \begin{equation}
 \overline{R}_{m,l,k,j,i}=\frac{B_m}{p_{m,l,k,j,i}N}
 \log_2\big{(}1+\eta_{m,l,k,j,i}\big{)},
 \label{expected_data_rate}
 \end{equation}
 where $\eta_{m,l,k,j,i}$ is given in~(\ref{SINR_user_ii}). Let $v_{m,l,k,j,i}$ denote the value of unit data to user $i$ in group $g$ when selecting SP $m$, subset $k$ of IRS $l$, and power level $J_{m,j}$. Denote $\gamma_m^I$ as the price per element in IRSs of SP $m$ and $\gamma_m^P$ as the price per unit power. The prices, i.e., $\gamma_m^I$ and $\gamma_m^P$, are set by SP $m$ that are constant. Since users in each group share the same resources, they should share the resource cost. Then, the utility of the user is given by
\begin{equation}
 u_{m,l,k,j,i} = v_{m,l,k,j,i} \overline{R}_{m,l,k,j,i} - \frac{\gamma_m^I \norm{\boldsymbol{\Theta}_{m,l,k}}_0 - \gamma_m^P J_{m,j}} {p_{m,l,k,j,i}N},
 \label{utility_user}
\end{equation}
where the $l_0$-norm is used to count the number of non-zero elements of a diagonal matrix that here refers to the number of active reflection elements of IRS $l$ of SP $m$ that the user selects.

\section{Evolutionary Game Formulation}
\label{classical_evol_game}
In this section, we leverage the evolutionary game to model the dynamic SP and network service selection
of the users. We prove that the game can achieve the evolutionary equilibrium at which no user has an incentive to change their network service strategy. 

\subsection{Game Formulation}
\label{classical_evol_game_form}
Each user in the network is able to adapt their network selection over time, and it can achieve different utility at different time points. Thus, by taking the SP and network service selection strategies, the expected or average utility of user $i$ at time $t$ is
\begin{equation}
\overline{u}_{i}=\sum_{m=1}^M\sum_{l=1}^{I_m}\sum_{k=1}^{Q_{m,l}}\sum_{j=1}^{P_m}p_{m, l,k,j,i} u_{m,l, k,j,i}.
 \label{utility_user_average}
\end{equation} 

To model the SP and service adaptation of the users, we leverage the replicator dynamic process that is expressed as a series of ordinary differential equations as follows~\cite{gao2019dynamic}, \cite{feng2020dynamic}:
\begin{align}
\notag
\dot{p}_{m,l, k,j,i}(t) =&\mu p_{m, l,k,j,i}(t)\left [u_{m, l,k,j,i}(t)-\overline{u}_{i}(t)\right ], \\
& m \in \mathcal{M}, l \in \mathcal{I}_m, k \in \mathcal{Q}_{m,l}, j \in \mathcal{P}_m, i \in \mathcal{N}_g, \forall t,
\label{replicator_evolu}
\end{align}
where $\dot{p}_{m, l,k,j,i}(t)$ represents the first derivative of $p_{m,l, k,j,i}$ with respect to $t$, and $p_{m,l, k,j,i}(t_0)=p^0_{m, l,k,j,i}$ is the initial strategy of the user in group $g$ at $t_0$. The factor $\mu$ is the learning rate of the users that evaluates the strategy adaptation rate. The replicator dynamics process given in (\ref{replicator_evolu}) represents the population strategy evolution of the users in the network. That is, the population of users evolves over time, and the game converges to the evolutionary equilibrium. This means that the users select an SP and its service with higher utility over time, and the evolutionary equilibrium can be defined as the set of stable fixed points of the replicator dynamics. 

To show the existence of the evolutionary equilibrium of the game as defined in (\ref{replicator_evolu}), we use the following theorem. First, we let $f_{g,i}(t,p_{g,i})=\mu p_{g,i}(t)\left[u_{g,i}(t)-\overline{u}_i(t)\right]$, where $g$ denotes as the combination of indexes $(m,l,k,j)$, and then we rewrite equation (\ref{replicator_evolu}) as follows:
\begin{equation}
\dot{p}_{g,i}(t)  =f_{g,i}(t,p_{g,i}), \text{ } p_{g,i}(t_0)=p^0_{g,i}, \text{ } g \in \{1,\ldots,G\}, i \in \mathcal{N}_g.
\label{replicator_evolu_2}
\end{equation}

\begin{theorem}~\cite{picard_theorem}
Suppose that functions $f_{g,i}(t,p_{g,i})$ $\frac{\partial f_{g,i}}{\partial p_{g,i}}(t,p_{g,i})$ are continuous in some open rectangle $\left\{ (t,p_{g,i}): 0 \leq  t \leq  \tau, 0 < p_{g,i} \leq 1\right\}$ that contains point $(t_0,p^0_{g,i})$. Then, the problem in (\ref{replicator_evolu_2}) has a unique solution in the interval of $I=\left[t_0-\varsigma, t_0+\varsigma\right]$, where $\varsigma> 0$. Moreover, the Picard iteration defined by 
\begin{equation}
p^{<z+1>}_{g,i}(t)  =p^0_{g,i} + \int_{t_0}^t f_{g,i}(t,p^{<z>}_{g,i}(t) ) \;\mathrm{d}t
\label{Picard_iteration}
\end{equation}
produces a sequence of functions ${ p^{<z>}_{g,i}(t)}$ that converges to the solution uniformly on $I$.
\label{theorem_picard}
\end{theorem}

\begin{proof}
Theorem \ref{theorem_picard} means that problem in (\ref{replicator_evolu_2}) converges to a unique solution, i.e., the equilibrium of the game defined in~(\ref{replicator_evolu_2}), given that function $f_{g,i}(t,p_{g,i})$ and its derivative are continuous with respect to time $t$. Therefore, we first show that $f_{g,i}(t,p_{g,i})$ and $\frac{\partial f_{g,i}}{\partial p_{g,i}}(t,p_{g,i})$ are continuous functions in the rectange $\left\{ (t,p_{g,i}): 0 \leq  t \leq  \tau, 0\leq p_{g,i} \leq  1\right \}$. Indeed, it is clear that function $p_{g,i}(t)=\frac{N_{g,i}(t)}{N}$ is continuous at every $t_0 \in [ 0, \tau]$. Moreover, due to the static flat-fading channel model, variables $\mathbf{h}_{m,i}(t)$, $\mathbf{h}_{g,i}^{\rm{IU}}(t)$, and $\mathbf{G}_{m,l,k}(t)$ are constant and thereby continuous at every $t_0 \in [ 0, \tau]$. Correspondingly, $\eta_{g,i}(t)$ is continuous at every $t_0 \in [ 0, \tau]$, and functions $\overline{R}_{g,i}(t), u_{g,i}(t)$, and $\overline{u}_i(t)$ are also continuous at every $t_0 \in [ 0, \tau]$ if $p_{g,i}(t_0)\neq 0$. Since $f_{g,i}(t,p_{g,i})=\mu p_{g,i}(t)[u_{g,i}(t)-\overline{u}_i(t)]$ and $\frac{\partial f_{g,i}}{\partial p_{g,i}}(t,p_{g,i})=\mu\left[u_{g,i}(t)-\overline{u}_i(t)\right]$, then $f_{g,i}(t,p_{g,i})$ and $\frac{\partial f_{g,i}}{\partial p_{g,i}}(t,p_{g,i})$ are continuous functions in the open rectangle $\left\{ (t,p_{g,i}): 0 \leq  t \leq  \tau, 0 < p_{g,i} \leq 1\right\}$. 

Given that the continuity of $f_{g,i}(t,p_{g,i})$ and its derivative, there are many ways to prove Theorem~\ref{theorem_picard}. One of them is leveraging the Banach Fixed Point Theorem (BFPT)~\cite{ciesielski2007stefan} to approximate a solution, i.e., a fixed point, to (\ref{replicator_evolu_2}) by constructing a sequence of functions that converges to a unique solution. The proof of Theorem~\ref{theorem_picard} using the BFPT is well explained and presented in~\cite{picard_theorem}. The  unique solution refers to the game equilibrium at which 1) all the users achieve the same utility and 2) no user has an incentive to change its network service selection.
\end{proof}

The overall process of the network service selection is summarized as follows. Initially, each user randomly selects an SP and a service of the SP. Given the user selection, the SP determines the optimal phase-shift and beamforming for its associated users according to Algorithm~\ref{phase_shift}. The user computes its utility according to (\ref{utility_user}) and transmits the utility information to the SP. The user compares its utility and the expected utility determined by (\ref{utility_user_average}) and can change its network service selection to achieve a higher utility value. After all the users achieve the same utility by choosing any strategies, then no user has an incentive to change its network service selection and the game converges to the evolutionary equilibrium. 

\begin{algorithm}
\footnotesize
        \caption{\footnotesize Optimizing phase-shift and beamforming for user $i$ when selecting SP $m$, subset $k$ of IRS $l$, and power level $J_{m,j}$~\cite{yu2019miso} (in this algorithm, function $\rm{unt}(\mathbf{a})$ is defined as $\rm{unt}(\mathbf{a})=[a_1/|a_1|,\ldots, a_n/|a_n|]$)}\label{phase_shift}
        \begin{algorithmic}[1]
	\State Output: $\mathbf{w}_{m,j}$ and $\mathbf{\Theta}_{m,l,k}$ for each user;
            \State Initialize: $t=0$, $\epsilon_1=0$, $\mathbf{\Theta}_{m,l,k}^0$;
            \State Calculate $\mathbf{R}_{11}=\text{diag}(({\mathbf{h}^{\rm{IU}}_{m,l,k,i})^{\rm{H}})\mathbf{G}_{m,l,k}\mathbf{G}_{m,l,k}^{\rm{H}}\text{diag}(\mathbf{h}^{\rm{IU}}_{m,l,k,i}}) $;
	 \State Construct $\mathbf{R}$ = \[
 \begin{bmatrix}
\mathbf{R}_{11}&  \text{diag}(({\mathbf{h}^{\rm{IU}}_{m,l,k,i}})^{\rm{H}})\mathbf{G}_{m,l,k}\mathbf{h}_{m,l,k,i} \\
 \mathbf{h}_{m,i}^{\rm{H}}\mathbf{G}_{m,l,k}^{\rm{H}}\text{diag}(\mathbf{h}^{\rm{IU}}_{m,l,k,i}) & 0 
 \end{bmatrix}
\]
            \Repeat 
            \State $\mathbf{v}^{(t)}=[\rm{diag}(\mathbf{\Theta}_{m,l,k}), t]^{\top}$;
		\State Calculate $\mathbf{v}^{(t+1)}$  according by $\mathbf{v}^{(t+1)}$=unt$(\mathbf{Rv}^{(t)})$;
		\State $t \gets t+1$;
		\Until{$\parallel{\mathbf{Rv}^{(t+1)}\parallel_{1}-\parallel\mathbf{Rv}^{(t)}\parallel_{1}} \leq \epsilon$};
           \State Take first $kE_{m,l}$ elements of ${(\mathbf{v}^{t+1})}^{*}$ as the main diagonal of $\mathbf{\Theta}_{m,l,k}$;
	\State Compute $\mathbf{w}_{m,j,i}=\sqrt{J_{m,j}}\frac{\mathbf{G}_{m,l,k}^{\rm{H}}\rm{diag}(\mathbf{h}^{\rm{IU}}_{m,l,k,i})\mathbf{\Theta}^{\rm{H}}_{m,l,k}+\mathbf{h}_{m,l,k,i}}{\parallel\mathbf{G}_{m,l,k}^{\rm{H}}\rm{diag}(\mathbf{h}^{\rm{IU}}_{m,l,k,i})\mathbf{\Theta}^{\rm{H}}_{m,l,k}+\mathbf{h}_{m,i}\parallel}$.
       \end{algorithmic}
    \end{algorithm}

The computational complexity of the algorithm is mainly caused from 1) the phase-shift and beamforming optimization implemented at the SPs (BSs) side and 2) the utility computation implemented at the user side. When the users select a network service of SP $m$, the SP optimizes the phase-shift matrix and beamforming for each user using Algorithm~\ref{phase_shift} that requires $5K_m^2L_m + K_m^2(5L_m+2) + K_m(L_m+2)+ 2L_m+1$ multiplications and additions. When the size increases to infinite, the complexity of the algorithm for each user is $\mathcal{O}(n^3)$. Each user calculates its own utility based on the prices and network services that it selects. Thus, the computational complexity of the user does not increase with the total numbers of users and the SPs. Thus, the complexity of the algorithm implemented at each user is $\mathcal{O}(1)$. This implies that the game approach is computationally efficient and highly scalable. 

\subsection{Delay in Replicator Dynamics}
In the game model discussed in the previous section, to make the decision on SP and service selections, the users need information about the average utility, i.e., $\overline{u}_{i}$, and the proportion of users choosing different strategies, i.e., $p_{m,l,k,j,i}$, from the BSs. However, the up-to-date information may not be available at the users due to the communication latency. Thus, at time instance $t$, the users may need to use the information at time $t-\delta$, i.e., delay for $\delta$ time units, to make the SP and service selections. Thus, the delayed replicator dynamic process is expressed as
\begin{align}
\label{replicator_evolu_delay}
\notag
\dot{p}_{m, l,k,j,i}(t) =&\mu p_{m, l,k,j,i}(t-\delta)\left[u_{m, l,k,j,i}(t-\delta)-\overline{u}_i(t-\delta)\right], \\
& m \in \mathcal{M}, l \in \mathcal{L}_m, k \in \mathcal{Q}_{m,l}, j \in \mathcal{P}_m, i \in \mathcal{N}_g ,\forall t.
\end{align}

Note that as delay $\delta$ is large, the decisions of the users based on the outdated information tend to be inaccurate. In this case, the SP and service selections may not converge. How to determine $\delta^*$ such that the selections converge is challenging. As an example, consider a simple scenario with $M=2$, and SP $m$ offers one service including subset $\boldsymbol{\Theta}_{m}$ and power level $J_m$:
\begin{theorem}
The evolutionary game can converge to a stable equilibrium if the value of $\delta$ is satisfied the following condition:
\begin{equation}
\delta^* < \frac{\pi}{2\mu \sum_{m\in \mathcal{M}} \frac{v_mB_m\log_2(1 + \eta_m) -\gamma_m^I||\boldsymbol{\Theta}_{m}||_0-\gamma_m^PJ_m }{N} }.
\label{stability}
\end{equation}
\label{theorem_stability}
\end{theorem}
\begin{proof}
The delayed replicator dynamics in (\ref{replicator_evolu_delay}) can be rewritten as 
\begin{equation}
\mathbf{\dot{p}}(t) =\mathbf{A} \mathbf{p}(t-\delta)+\mathbf{c}, 
\end{equation}
where $\mathbf{\dot{p}}(t) ={[{\dot{p}_{1}}(t),\ldots, {\dot{p}_{M}}(t)]}^{\top}$, $\mathbf{p}(t) ={[{p}_{1}(t-\delta), \ldots, {p}_{M}(t-\delta)]}^{\top}$, 
$\mathbf{c} ={[\frac{\mu a_{1}}{N}, \ldots, \frac{\mu a_{M}}{N}]}^{\top}$ with
$a_{m} = \frac{v_mB_m\log_2(1 + \eta_m) -\gamma_m^I||\boldsymbol{\Theta}_{m}||_0-\gamma_m^PJ_m }{N}$, and $\mathbf{A}=-\kappa\mathbf{I}$.  Here, $I$ is the identity matrix of size $M$, and $\kappa$ is defined as
\begin{align}
\notag
\kappa=\mu \sum_{m\in \mathcal{M}} \frac{v_mB_m\log_2(1 + \eta_m) -\gamma_m^I||\boldsymbol{\Theta}_{m}||_0-\gamma_m^PJ_m }{N}.
\end{align}

Otherwise, the evolutionary game with the delayed replicator dynamics can converge to a stable equilibrium if the real parts of all the roots are negative ~\cite{gopalsamy2013stability}. This is equivalently the condition $\kappa\delta<\pi$, and thus we have
\begin{align}
\notag
\delta^* < \frac{\pi}{2\kappa} = \frac{\pi}{2\mu \sum_{m\in \mathcal{M}} \frac{v_mB_m\log_2(1 + \eta_m) -\gamma_m^I||\boldsymbol{\Theta}_{m}||_0-\gamma_m^PJ_m }{N} }.
\end{align}
\end{proof}

 Theorem~\ref{theorem_stability} means that the evolutionary game is guaranteed to converge to the equilibrium as the users use information at $t< \delta^*$ for their decisions.

\section{Fractional Evolutionary Game Formulation}
\label{frac_game}
In this section, we discuss the use of the fractional evolutionary game to model the SP and network service selection of the users with memory effect in the IRS-enabled terahertz system. In particular, we first present the concept of memory-aware economic process. Then, we present how to cast the evolutionary game that describes the SP and network service selection into a fractional evolutionary game by using the memory-aware economic processes. Finally, we analyze the equilibrium of the game. 
 
 \subsection{Memory-aware Economic Process}
With the classical evolutionary game as presented in Section~\ref{classical_evol_game}, each user, say user $i$, decides on the SP and network service selection according to its instantaneously achievable utility functions, i.e., functions $u_{m, l,k,j,i}(t)$ and $\overline{u}_{i}(t)$ at time instant $t$. In reality, the users take into account their memory, i.e., of service experience, on their strategy decisions. Specifically, the selection decision of the users at time $t$ is based not only on the information about the state of the process at time $t$, but also on the information about the process states at previous time instants $\tau \in [0,t]$.~This is considered to be a memory-aware economic process~\cite{tarasova2018concept},~\cite{tarasova2016generalization}. To describe the memory-aware economic process, we consider a typical economic model with two variables, namely \textit{exogenous variable} and \textit{endogenous variable}. The exogenous variable and endogenous variable are the input and output of the economic model, respectively. This means that the endogenous variable depends on the exogenous variable, and they are similar to the independent and dependent variables, respectively. Denote $X(t)$ as the exogenous variable and $Y(t)$ as the endogenous variable variable, in which the exogenous variable changes depends on the changes of the endogenous variable. Then, the economic process is typically expressed by $Y(t)=F^t_0(X(\tau)) + Y_0$, where $\tau \in [0,t]$, $Y_0$ is the initial state of the output of the process, and $F^t_0$ is an operator. To enable the memory awareness of the economic process, the operator is defined as $F^t_0(X(\tau)):= \int_0^tM_{\beta}(t-\tau)X(\tau)\rm{d}t$, where $M_{\beta}(t-\tau)$ is the weight function that represents how the input $X(\tau)$ at time $\tau$ impacts on the output $Y(t)$ at time $t$. In general, function $M_{\beta}(t-\tau)$ changes with respect to $\tau$ so as to capture the dynamic characteristic of the memory. Furthermore, by taking the time derivative of $Y(t)$, we have $\frac{\rm{d}}{\rm{d}t}Y(t)= M_{\beta}(t)X(0)+  \int_0^tM_{\beta}(t-\tau) [\frac{\rm{d}}{\rm{d}t}Y(t) ]  X(\tau)\rm{d}\tau$ that depends on both $X(t)$ and $X(\tau)$ with $\tau \in [0,t)$. The formulation of $M_{\beta}(t-\tau)$ is $M_{\beta}(t-\tau)=\frac{1}{\Gamma(\beta)(t-\tau)^{1-\beta}}$, where $\Gamma(\cdot)$ is the gamma function that is defined by $\Gamma(\beta)= \int_0^{+\infty}x^{\beta-1}e^{-x}\rm{d}x$.

Since $Y(t)$ depends on both $X(t)$ and $X(\tau)$, the economic process is namely \textit{memory-aware economic process} that can be expressed in the fractional equation by taking the derivation of $Y (t)$ at the order of $\beta$ through the left-sided Caputo fractional derivative as follows:
\begin{equation}
{}^C_0D^{\beta}_{t} Y(t)=X(t),
\end{equation}
where $Y (0) = Y_0$ is the initial state, and ${}^C_0D^{\beta}_{t} Y(t)$ is the left-sided Caputo fractional derivative~\cite{tarasova2017logistic} of $Y (t)$ at the order of $\beta$ that is given by:
\begin{equation}
{}^C_0D^{\beta}_{t} Y(t)=\frac{1}{\Gamma(\ceil{\beta}-\beta)}\int^t_0\frac{Y^{(\ceil{\beta})}(\tau)}{(t-\tau)^{\beta+1-\ceil{\beta}}}\rm{d}\tau,
\label{fractional_derivative}
\end{equation}
where $\ceil{\cdot}$ is the ceiling function.

The memory-aware economic process given in (\ref{fractional_derivative}) has two key properties. First, the past experiences of the user at different time instances have different impacts on its decision-making so as to capture dynamically the memory of the user. Second, the user is affected by the experience within the memory rather than that at the current time, and consequently the memory-aware users can make network selection decisions differently from the memory-unaware users. Given the properties, we incorporate the memory awareness of the economic process when modeling the SP and network service selection of the users. The memory-aware economic process can be modeled as the fractional evolutionary game that is presented in the next section.

 \subsection{Fractional Game Formulation}
 For convenience, we rewrite the replicator dynamic process of the users in the classical evolutionary game as expressed in (\ref{replicator_evolu}) as follows:
\begin{align}
\notag
\dot{p}_{m,l, k,j,i}(t) =&\mu p_{m, l,k,j,i}(t)\left[u_{m, l,k,j,i}(t)-\overline{u}_{i}(t)\right], \\
& m \in \mathcal{M}, l \in \mathcal{I}_m, k \in \mathcal{Q}_{m,l}, j \in \mathcal{P}_m, i \in \mathcal{N}_g, \forall t.
\label{replicator_evolu_re}
\end{align}

Then, given the utility functions and the average utility of the users defined in~(\ref{utility_user}) and~(\ref{utility_user_average}), respectively, and by incorporating the memory characteristic of the users, we can formulate the fractional evolutionary game as follows:
\begin{align}
\label{replicator_evolu_fractional}
\notag
{}^C_0D^{\beta}_{t} p_{m, l,k,j,i}(t) =&\mu p_{m, l,k,j,i}(t)[u_{m, l,k,j,i}(t)-\overline{u}_i(t)], \\
& m \in \mathcal{M}, l \in \mathcal{L}_m, k \in \mathcal{Q}_{m,l}, j \in \mathcal{P}_m, i \in \mathcal{N}_g, \forall t,
\end{align}
where $\beta \in (0,2), \beta \neq 1$ is the order of the Caputo fractional derivative, and it is called memory effect coefficient.~The physical meaning of the left-sided Caputo is further explained and discussed in Section~\ref{perform_eval}. The equilibrium analysis of the fractional evolutionary game is presented in the next section.


 \subsection{Equilibrium Analysis}
 In this section, we theoretically discuss the existence and the uniqueness of the equilibrium, and the unique and stable equilibrium is admitted as the solution of the fractional evolutionary game defined in~(\ref{replicator_evolu_fractional}). The specific steps are as follows. First, we transfer the fractional game defined in~(\ref{replicator_evolu_fractional}) into an equivalent problem, i.e,~(\ref{Picard_iteration_fractional_game_equilibrium}), and the equivalence between which is verified in Theorem~\ref{th:equilivent_problem_FEG}. Then, to prove the existence and uniqueness of equilibrium of the game defined in~(\ref{replicator_evolu_fractional}), we provide the proof of the uniqueness of the solution to the equivalent problem defined in ~(\ref{Picard_iteration_fractional_game_equilibrium}).

For the ease of presentation, we let $\mathbf{P}(t)=[p_{m,l,k,j,i}\left(t\right)]_{ m \in \mathcal{M}, l \in \mathcal{L}_m, k \in \mathcal{Q}_{m,l}, j \in \mathcal{P}_m, i \in \mathcal{N}_g}$ and $\mathbf{E}(\mathbf{P}(t))= \big{[}\mu p_{m, l,k,j,i}(t)[u_{m, l,k,j,i}(t)-\overline{u}_i(t)]\big{]}_{ m \in \mathcal{M}, l \in \mathcal{L}_m, k \in \mathcal{Q}_{m,l}, j \in \mathcal{P}_m, i \in \mathcal{N}_g}$, and reorganize the fractional evolutionary game defined in~(\ref{replicator_evolu_fractional}) as follows:
 \begin{equation}
\label{replicator_evolu_delay_power_simplied}
{}^C_0D^{\beta}_{t} \mathbf{P}(t)=\mathbf{E}(\mathbf{P}(t)),
\end{equation}
with the initial strategy $\mathbf{P}(0)=\mathbf{P}^0=[p^0_{m,l,k,j,i}]_{ m \in \mathcal{M}, l \in \mathcal{L}_m, k \in \mathcal{Q}_{m,l}, j \in \mathcal{P}_m, i \in \mathcal{N}_g}$ and the time horizon $\mathcal{T}=[0,T]$.

\begin{theorem}\label{th:equilivent_problem_FEG}
If all the elements of vector $\mathbf{E}$ in~(\ref{replicator_evolu_delay_power_simplied}), i.e., $e_n$ (element $n$ of vector $\mathbf{E}$) for all $n \in \left\{\left. \left(m,l,k,j,i\right) \right| m \in \mathcal{M}, l \in \mathcal{L}_m, k \in \mathcal{Q}_{m,l}, j \in \mathcal{P}_m, i \in \mathcal{N}_g \right\}$, can satisfy the following two conditions:
\begin{itemize}
\item $e_n \in {\cal{C}}^2$ with ${\cal{C}}^2$ being the set of the twice differentiable functions;

\item $\frac{\partial} {\partial p_{m,l,k,j,i} }e_n$ exists and is bounded for all $m \in \mathcal{M}, l \in \mathcal{L}_m, k \in \mathcal{Q}_{m,l}, j \in \mathcal{P}_m,i \in \mathcal{N}_g$.
\end{itemize}
Then,~(\ref{replicator_evolu_delay_power_simplied}) can be equivalently transformed into the following problem
\begin{equation}
\mathbf{P}(t)= \mathbf{P}^0 +   {}_0\rm{I}^\beta_t\mathbf{E}(\mathbf{P}(t)), \forall t\in \mathcal{T}.
\label{Picard_iteration_fractional_game_equilibrium}
\end{equation}
The second condition means that for all $n \in \left\{\left. \left(m,l,k,j,i\right) \right| m \in \mathcal{M}, l \in \mathcal{L}_m, k \in \mathcal{Q}_{m,l}, j \in \mathcal{P}_m, i \in \mathcal{N}_g \right\}$, there exists $L \in \mathbb{R}^+$ such that $|e_n({\hat{\mathbf{P}}}(t)) - e_n({\tilde{\mathbf{P}}}(t))| < L||({\hat{\mathbf{P}}}(t)) - {\tilde{\mathbf{P}}}(t)||_{\mathcal{L}_1}$, which implies the satisfaction of the Lipschitz condition. 

\end{theorem}

\begin{proof}
According to~(\ref{Picard_iteration_fractional_game_equilibrium}), the $\ceil{\beta}$-th derivative of $\mathbf{P}(t)$ with respect to $t$ is as follows:
\begin{equation}
\frac{\rm{d}^{\ceil{\beta}}}{\rm{d}t^{\ceil{\beta}}}  \mathbf{P}(t)= \frac{\rm{d}^{\ceil{\beta}}}{\rm{d}t^{\ceil{\beta}}} \big{[} {}_0\rm{I}^\beta_t\mathbf{E}(\mathbf{P}(t)) \big{]} = {}^{RL}_0D^{\beta - \ceil{\beta}}_{t} {\mathbf{E}}\left(\mathbf{P}(t)\right),
\label{Picard_iteration_fractional_game_equilibrium_proof_1}
\end{equation}
with ${}^{RL}_0D^{\beta}_{t} {\mathbf{E}}\left(\mathbf{P}(t)\right)$ being defined as the left-sided Riemann-Liouville fractional derivative with respect to $t$, and the following derivation is satisfied
\begin{equation}\label{eq:RL_upper_bound}
\begin{aligned}
{}^{RL}_0D^{\ceil{\beta} - \beta}_{t} {\mathbf{E}}\left(\mathbf{P}(t)\right) &= \frac{\rm{d}}{\rm{d}t} \Big{[} \frac{1}{\Gamma(1-\ceil{\beta}+\beta)} \int^t_0 \frac{\mathbf{E}(\mathbf{P}(\tau))} { (t-\tau)^{\ceil{\beta}-\beta} } \rm{d}\tau \Big{]}\\
&=  \frac{1}{\Gamma(1-\ceil{\beta}+\beta)} \frac{\rm{d}}{\rm{d}t} \int^t_0\theta^{\beta -\ceil{\beta}} \mathbf{E}(\mathbf{P}(t-\theta)) \rm{d}\theta\\
&=  \frac{1}{\Gamma(1-\ceil{\beta}+\beta)}\Big{[}t^{\beta-\ceil{\beta}} \mathbf{E}(\mathbf{P}^0) + \int^t_0\theta^{\beta-\ceil{\beta}} \frac{\rm{d}}{\rm{d}t} \mathbf{E}(\mathbf{P}(t-\theta)) \rm{d}\theta \Big{]}\\
&=  \frac{1}{\Gamma(1-\ceil{\beta}+\beta)}\Big{[}t^{\beta-\ceil{\beta}} \mathbf{E}(\mathbf{P}^0) + \int^t_0(t-\tau)^{\beta-\ceil{\beta}} \frac{\rm{d}}{\rm{d}\tau} \mathbf{E}(\mathbf{P}(\tau)) \rm{d} \tau \Big{]} \\
&= \frac{  t^{\beta-\ceil{\beta}}  }{ \Gamma(1-\ceil{\beta}+\beta) }\mathbf{E}(\mathbf{P}^0) + {}_0\rm{I}^{\beta}_t \Big{[}\frac{\rm{d}^{\ceil{\beta}} }{\rm{d}t^{\ceil{\beta}} }\mathbf{E}(\mathbf{P}(t))    \Big{]}.
\end{aligned}
\end{equation}

Let $\sigma \in (0,t)$, the $\mathcal{L}^1$ norm of $\frac{  t^{\beta-\ceil{\beta}}  }{ \Gamma(1-\ceil{\beta}+\beta) }\mathbf{E}(\mathbf{P}^0)$ exists an upper bound that is derived as follows:
\begin{equation}\label{eq:RL_residual_term}
\Vert \frac{  t^{\beta-\ceil{\beta}}  }{ \Gamma(1-\ceil{\beta}+\beta) }\mathbf{E}(\mathbf{P}^0) \Vert_{\mathcal{L}^1} \le \Vert \frac{  \sigma^{\beta-\ceil{\beta}}  }{ \Gamma(1-\ceil{\beta}+\beta) }\mathbf{E}(\mathbf{P}^0)   \Vert_{\mathcal{L}^1}.
\end{equation}

Using the condition in Theorem~\ref{th:equilivent_problem_FEG} that $\frac{\delta} {\delta p_{m,l,k,j,i} }e_n$ exists and is bounded for all $m \in \mathcal{M}, l \in \mathcal{L}_m, k \in \mathcal{Q}_{m,l}, j \in \mathcal{P}_m,i \in \mathcal{N}_g$, and~(\ref{Picard_iteration_fractional_game_equilibrium_proof_1}) and~(\ref{eq:RL_upper_bound}) as well as~(\ref{eq:RL_residual_term}), we have 
\begin{equation} \label{eq:integer_derivative_P}
 \begin{aligned}
\Big{\Vert} \frac{\rm{d}^{\ceil{\beta}}}{\rm{d}t^{\ceil{\beta}}}  \mathbf{P}(t) \Big{\Vert}_{\mathcal{T}} &< \Vert \frac{  \sigma^{\beta-\ceil{\beta}}  }{ \Gamma(1-\ceil{\beta}+\beta) }\mathbf{E}(\mathbf{P}^0)   \Vert_{\mathcal{L}^1} +  \Big{\Vert} {}_0\rm{I}_t^\beta \frac{\rm{d}^{\ceil{\beta}}} {\rm{d}t^{\ceil{\beta}}} \mathbf{E}(\mathbf{P}(t)) \Big{\Vert}_{\mathcal{T}}\\
&< \Vert \frac{  \sigma^{\beta-\ceil{\beta}}  }{ \Gamma(1-\ceil{\beta}+\beta) }\mathbf{E}(\mathbf{P}^0)   \Vert_{\mathcal{L}^1} +  \Big{\Vert} {}_0\rm{I}_t^\beta \frac{\rm{d}^{\ceil{\beta}}} {\rm{d}t^{\ceil{\beta}}} \mathbf{P}(t) \Big{\Vert}_{\mathcal{T}} A L, 
 \end{aligned}
\end{equation}
 where $\Vert z \Vert_{\mathcal{T}} = \int_{\mathcal{T}} \exp\left(-\mu t\right)\Vert z\Vert_{\mathcal{L}^1}\rm{d}t $ and $A$ is the cardinality of $\left\{\left. \left(m,l,k,j,i\right) \right| m \in \mathcal{M}, l \in \mathcal{L}_m, k \in \mathcal{Q}_{m,l},\right.$ $\left. j \in \mathcal{P}_m, i \in \mathcal{N}_g \right\}$. For the last term in~(\ref{eq:integer_derivative_P}), we have 
\begin{equation}\label{eq:fractional_integral_integer_derivative_P}
\begin{aligned}
  \Big{\Vert} {}_0\rm{I}_t^\beta \frac{\rm{d}^{\ceil{\beta}}} {\rm{d}t^{\ceil{\beta}}} \mathbf{P}(t) \Big{\Vert}_{\mathcal{T}}&= \int^T_0 \exp\left(-\mu t\right) \Big{\Vert} {}_0\rm{I}_t^\beta \frac{\rm{d}^{\ceil{\beta}}} {\rm{d}t^{\ceil{\beta}}} \mathbf{P}(t) \Big{\Vert}  \rm{d}t \leq \int_0^T \exp\left(-\mu t\right) \int^t_0 \frac{1}{\Gamma(\beta)} \frac{ \Vert  \frac{ \rm{d}^{\ceil{\beta}} }{\rm{d}s^{\ceil{\beta}}} \mathbf{P}(s) \Vert }  {(t-s)^{1-\beta}}\rm{d}s\rm{d}t\\
  &= \int_0^T \frac{\exp\left(-\mu s\right)} {\Gamma(\beta)} \Vert  \frac{ \rm{d}^{\ceil{\beta}} }{\rm{d}s^{\ceil{\beta}}} \mathbf{P}(s) \Vert \int^T_s \exp\left(-\mu(t-s)\right)(t-s)^{\beta-1}\rm{d}t\rm{d}s\\
  &= \int_0^T \frac{\exp\left(-\mu s\right)} {\Gamma(\beta)} \Vert  \frac{ \rm{d}^{\ceil{\beta}} }{\rm{d}s^{\ceil{\beta}}} \mathbf{P}(s) \Vert \int^{\mu(T-s)}_0 \exp\left(-\Psi\right) (\frac{\Psi}{\mu})^{\beta-1} \rm{d}(\frac{\Psi}{\mu})  \rm{d}s\\
& < \frac{1} {\mu^\beta \Gamma(\beta)}\int_0^T \exp\left(-\mu s\right) \Vert  \frac{ \rm{d}^{\ceil{\beta}} }{\rm{d}s^{\ceil{\beta}}} \mathbf{P}(s) \Vert \rm{d}s \int^{+\infty}_0 \exp\left(-\Psi\right) \Psi^{\beta-1} \rm{d}\Psi  \\
& = \frac{1} {\mu^\beta}\Vert  \frac{ \rm{d}^{\ceil{\beta}} }{\rm{d}s^{\ceil{\beta}}} \mathbf{P}(t) \Vert_{\cal{T}}.
\end{aligned}
\end{equation}
Then, we substitute~(\ref{eq:fractional_integral_integer_derivative_P}) into~(\ref{eq:integer_derivative_P}) as follows
\begin{equation}
 \begin{aligned}
&\Big{\Vert} \frac{\rm{d}^{\ceil{\beta}}}{\rm{d}t^{\ceil{\beta}}}  \mathbf{P}(t) \Big{\Vert}_{\mathcal{T}} < \Vert \frac{  \sigma^{\beta-\ceil{\beta}}  }{ \Gamma(1-\ceil{\beta}+\beta) }\mathbf{E}(\mathbf{P}^0)   \Vert_{\mathcal{L}^1} +  \frac{AL} {\mu^\beta}\Vert  \frac{ \rm{d}^{\ceil{\beta}} }{\rm{d}s^{\ceil{\beta}}} \mathbf{P}(t) \Vert_{\cal{T}}\\
\Leftrightarrow & \Big{\Vert} \frac{\rm{d}^{\ceil{\beta}}}{\rm{d}t^{\ceil{\beta}}}  \mathbf{P}(t) \Big{\Vert}_{\mathcal{T}} < \frac{1}{1 - \frac{AL} {\mu^\beta}}\Vert \frac{  \sigma^{\beta-\ceil{\beta}}  }{ \Gamma(1-\ceil{\beta}+\beta) }\mathbf{E}(\mathbf{P}^0)   \Vert_{\mathcal{L}^1},
 \end{aligned}
\end{equation}
which implies that if $\mu$ is sufficiently large such that $\frac{AL} {\mu^\beta} < 1$, $\Big{\Vert} \frac{\rm{d}^{\ceil{\beta}}}{\rm{d}t^{\ceil{\beta}}}  \mathbf{P}(t) \Big{\Vert}_{\mathcal{T}}$ has an upper bound. In this case, the fractional derivative of $\mathbf{P}(t)$ with the order of $\beta\in \left(0,1\right)\cup\left(1,2\right)$ exists and can be obtained in the following:
\begin{equation}
{}^C_0D^{\beta}_{t} \mathbf{P}(t) = {}_0\rm{I}^{{\ceil{\beta}} - \beta}_t \frac{\rm{d}^{\ceil{\beta}}}{\rm{d}t^{\ceil{\beta}}}  \mathbf{P}(t) = {}_0\rm{I}^{{\ceil{\beta}} - \beta}_t \left\{\frac{  t^{\beta-\ceil{\beta}}  }{ \Gamma(1-\ceil{\beta}+\beta) }\mathbf{E}(\mathbf{P}^0) + {}_0\rm{I}^{\beta}_t \Big{[}\frac{\rm{d}^{\ceil{\beta}} }{\rm{d}t^{\ceil{\beta}} }\mathbf{E}(\mathbf{P}(t))    \Big{]} \right\} = \mathbf{E}(\mathbf{P}(t)) ,
\end{equation}
\end{proof}
and hence the equivalence between~(\ref{replicator_evolu_delay_power_simplied}) and~(\ref{Picard_iteration_fractional_game_equilibrium}) has been verified, which completes the proof. 

\begin{theorem} \label{th:existence_unique_equlivent_problem_FEG}
Given the conditions in Theorem~\ref{th:equilivent_problem_FEG}, the uniqueness of the solution to the problem defined in~(\ref{Picard_iteration_fractional_game_equilibrium}) can be guaranteed.
\label{admit_solution}
\end{theorem}

\begin{proof}
First, by defining an operator $\Lambda: \mathcal{D}_{\mathbf{P}} \mapsto  \mathcal{D}_{\mathbf{P}}$, where $\mathcal{D}_{\mathbf{P}}$ is the feasible domain of ${\mathbf{P}}$, there exists an inequality expression as follows:
\begin{equation}
 \Vert \Lambda {\hat{\mathbf{P}}}(t) - \Lambda {\tilde{\mathbf{P}}}(t)\Vert_{\mathcal{T}} < \frac{AL}{\mu^\beta}\Vert {\hat{\mathbf{P}}}(t) - {\tilde{\mathbf{P}}}(t)\Vert_{\mathcal{T}},
 \label{inequality_operator}
\end{equation}
and the specific derivation of which has been shown as follows
\begin{equation}\label{eq:bounded_operator}
\begin{aligned}
& \Vert\Lambda {\hat{\mathbf{P}}}(t) - \Lambda {\tilde{\mathbf{P}}}(t)\Vert_{\mathcal{T}} = \int^T_0 \exp\left(-\mu t\right)\Vert {}_0\mathbf{I}^\beta_t \mathbf{E}({\hat{\mathbf{P}}}(t))-{}_0\mathbf{I}^\beta_t \mathbf{E}({\tilde{\mathbf{P}}}(t)) \Vert\rm{dt} \\
&< AL \Big{[}  \int^T_0 \exp\left(-\mu t\right)\Vert {}_0\mathbf{I}^\beta_t {\hat{\mathbf{P}}}(t)-{}_0\mathbf{I}^\beta_t {\tilde{\mathbf{P}}}(t) \Vert\rm{dt}  \Big{]}\\
 &\leq AL \Big{[}  \int^T_0 \exp\left(-\mu t\right) \int^t_0 \frac{ \Vert {\hat{\mathbf{P}}}(s)- {\tilde{\mathbf{P}}}(s)\Vert} {\Gamma(\beta) (t-s)^{1-\beta} } {\rm{d}}s{\rm{d}}t \Big{]} \\
 &=  \frac{AL}{\Gamma(\beta)} \Big{[}  \int^T_0  \int^T_s \exp\left(-\mu t\right) \frac{ \Vert {\hat{\mathbf{P}}}(s)- {\tilde{\mathbf{P}}}(s)\Vert} { (t-s)^{1-\beta} } \rm{dtds} \Big{]}\\
 &=  \frac{AL}{\Gamma(\beta)} \Big{[}  \int^T_0  \exp\left(-\mu s\right)  \Vert {\hat{\mathbf{P}}}(s)- {\tilde{\mathbf{P}}}(s) \Vert  \int^{T}_s\frac{ \exp\left(-\mu (t- s)\right) } { (t-s)^{1-\beta} } {\rm{d}}s{\rm{d}}t \Big{]}\\
 &= \frac{AL}{\mu^\beta \Gamma(\beta)} \Big{[}  \int^T_0  \exp\left(-\mu s\right)  \Vert  {\hat{\mathbf{P}}}(s)- {\tilde{\mathbf{P}}}(s) \Vert  \int^{\mu(T-s)}_0 \exp\left(-\psi\right)\psi^{\beta-1} {\rm{d}}\psi {\rm{d}}s \Big{]}\\
 &< \frac{AL}{\mu^{\beta} \Gamma(\beta)}\Vert {\hat{\mathbf{P}}}(t)- {\tilde{\mathbf{P}}}(t) \Vert_\mathcal{T} \int^{+\infty}_0\exp\left(-\sigma\right)\sigma^{\beta -1}\rm{d}\sigma = \frac{AL}{\mu^\beta} \Vert {\hat{\mathbf{P}}}(t)- {\tilde{\mathbf{P}}}(t) \Vert_\mathcal{T}. 
\end{aligned}
\end{equation} 

Based on~(\ref{eq:bounded_operator}), we can conclude that $\Vert\Lambda {\hat{\mathbf{P}}}(t) - \Lambda {\tilde{\mathbf{P}}}(t)\Vert_{\mathcal{T}} < \Vert {\hat{\mathbf{P}}}(t)- {\tilde{\mathbf{P}}}(t) \Vert_\mathcal{T} $ if $\mu^\beta \ge AL$. In this case, the operator $\Lambda$ satisfies the fixed point theorem, which indicates the uniqueness of the solution to~(\ref{Picard_iteration_fractional_game_equilibrium}). By following this, there exists a unique solution to the fractional evolutionary game defined in~(\ref{replicator_evolu_delay_power_simplied}), which completes this proof. 

\end{proof}
   
\section{Performance Evaluation}
\label{perform_eval}
In this section, we present and discuss simulation results obtained by the proposed evolutionary game approaches. To evaluate the game approaches, we consider three cases, i.e., $\beta=1$ corresponding to the classical evolutionary game, and $\beta=1.1$ and $0.8$ corresponding to the fractional evolutionary games. For the comparison purpose, we consider a network that consists of two SPs, namely SP 1 and SP 2, and $100$ users. Each SP deploys a BS that is equipped with $4$ antennas. SP 1 deploys 2 IRSs, namely IRS $1_1$ and $1_2$, and SP 2 deploys 1 IRS, namely IRS $2$. SP 1 divides each IRS into two modules and offers $1$ power level, i.e., $J_{1,1}=30$ dBm. SP 2 does not divide its IRS and offers $2$ power levels, i.e., $J_{2,1}=25$ dBm and $J_{2,2}=35$ dBm. As such, SP 1 offers 4 services, and SP 2 offers 2 services that the users can select. Correspondingly, $100$ users are divided into $6$ groups. Note that the SPs can offer more services and our proposed game approaches are scalable since the complexity of the algorithm implemented at each user is $\mathcal{O}(1)$ as analyzed in Section~\ref{classical_evol_game_form}. The locations of the BSs, IRSs, and users are shown in Fig.~\ref{Location}. The simulation parameters are provided in Table~\ref{table:parameters}. In particular for the involved channels, the non-LoS elements is proved to be much weaker than the LoS element, i.e., lower than $20$ dB~\cite{han2014multi}, and thus similar to~\cite{gao2016fast},~we mainly consider the channel with only LoS element, i.e., $\mathcal{L}=\mathcal{L}_{\text{BS-I}}=\mathcal{L}_{\text{I-U}}=1$. 
\begin{figure}[h]
 \centering
\includegraphics[width=0.5\linewidth]{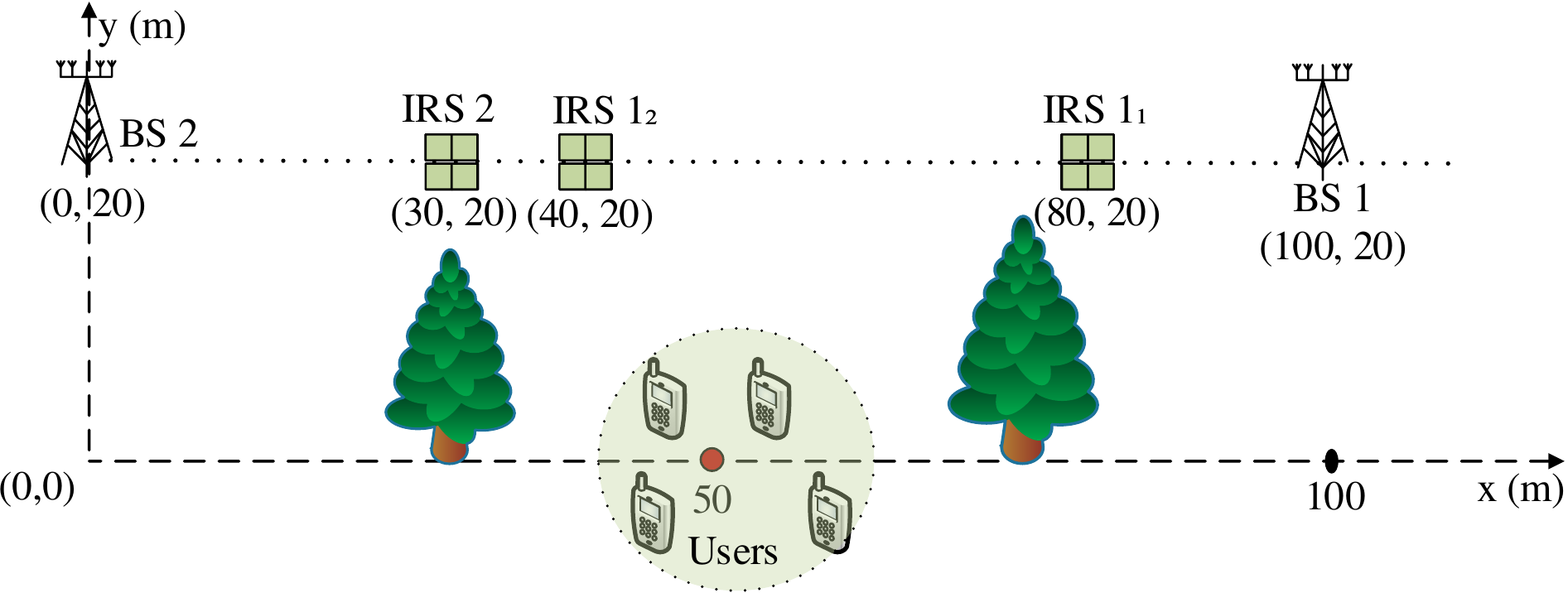}
 \caption{Coordinates of BSs, IRSs, and users.}
  \label{Location}
\end{figure}

\begin{table}[t]
\caption{\small Simulation parameters.}
\label{table:parameters_CRN}
\footnotesize 	
\centering
\begin{tabular}{lc|lc|lc|lc}
\hline\hline
{\em Parameters} & {\em Value} & {\em Parameters} & {\em Value}  & {\em Parameters} & {\em Value}&   {\em BSs, IRSs and user} & {\em Coordinate $(m)$ }\\ [1ex]
\hline
 $M$    & $2$           &    $B_1=B_2$ & $2$ MHz  &        $v_1$, $v_2$ &   $10^{-6}$ &    BS $1$   & [100 20]             \\ 
\hline
 $L_1=L_2$    & $4$          &       $f$& $0.3$ THz &     $\gamma^I_1$ &   $2.5\times10^{-4}$& BS 2    &  [0 20]           \\ 
\hline
$P_1$ & $1$              &  $c$& $3 \times 10^8$ m/s  &     $\gamma^I_2$ &   $10^{-3}$   &IRS $1_1$  &   [80 20]         \\ 
\hline
$P_2$ &  $ 2 $     &     $N$  & $100$  &   $\gamma^P_1$    & $0.05$ &IRS $1_2$  &   [40 20]         \\ 
\hline
$J_{1,1}$ & $30$ dBm    &    $Q_{1,1}=Q_{1,2}$ & $2 $ &    $\gamma^P_2$    & $0.03$&  IRS $2_1 $  &   [30 20]         \\ 
\hline 
$J_{2,1}$ & $25$ dBm       &  $Q_{2,1}$ & $1$&      $\gamma^P_2$    & $0.03$  &User group  & [50 0]      \\ 
\hline
$J_{2,2}$ & $35$ dBm       &    $E_{1,1}$, $E_{1,2}$, $E_{2,1}$ & $8$  &       $(\mu, \zeta)$ & $(e^{-2},0)$ &  &  \\ 
\hline
\end{tabular}
\label{table:parameters}
\end{table}

\begin{figure*}[t]
  \centering
  \subcaptionbox{}[.23\linewidth][c]{%
    \includegraphics[width=1\linewidth]{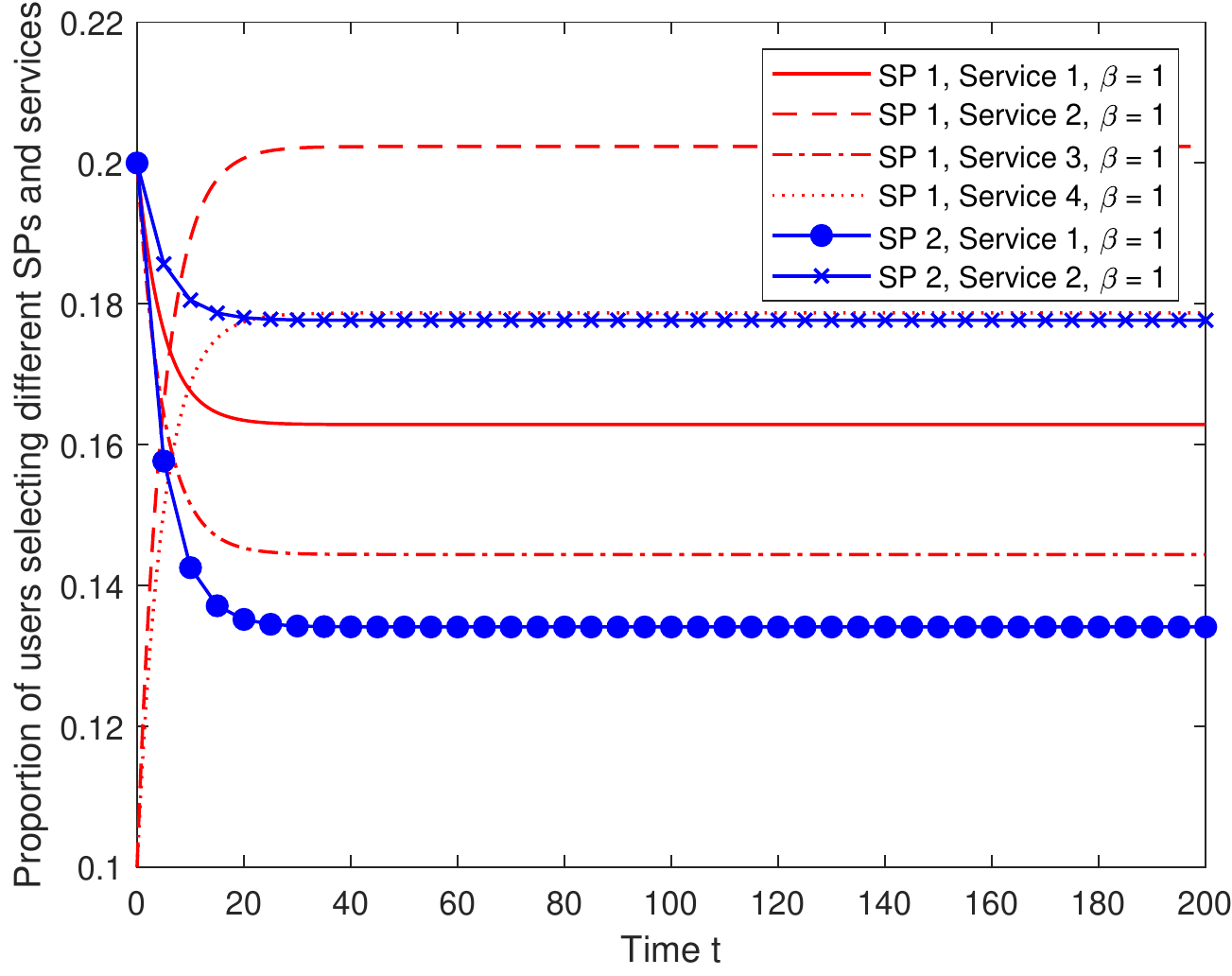}}\quad
  \subcaptionbox{}[.23\linewidth][c]{%
    \includegraphics[width=1\linewidth]{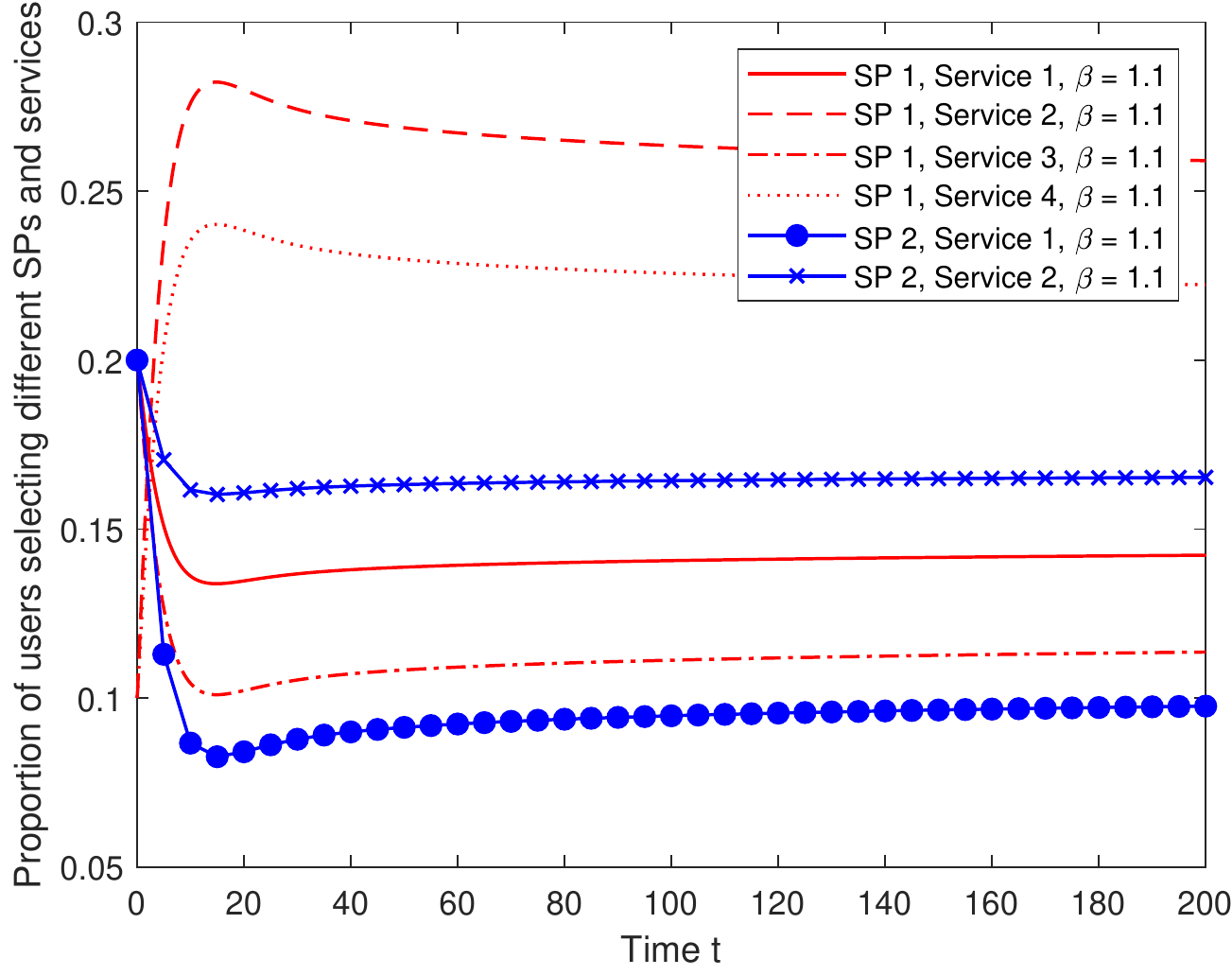}}\quad
  \subcaptionbox{}[.23\linewidth][c]{%
    \includegraphics[width=1\linewidth]{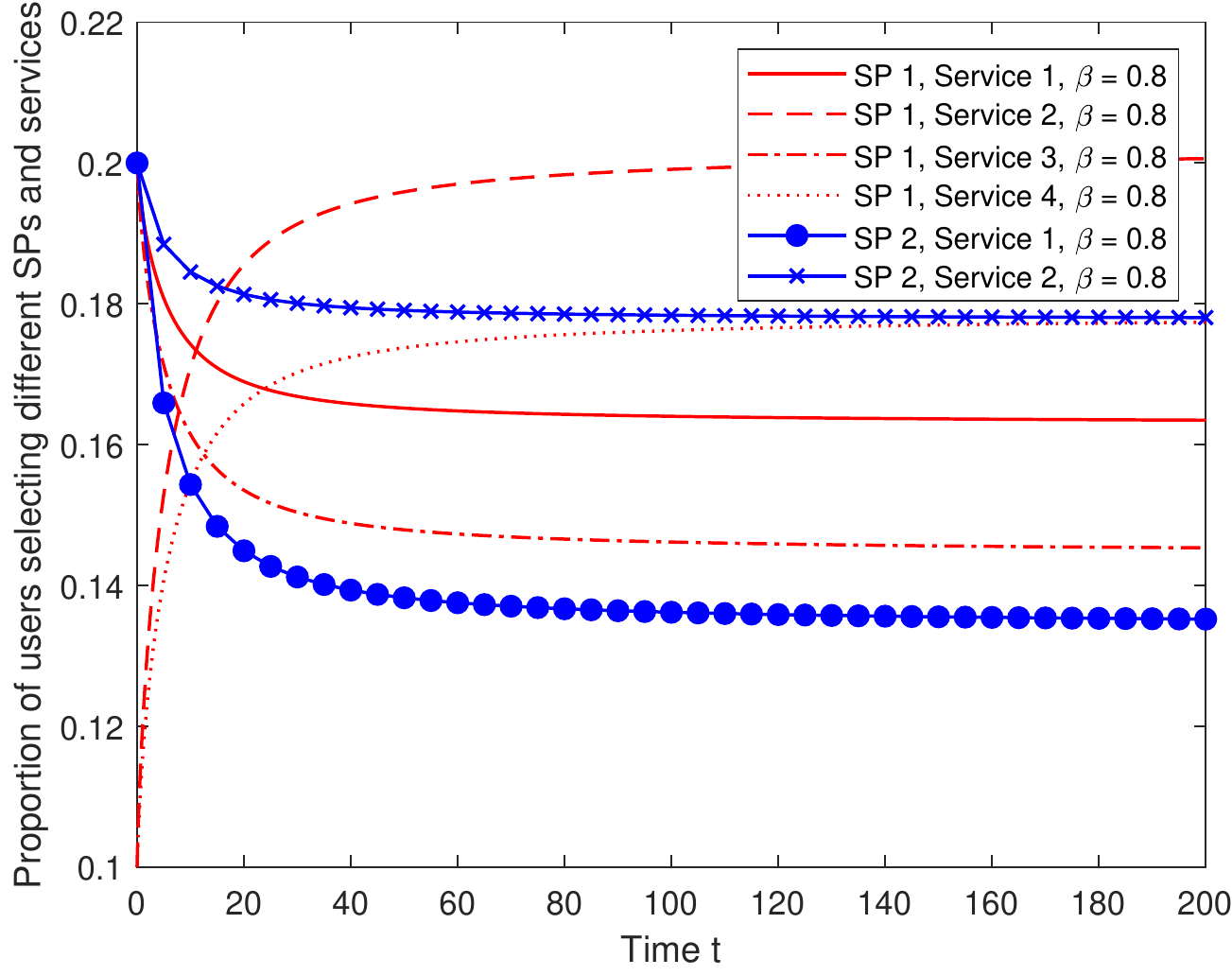}}
      \subcaptionbox{}[.23\linewidth][c]{%
    \includegraphics[width=1\linewidth]{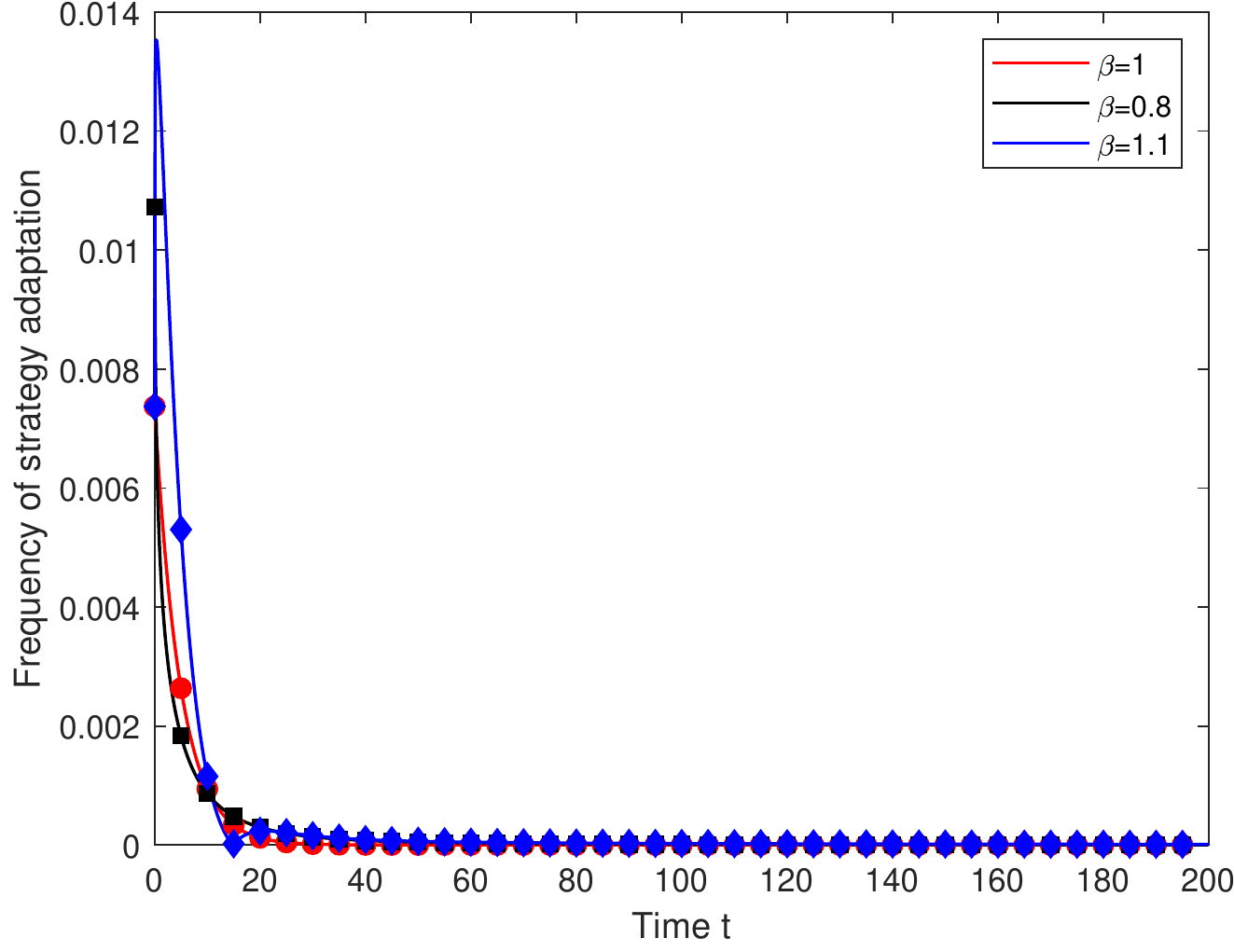}}
\caption{Proportion of users selecting different SPs and services.}
	\label{proportion1}
\end{figure*}


First, we discuss strategies that the users choose different SPs and services over evolutionary time. Figures~\ref{proportion1}(a), (b), and (c) illustrate the results for the evolutionary games with $\beta=1, 1.1$ and $0.8$, respectively. As seen, the strategies of the users choosing different SPs and services eventually converge to an equilibrium point over time. Moreover, during the initial phase, the users' strategies in the evolutionary games with $\beta=1$ and $0.8$ fluctuate in a range smaller than those in the evolutionary game with $\beta=1.1$. The results indicate that the adaptations of the users' strategies in the game with $\beta = 1.1$ are faster than those in the games with $\beta = 1$ and $0.8$. These results are further verified in Fig.~\ref{proportion1}(d) in which the strategy adaptation frequency of the users in the game with $\beta = 1.1$ is higher than those in the games with $\beta = 1$ and $0.8$. Note that as the users' strategies have a larger fluctuation, the convergence speed can be slower. Thus, as we can observe from Fig.~\ref{proportion1}(d), the game with $\beta = 1.1$ converges to the equilibrium more slowly than those in the games with $\beta=1$ and $0.8$. 

\begin{figure*}[t]
  \centering
  \subcaptionbox{}[.27\linewidth][c]{%
    \includegraphics[width=1\linewidth]{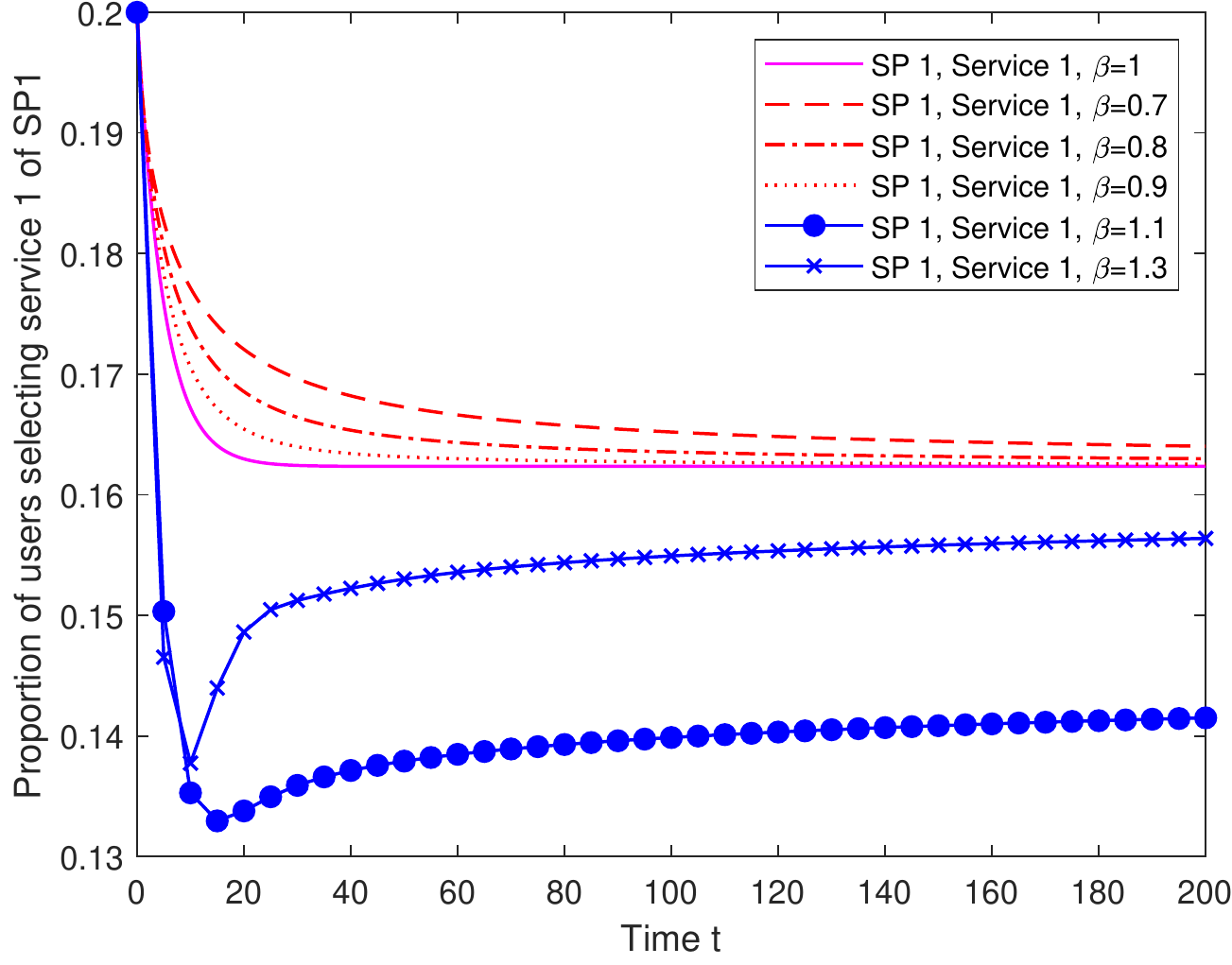}}\quad
  \subcaptionbox{}[.27\linewidth][c]{%
    \includegraphics[width=1\linewidth]{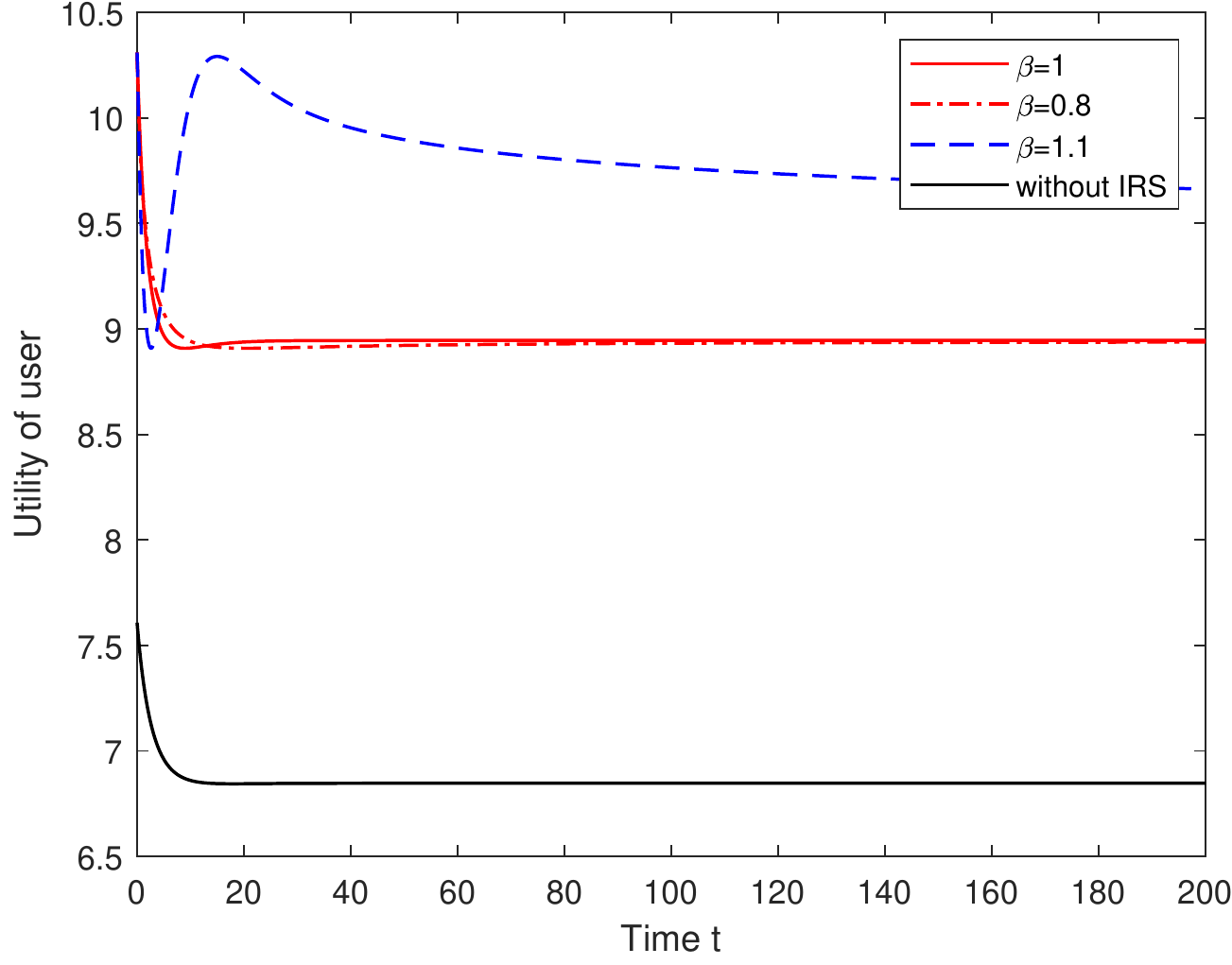}}

\caption{The impact of $\beta$ on the evolutionary games.}
	\label{impact_beta}
\end{figure*}




Now, we discuss how the memory effect coefficient, i.e., $\beta$, impacts on the evolutionary games. As shown in Fig.~\ref{impact_beta}(a), with $\beta < 1$ and $\beta > 1$, as $\beta$ increases, the convergence time of the user's strategy is shorter. This implies that as $\beta$ increases, the replicator dynamics converges faster. Moreover, as shown in Fig.~\ref{impact_beta}(b), as $\beta$ increases, the rate that the corresponding games converge to the vicinity of the equilibrium is faster. This means that the adaptation rate of the user's strategy increases with the increase of $\beta$. This results is also consistent with descriptions in Figs.~\ref{proportion1}(a), (b), (c), and (d). 

Next, it is important to show the utility that the users can achieve when different games are used. For this, we vary the values of $\beta$, i.e., the memory effect, and we evaluate the total utility the the users achieve. As shown in Fig.~\ref{impact_beta}(b), the utility for the users with $\beta=0.8 < 1$ is worse than that for the users with $\beta=1$. Meanwhile, with $ \beta = 1.1 > 1$, the users achieve higher utility values when the users have no memory effect, i.e., $\beta=1$. Since the users achieve higher utility values with $\beta = 1.1$, we can say that the memory effect with $\beta > 1$ is a positive effect. This further implies that to achieve a higher utility value, the users should incorporate both the past and instantaneously achievable experiences for their network selection. In addition, as seen from Fig.~\ref{impact_beta}(b), the total utility of users with IRSs in the games, i.e., $\beta = 0.8$, $1$, and $1.1$ is much higher than that of users without IRSs. This result demonstrates that deploying IRSs increases the throughput of the users.
\begin{figure*}[t]
  \centering
  \subcaptionbox{}[.3\linewidth][c]{%
    \includegraphics[width=1\linewidth]{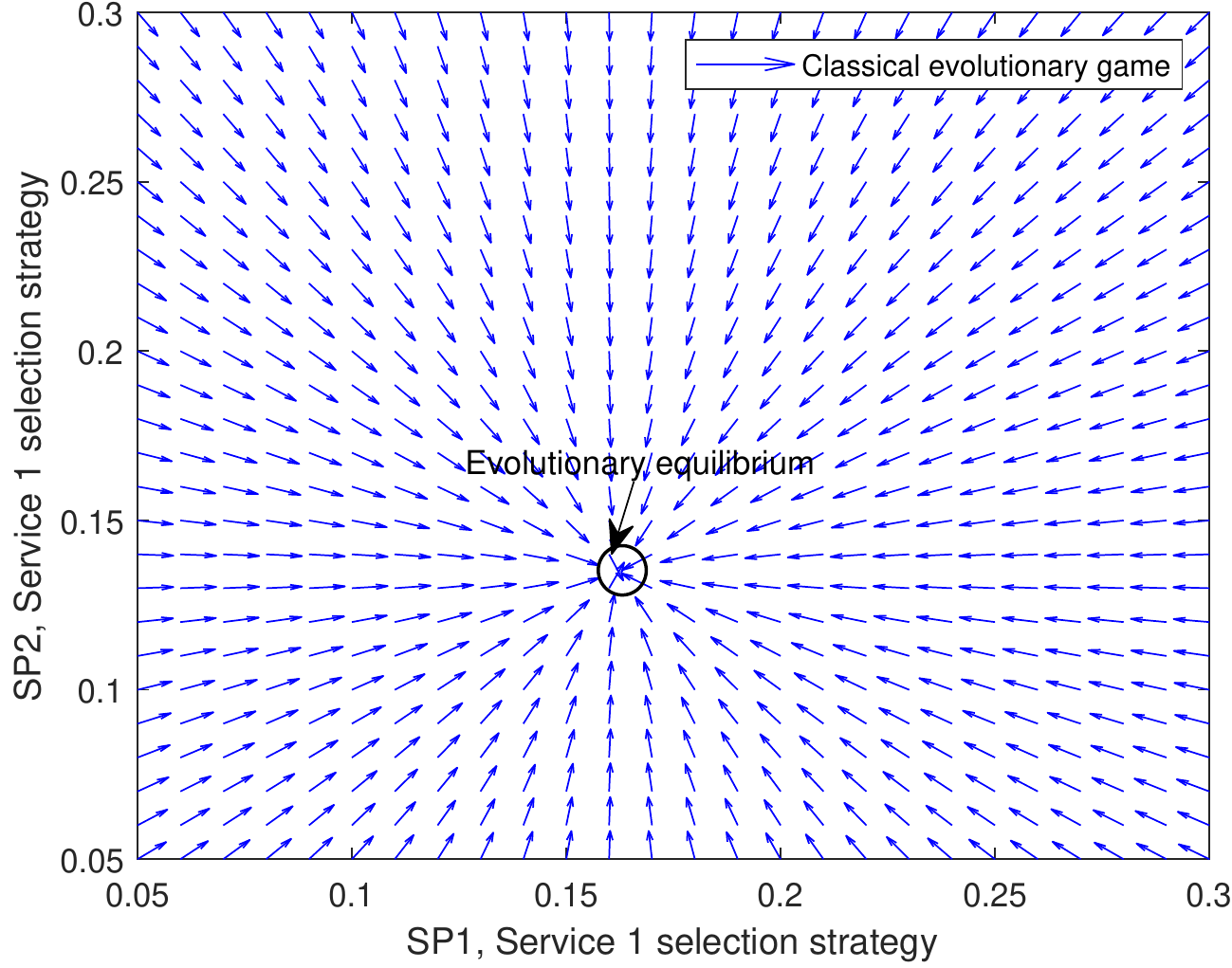}}\quad
  \subcaptionbox{}[.3\linewidth][c]{%
    \includegraphics[width=1\linewidth]{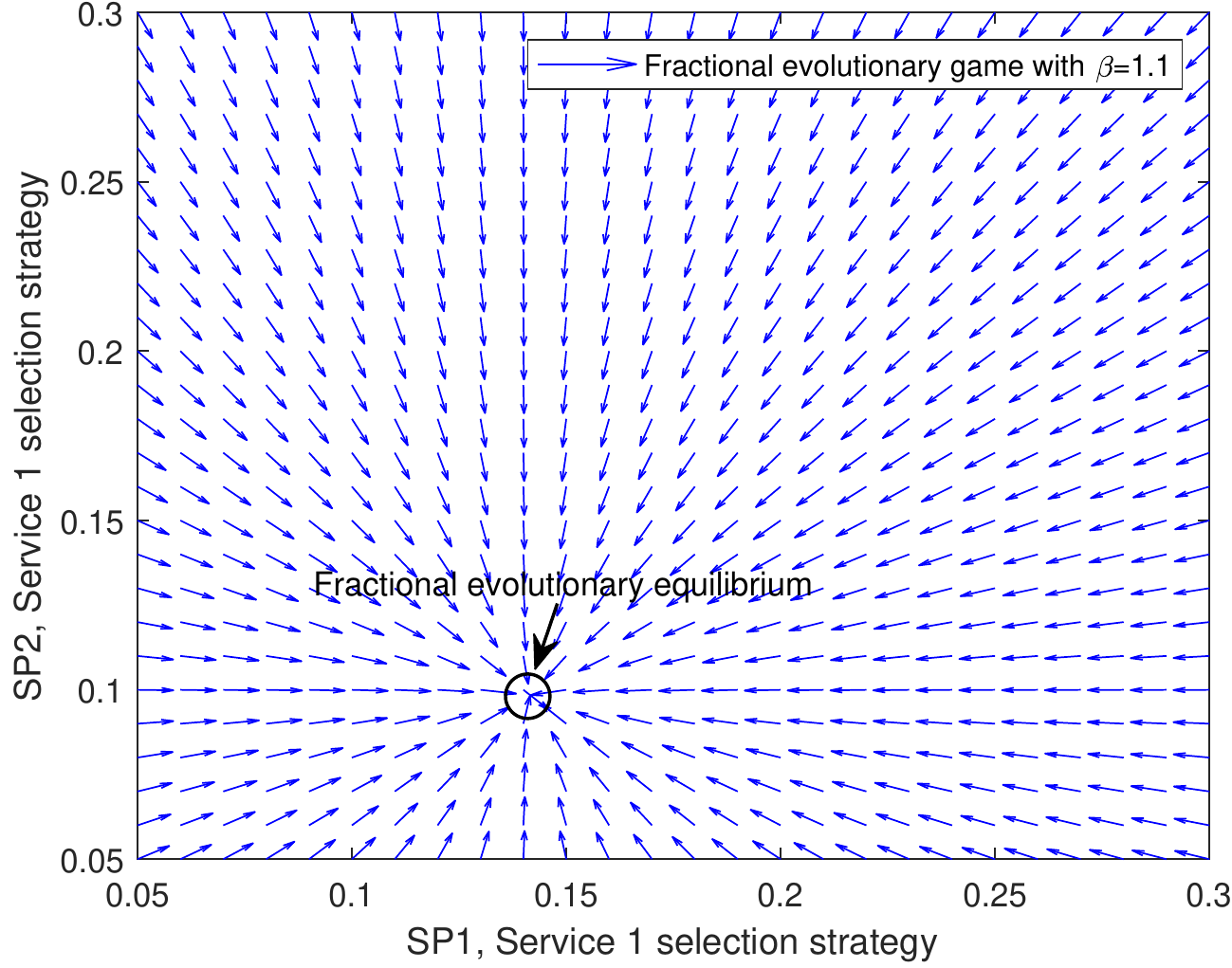}}\quad
  \subcaptionbox{}[.3\linewidth][c]{%
    \includegraphics[width=1\linewidth]{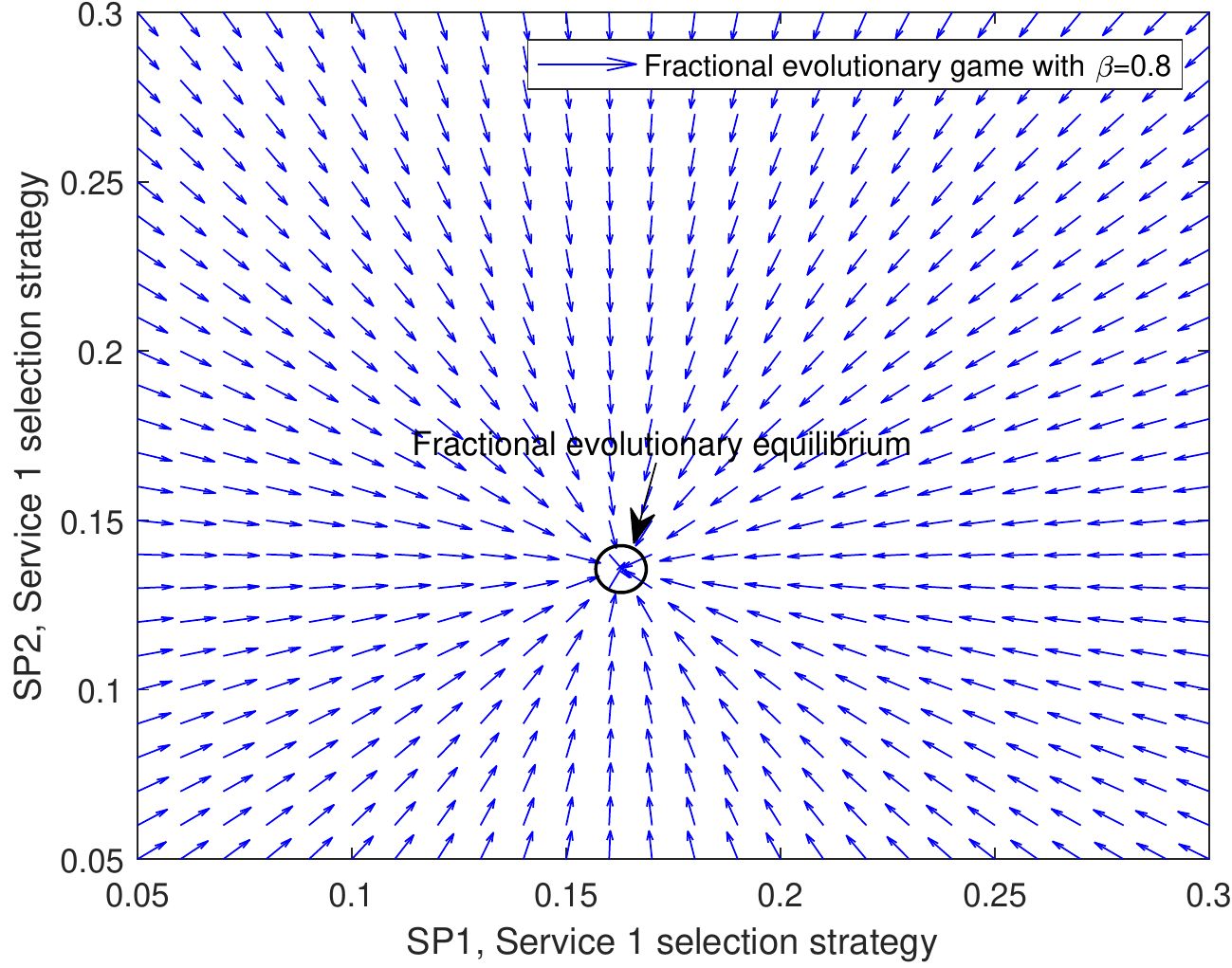}}

  \bigskip

  \subcaptionbox{}[.3\linewidth][c]{%
    \includegraphics[width=1\linewidth]{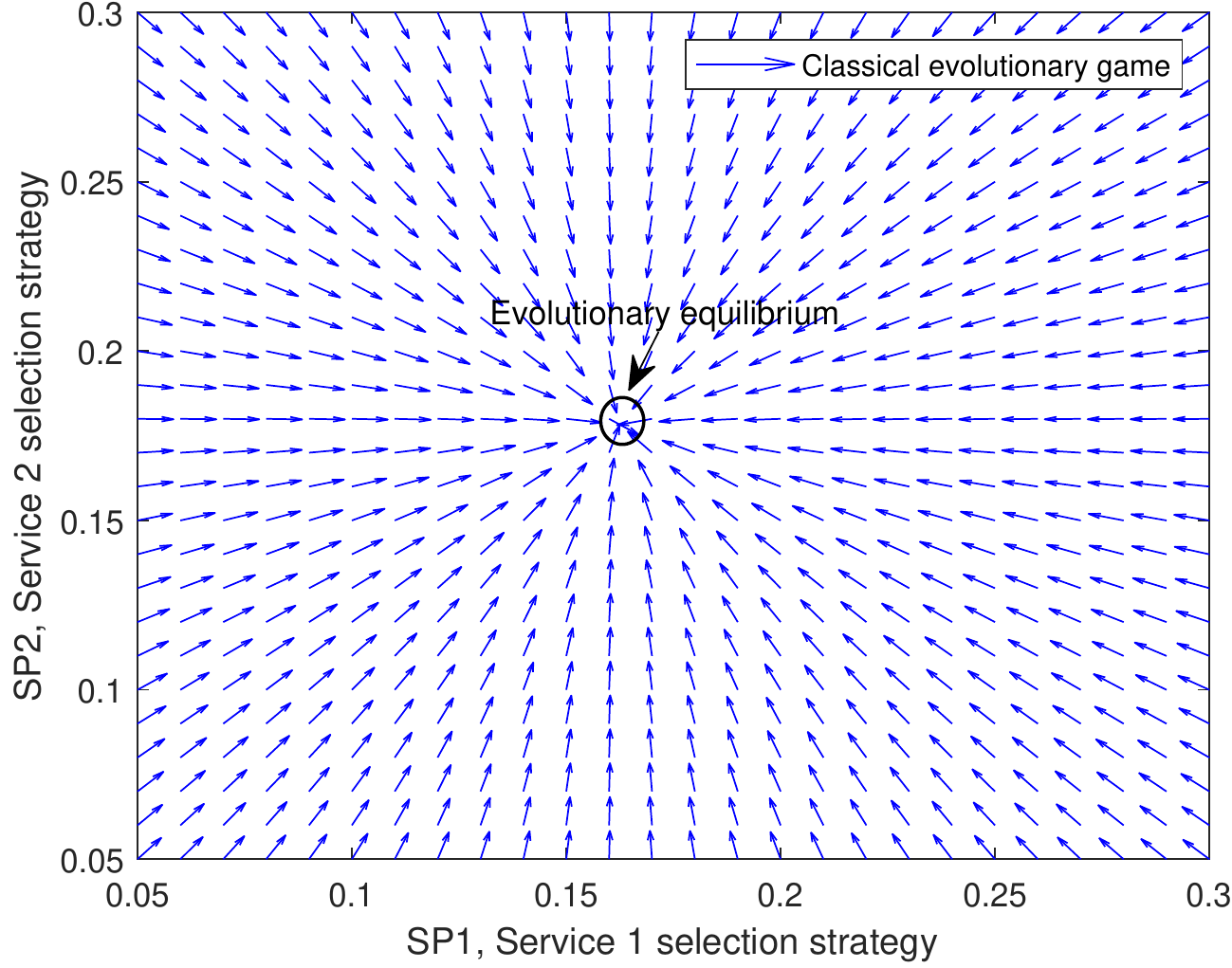}}\quad
  \subcaptionbox{}[.3\linewidth][c]{%
    \includegraphics[width=1\linewidth]{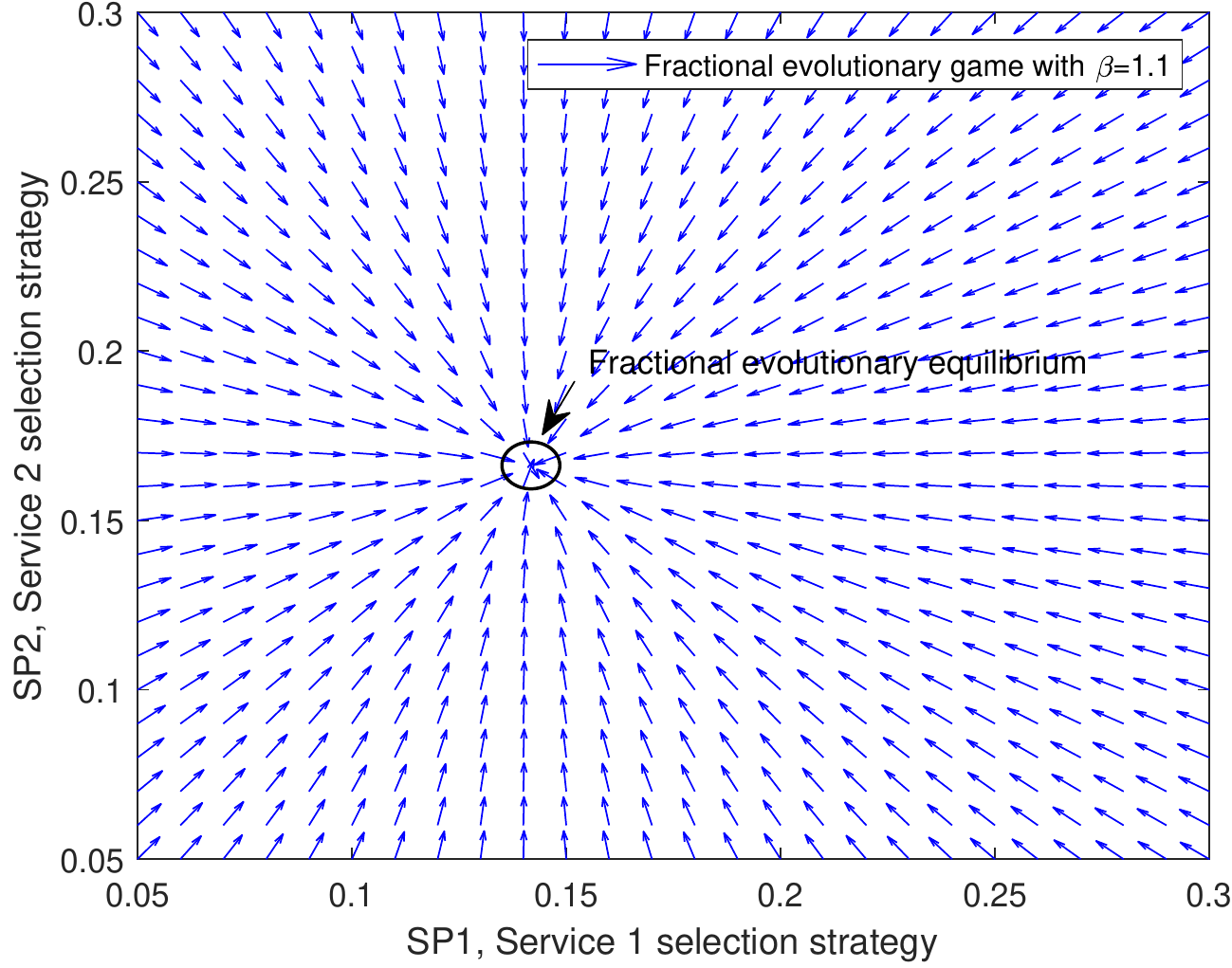}}\quad
  \subcaptionbox{}[.3\linewidth][c]{%
    \includegraphics[width=1\linewidth]{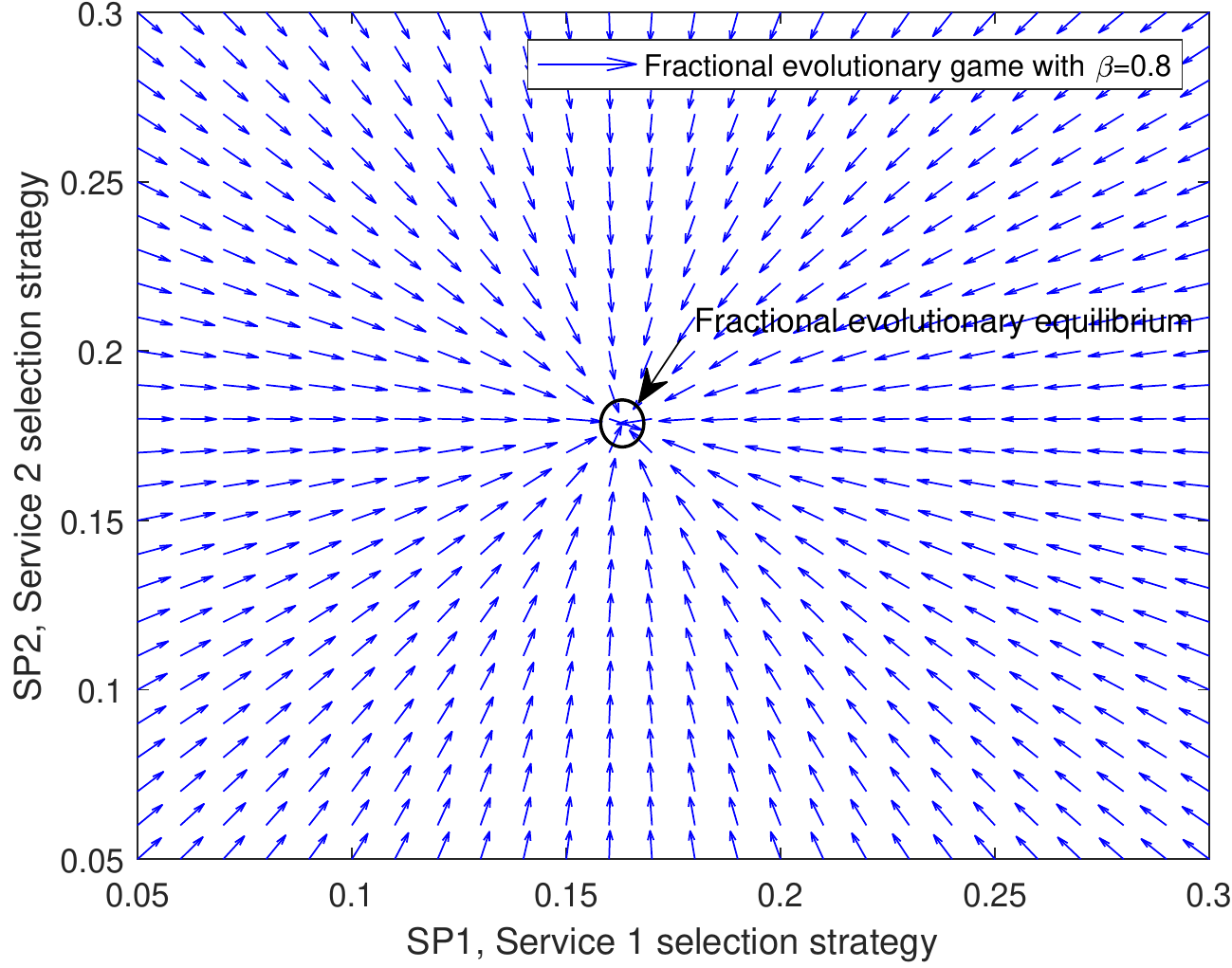}}
\caption{Direction field of the replicator dynamics when $t = 200$.}
	\label{direction_field}
		\vspace{-0.4cm}
\end{figure*}

To show that the users' strategies in the proposed game approaches can be stabilized at the equilibrium, we present the direction field of the replicator dynamics. As illustrated in Fig.~\ref{direction_field}, the strategies of the users eventually reach the equilibrium strategy after a certain time, i.e., $t = 200$, that is represented by the black circles. We can take the results shown in Fig~\ref{direction_field}(a) as an example. In the figure, we show the replicator dynamics of selection strategy of SP 1, service 1 and SP 2, service 1 $(p_{1,1}$ and $p_{2,1}$).  Assuming that the strategies the users select services provided by SP 2 (i.e., $p_{2,1}$ and $p_{2,2}$) achieve the equilibrium, and $p_{1,1}= 1 -  p_{1,2} - p_{1,3} - p_{1,4} - p_{2,1} - p_{2,2}$. As seen, the users are able to adapt their strategies by following the directions of the arrows. Furthermore, any initial strategy eventually reach the equilibrium that verifies the stability of our proposed game approaches. That is similar to Figs.~\ref{direction_field}(b), (c), (d), (e), and (f).


\begin{figure*}[t]
  \centering
  \subcaptionbox{}[.3\linewidth][c]{%
    \includegraphics[width=1\linewidth]{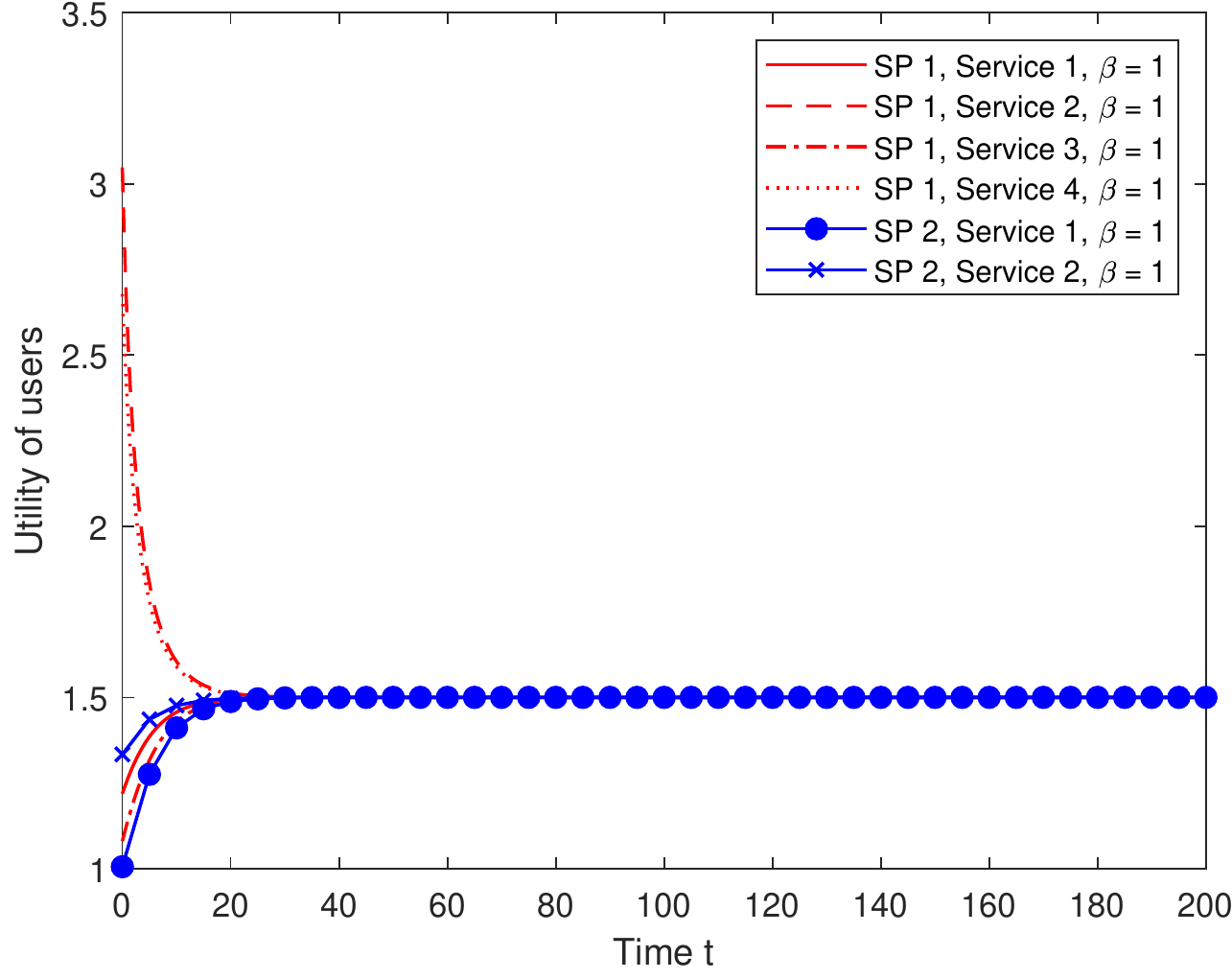}}\quad
  \subcaptionbox{}[.3\linewidth][c]{%
    \includegraphics[width=1\linewidth]{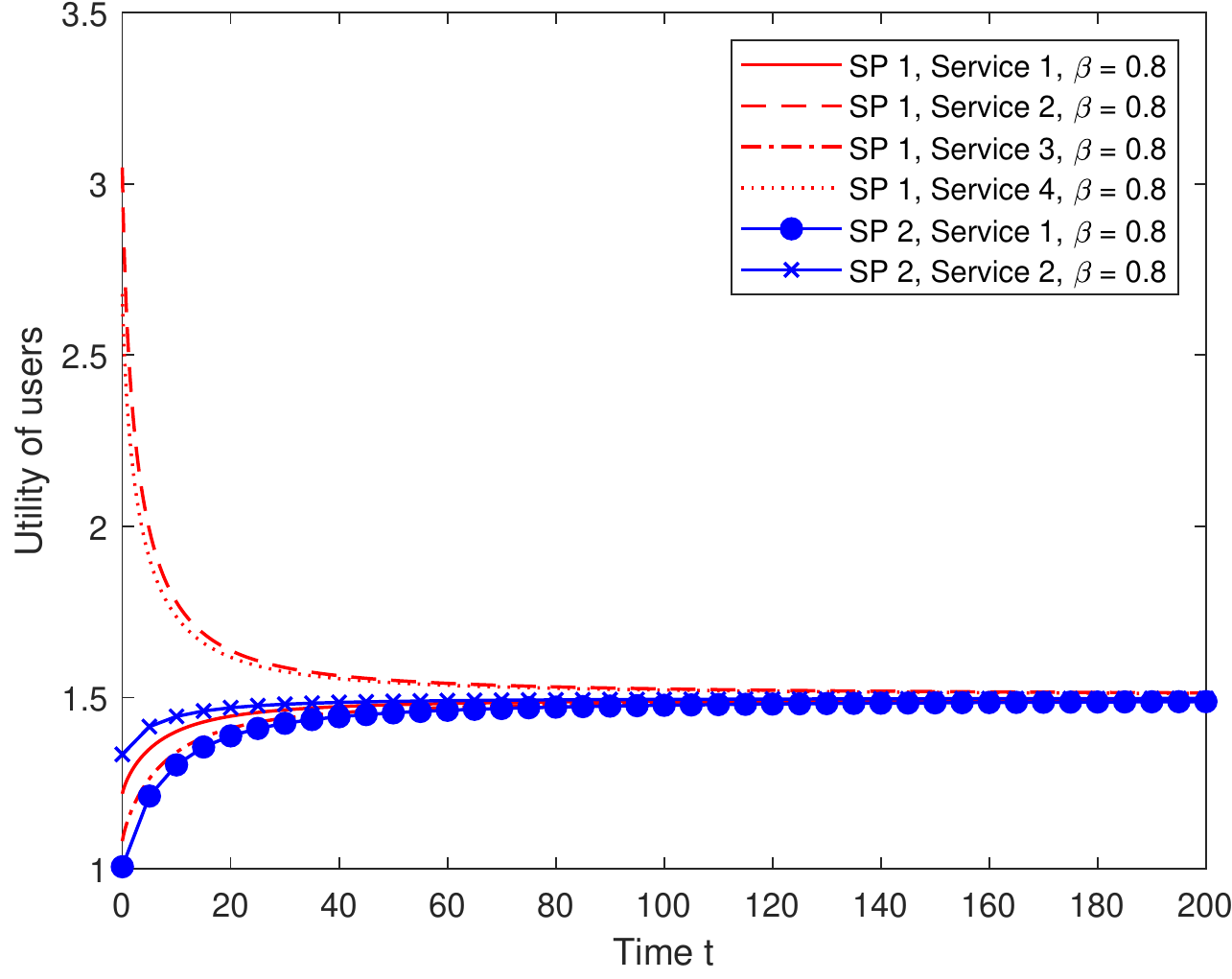}}\quad
  \subcaptionbox{}[.3\linewidth][c]{%
      \includegraphics[width=1\linewidth]{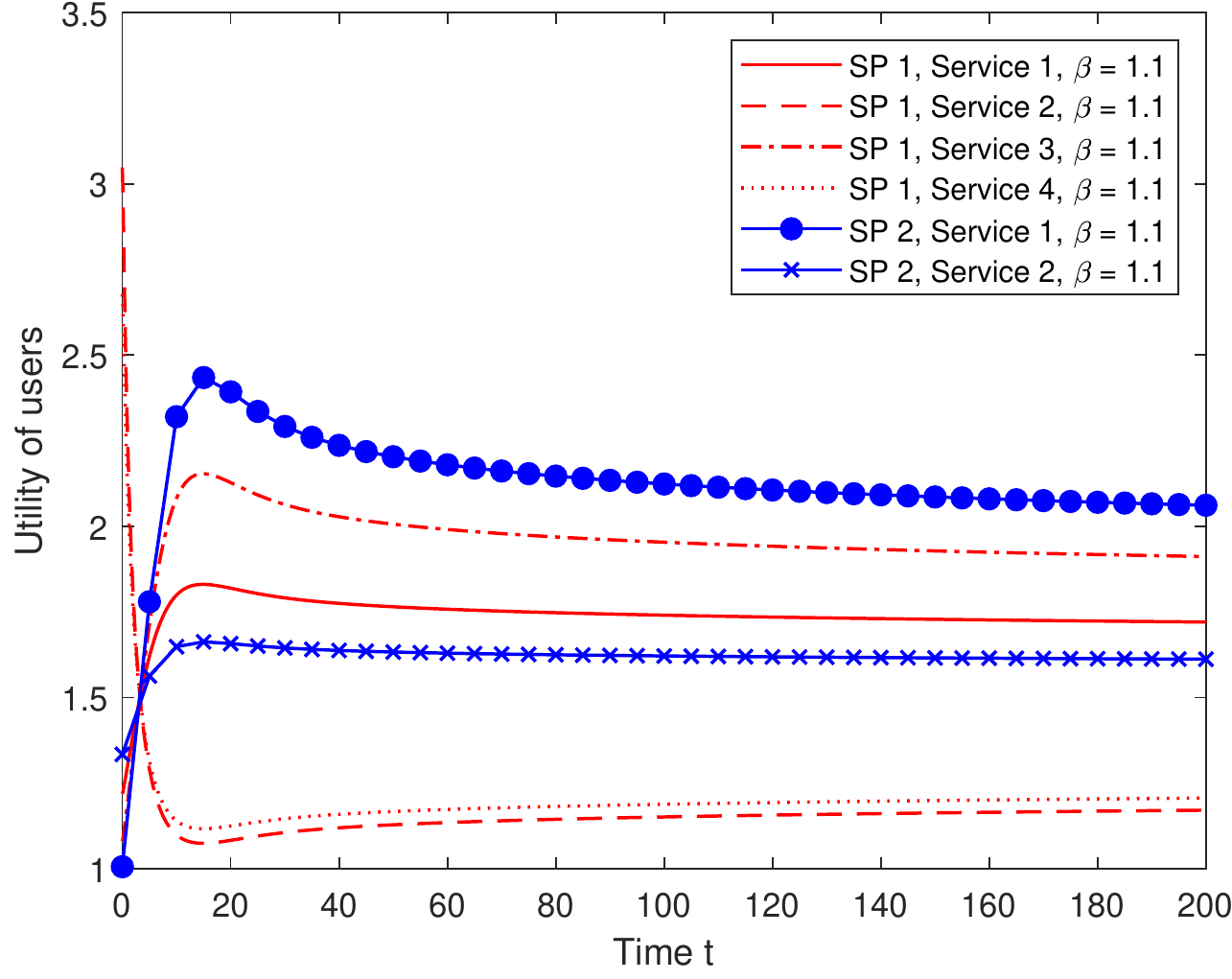}}

\caption{Utility of user groups.}
	\label{utility1}
		\vspace{-0.4cm}
\end{figure*}

Now, we discuss how the utilities of the users obtained at the equilibrium. Figures~\ref{utility1}(a), (b), and (c) show the utilities that the users achieve by selecting different SPs and network services over evolutionary time. As seen, the utilities of the users vary until the equilibrium is reached. At the equilibrium, the users have the same utility even if they select different SPs and services. The reason is that the evolutionary equilibrium is reached only when the utilities of the users choosing any SP and any service are equal to their expected utility.


\begin{figure*}[t]
  \centering
  \subcaptionbox{}[.3\linewidth][c]{%
    \includegraphics[width=1\linewidth]{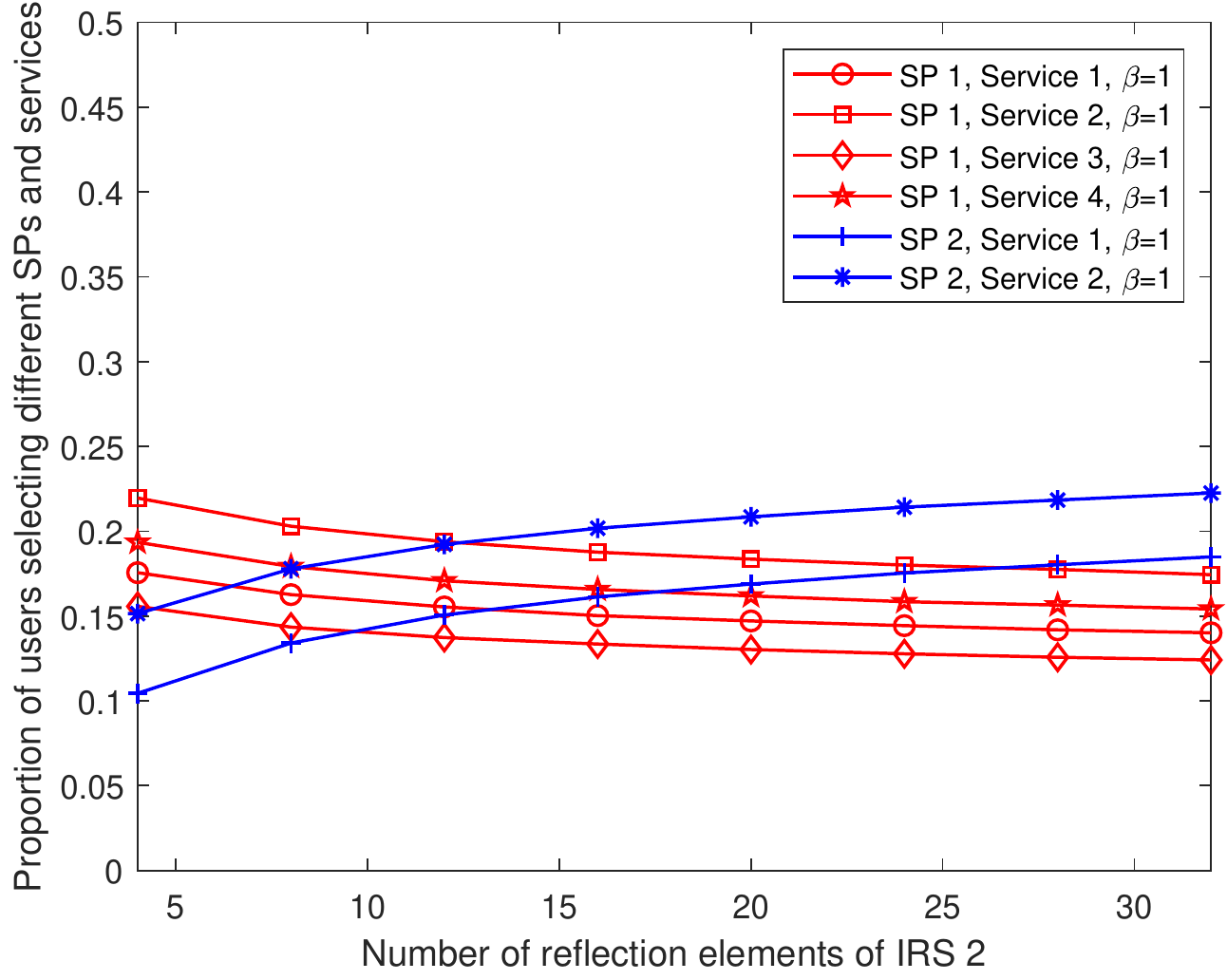}}\quad
  \subcaptionbox{}[.3\linewidth][c]{%
     \includegraphics[width=1\linewidth]{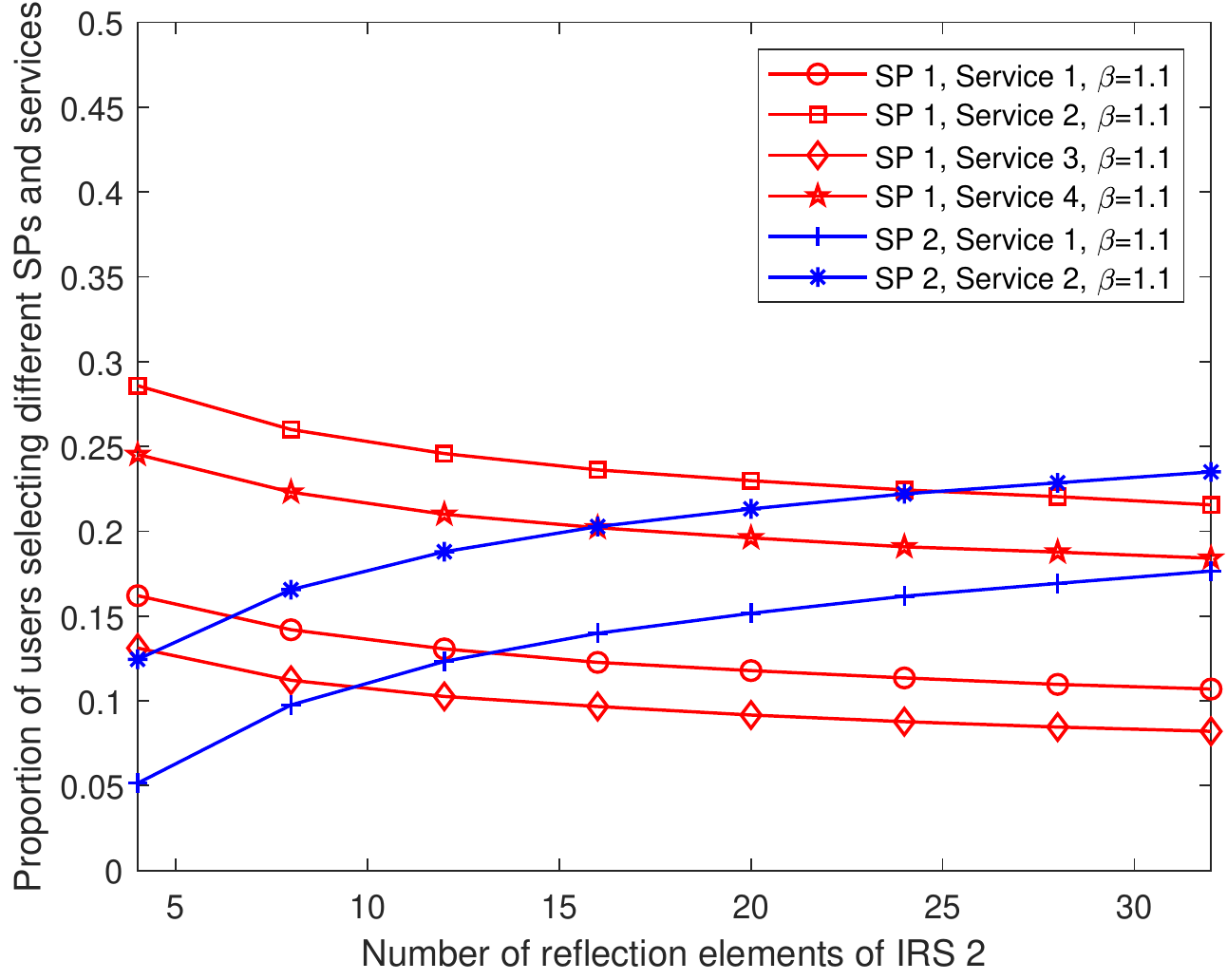}}\quad
  \subcaptionbox{}[.3\linewidth][c]{%
\includegraphics[width=1\linewidth]{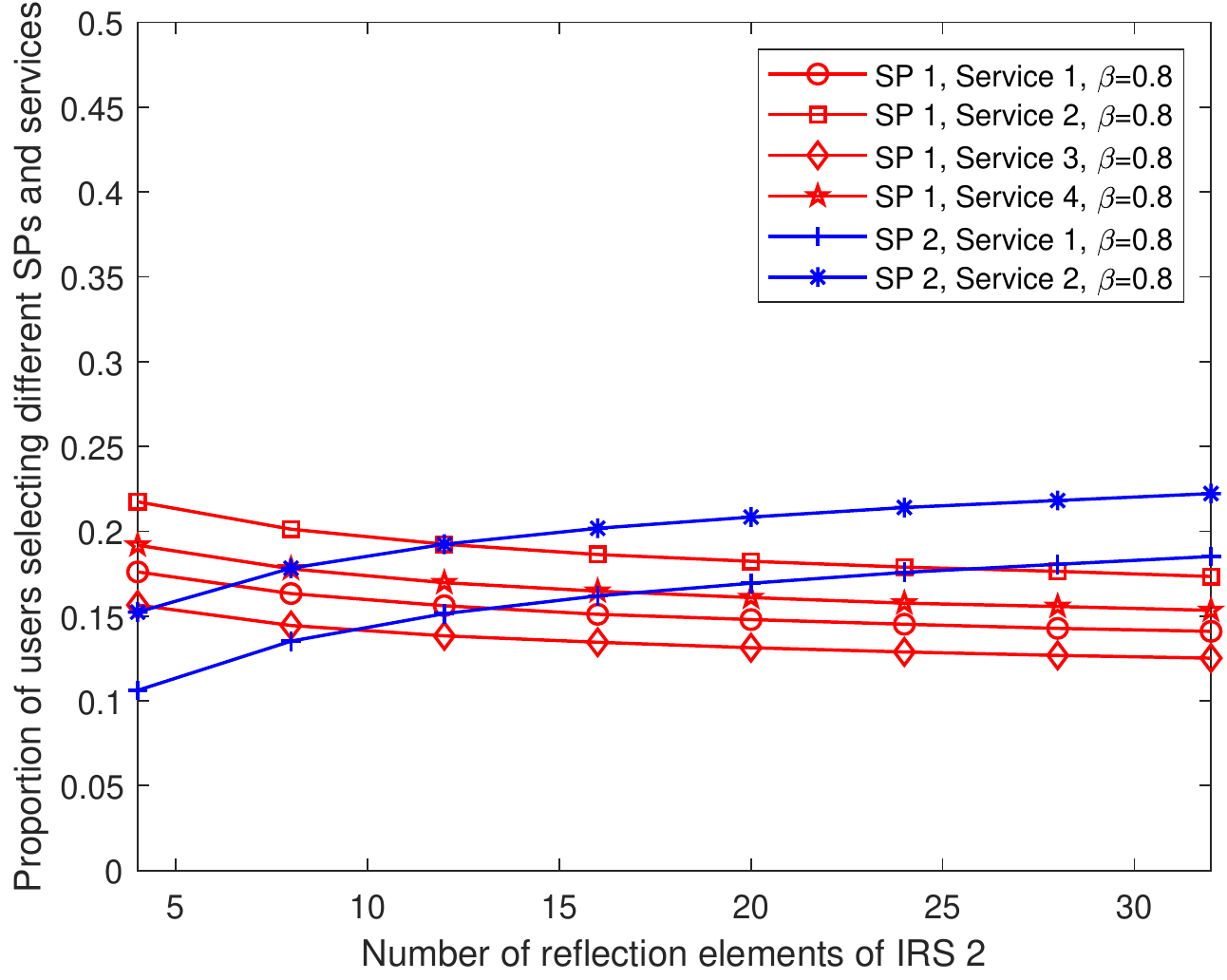}}

\caption{Proportion of users selecting different service as the size of IRS provided by SP 2 varies.}
	\label{proportioin_change_size_IRS}
\end{figure*}

Next, we investigate the impact of sizes of IRSs on the proportions of users selecting different SPs and services. Figures~\ref{proportioin_change_size_IRS} (a), (b), and (c) show the results obtained from the evolutionary game with $\beta=1, 1.1$ and $0.8$, respectively. In particular, we vary the size $K_{2}$ of IRS 2 of SP 2. As shown in Fig.~\ref{proportioin_change_size_IRS}(a), as $K_{2}$ increases, the proportions of users selecting services provided by SP 2 increase since the throughput and utility obtained by the users selecting services provided by SP 2 increase. However, as the size of IRS 2 is large, the increasing rate tends to be slower. This is because of that the users pay a very high resource cost if they select the services provided by SP 2, and thus they tend to select the services provided by SP 1. Figures~\ref{proportioin_change_size_IRS}(b) and (c) have the same pattern as Fig. 8(a), and the results can be explained in the same way.

\begin{figure*}[t]
  \centering
  \subcaptionbox{}[.3\linewidth][c]{%
    \includegraphics[width=1\linewidth]{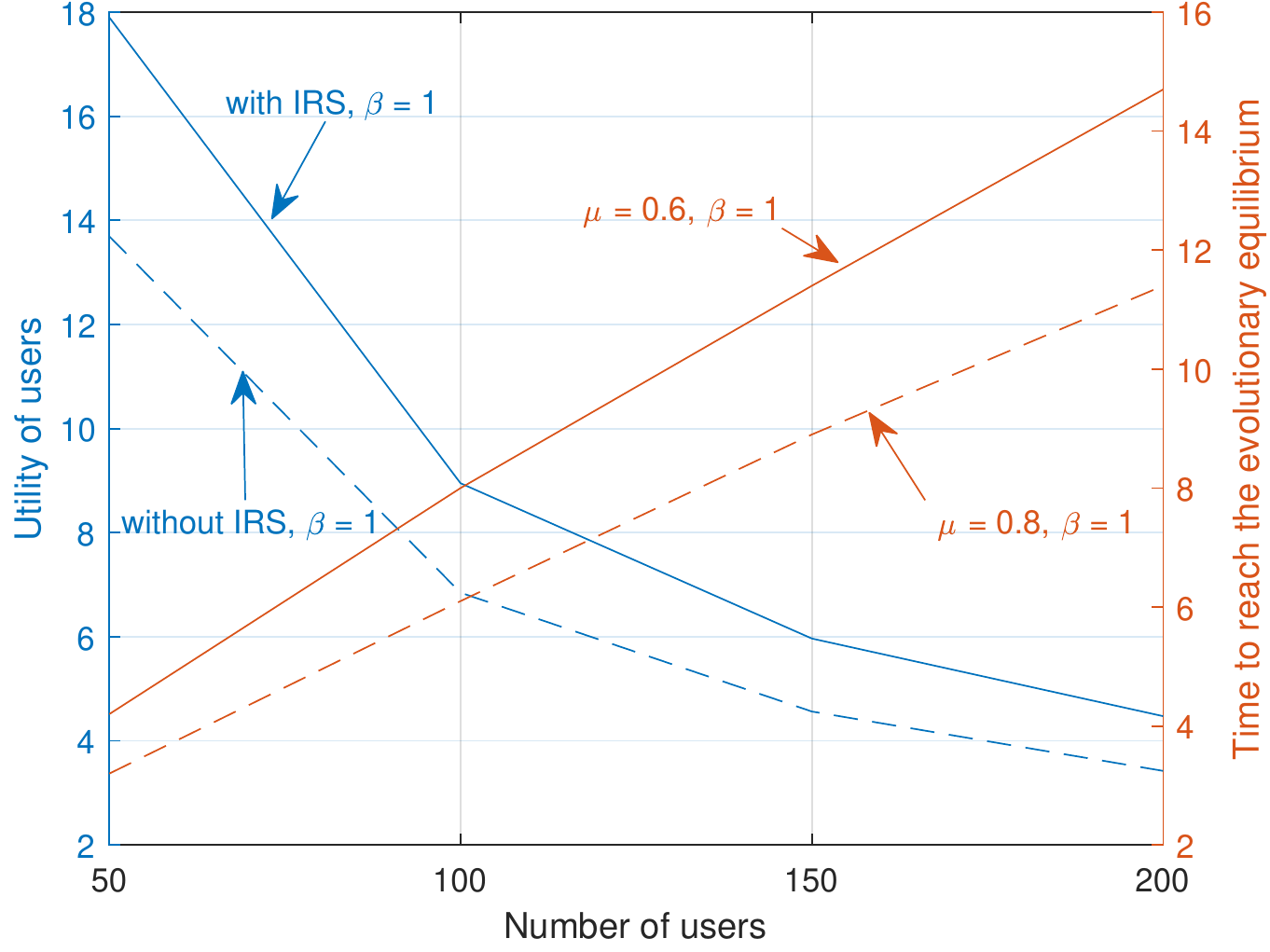}}\quad
  \subcaptionbox{}[.3\linewidth][c]{%
    \includegraphics[width=1\linewidth]{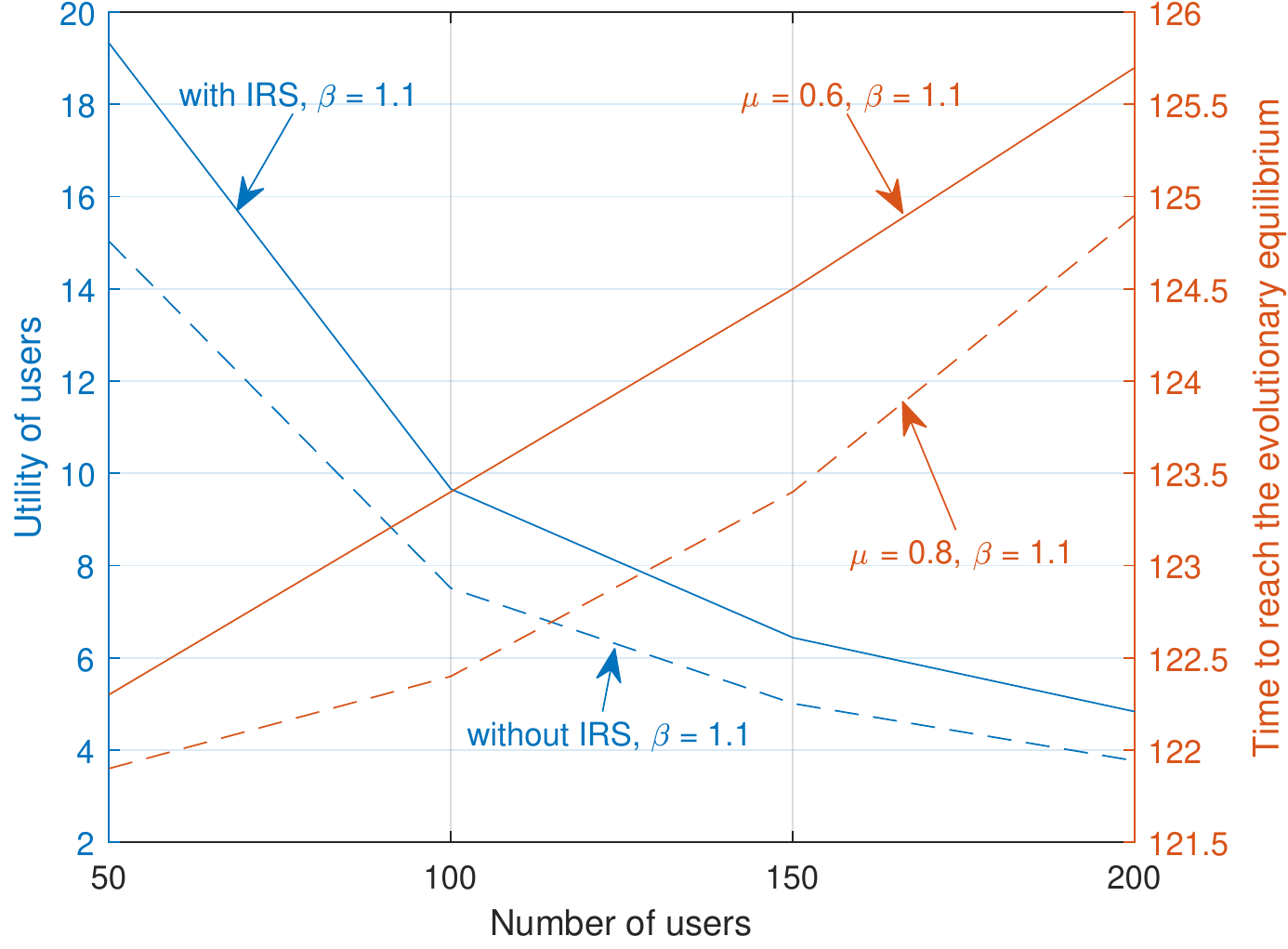}}\quad
  \subcaptionbox{}[.3\linewidth][c]{%
    \includegraphics[width=1\linewidth]{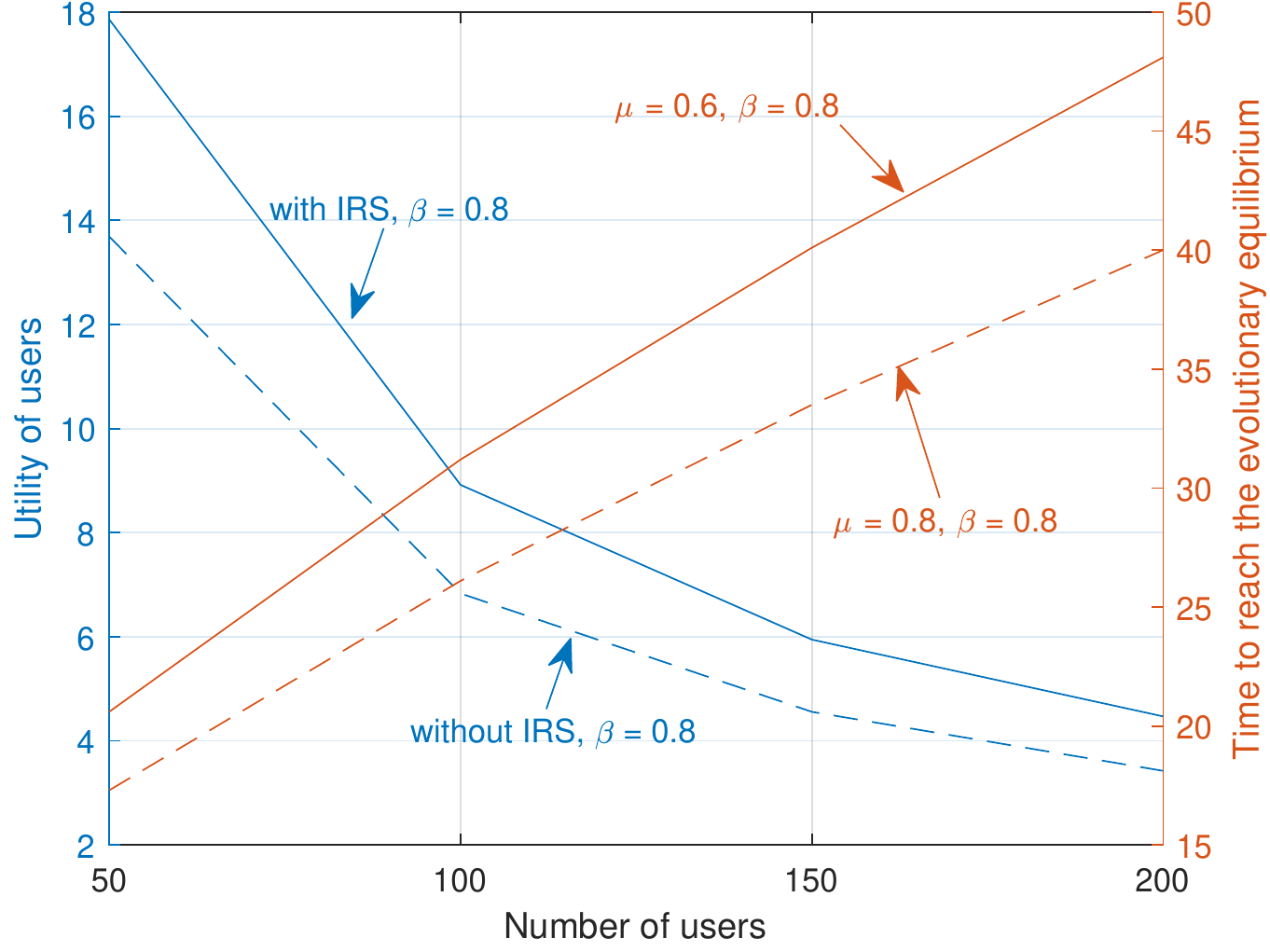}}

\caption{Impact of the number of users on the total utility and the time to reach the equilibrium.}
	\label{learning_rate}
		\vspace{-0.4cm}
\end{figure*}

Note that the time to reach the equilibrium can be different depending on the learning rate $\mu$ and the number of users. Figures~\ref{learning_rate}(a), (b), and (c) show the results obtained from the fractional evolutionary games with $\beta=1, 1.1$ and $0.8$, respectively. As seen, the evolutionary equilibrium is reached faster as $\mu$ is higher since the frequency of the strategy adaptation of the users is higher. Moreover, the games need more time to converge to the equilibrium as the number of users $N$ increases. Note that as $N$ increases, the total utility of the users decreases. This can be explained based on~(\ref{expected_data_rate}), more users share the fixed amount of bandwidth that results in reducing the throughput of the users. Moreover, from Fig.~\ref{learning_rate}(a), the total utility of users with IRSs is much higher than that of users without IRSs. The reason is that deploying IRSs increases the throughput of the users.


\begin{figure*}[t]
  \centering
  \subcaptionbox{}[.3\linewidth][c]{%
    \includegraphics[width=1\linewidth]{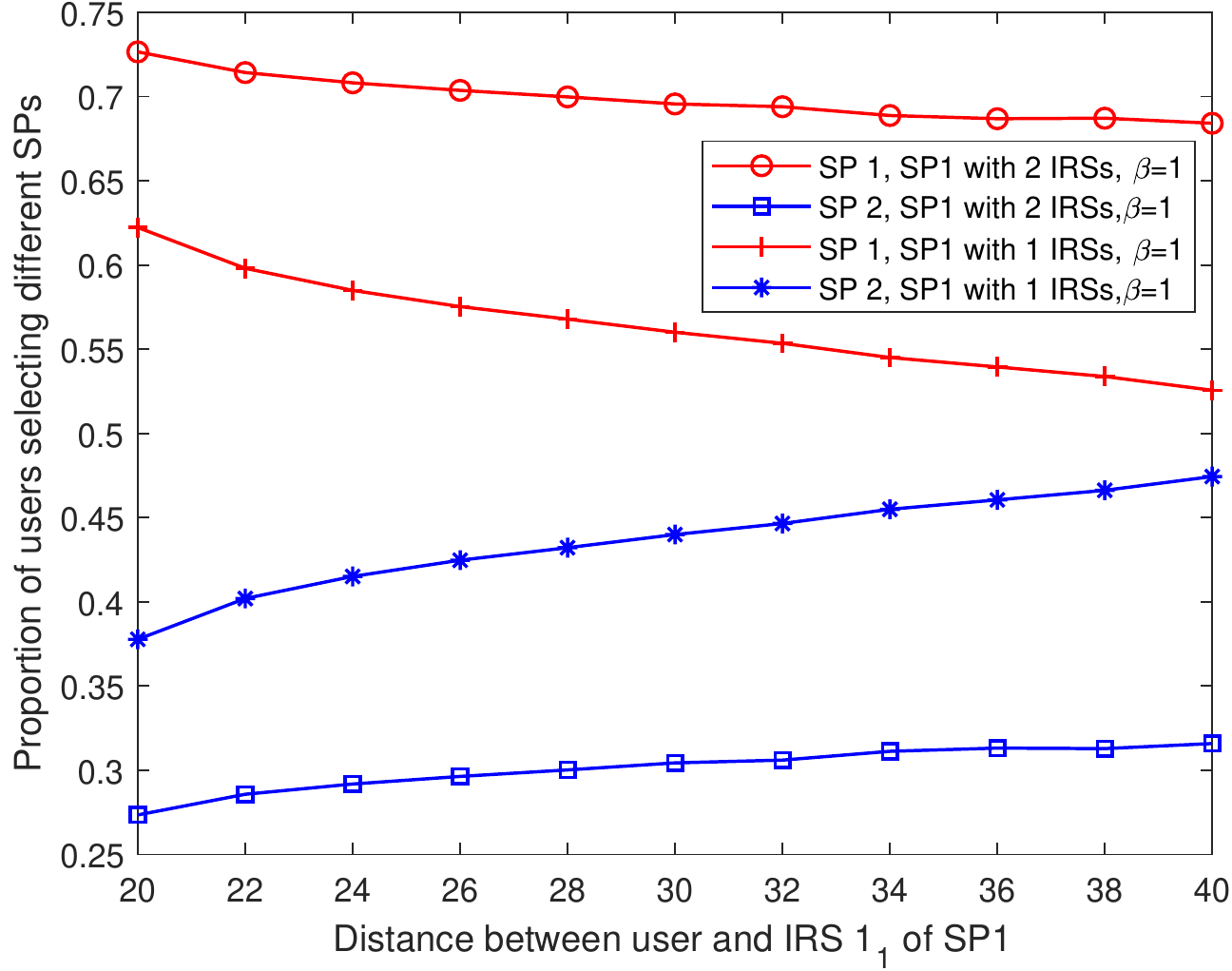}}\quad
  \subcaptionbox{}[.3\linewidth][c]{%
    \includegraphics[width=1\linewidth]{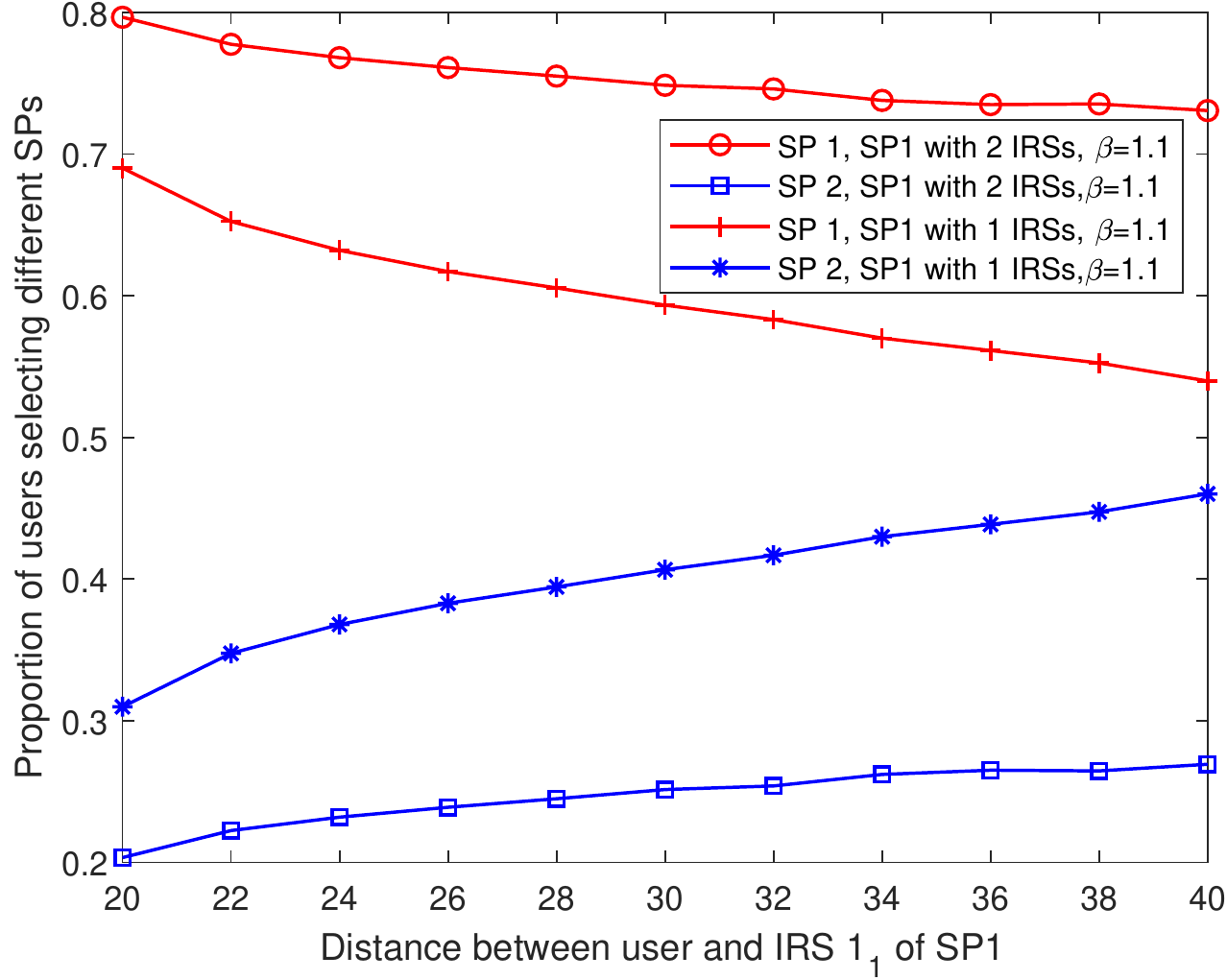}}\quad
  \subcaptionbox{}[.3\linewidth][c]{%
    \includegraphics[width=1\linewidth]{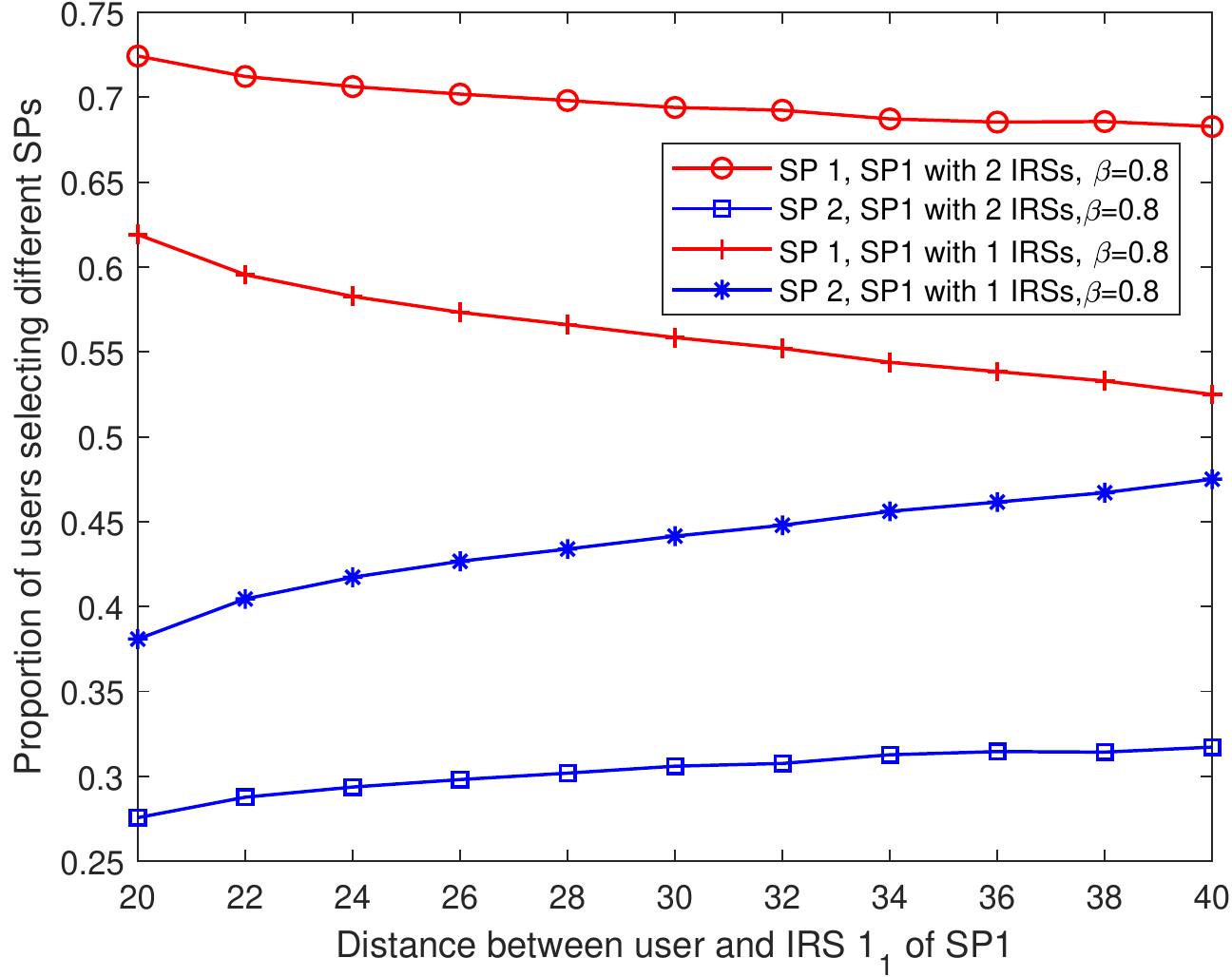}}

\caption{Proportions of users selecting different SPs vs. distance.}
	\label{proportion_change_distance}
		\vspace{-0.4cm}
\end{figure*}

Next, we discuss how the mobility of the users impacts the selection strategies of the users. In particular, we evaluate the proportions of users selecting different SPs as the distance between the users and IRS 1 provided by SP 1 varies. The results for the evolutionary games with $\beta=1, 1.1$ and $0.8$ are respectively shown in Figs.~\ref{proportion_change_distance}(a), (b), and (c). As observed from Figs.~\ref{proportion_change_distance}(a), (b), and (c), the proportion of users selecting SP 2 increases as the distance between the users and IRS 1 of SP 1 increases. This is because of that the throughput obtained by the users selecting services of SP 1 decreases. Thus, the users are willing to select services of SP 2. In addition, we consider the case as SP 1 deploys one IRS. As observed from Figs.~\ref{proportion_change_distance}(a), (b), and (c), as SP 1 deploys 1 IRS, the proportions that the users select SP 1 is lower than those that SP 1 deploys 2 IRSs. Especially,    decreases more slowly than those that SP 1 deploys 2 IRSs. Interestingly, as SP 1 deploys 1 IRS, the proportions that the users select SP 1 decrease faster than those that SP 1 deploys 2 IRSs. This implies that by deploying more IRSs, the SP can further improve the QoS of the users and can prevent the users to select the network service of other SPs.

\begin{figure*}[t]
  \centering
  \subcaptionbox{}[.3\linewidth][c]{%
    \includegraphics[width=1\linewidth]{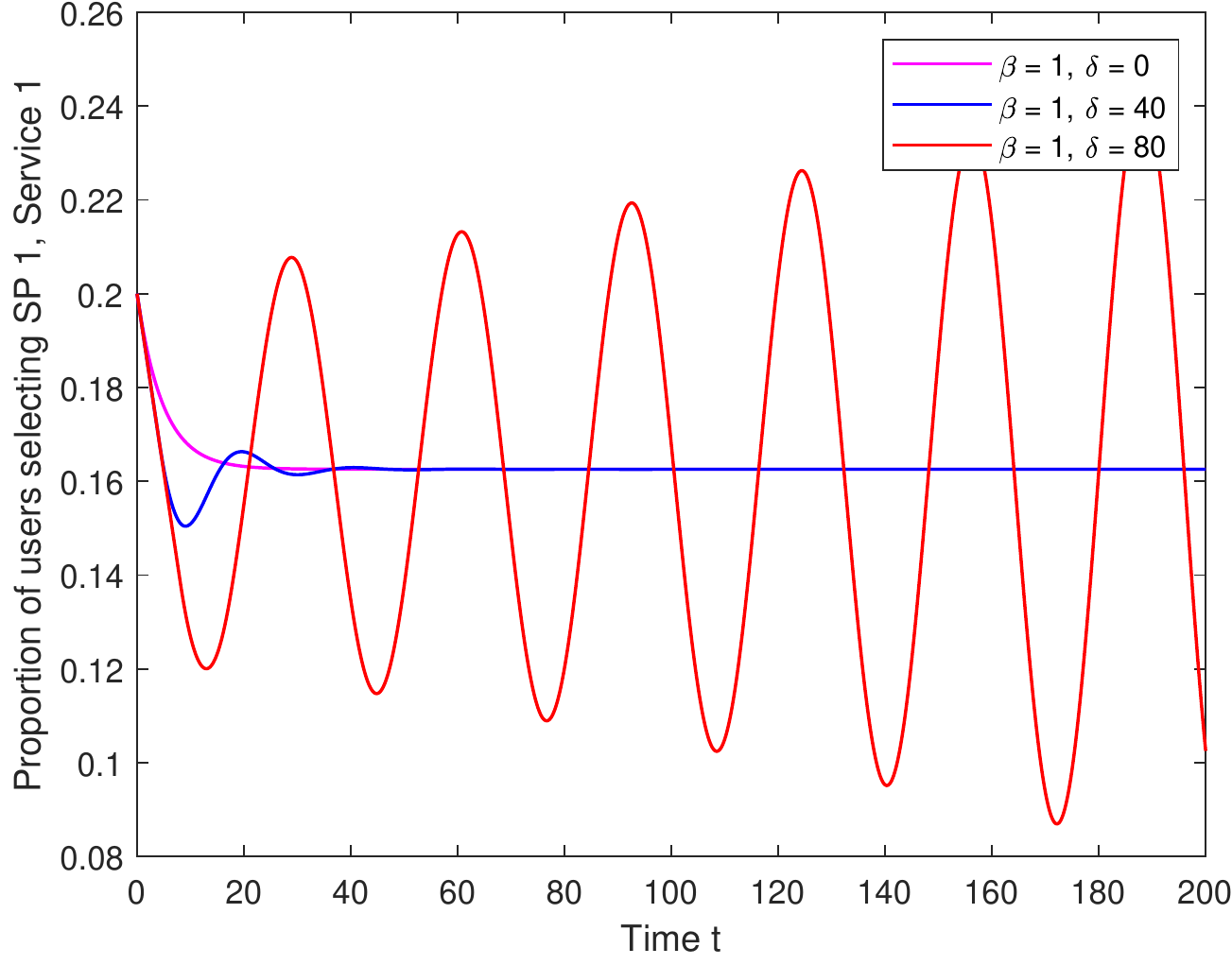}}\quad
  \subcaptionbox{}[.3\linewidth][c]{%
    \includegraphics[width=1\linewidth]{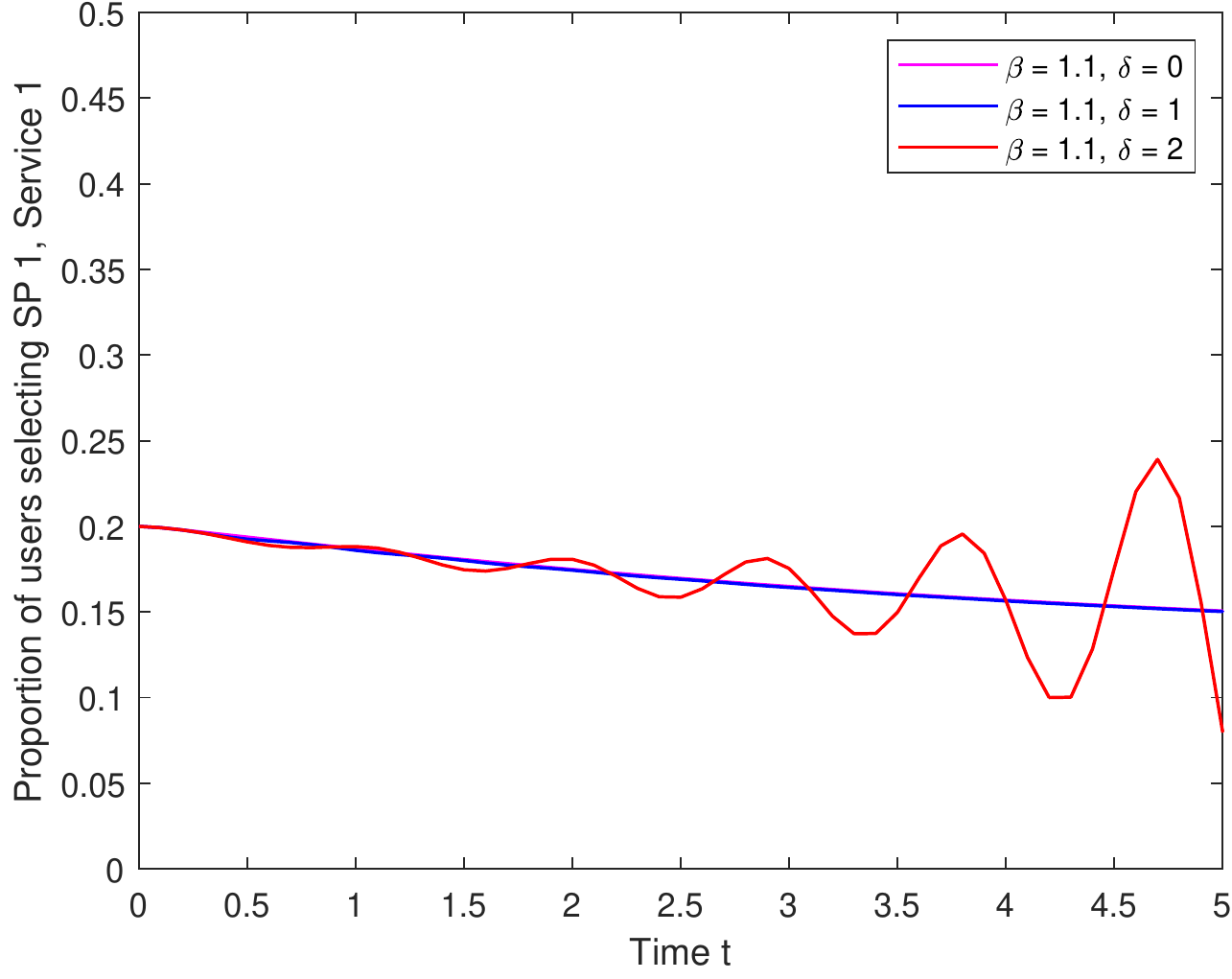}}\quad
  \subcaptionbox{}[.3\linewidth][c]{%
    \includegraphics[width=1\linewidth]{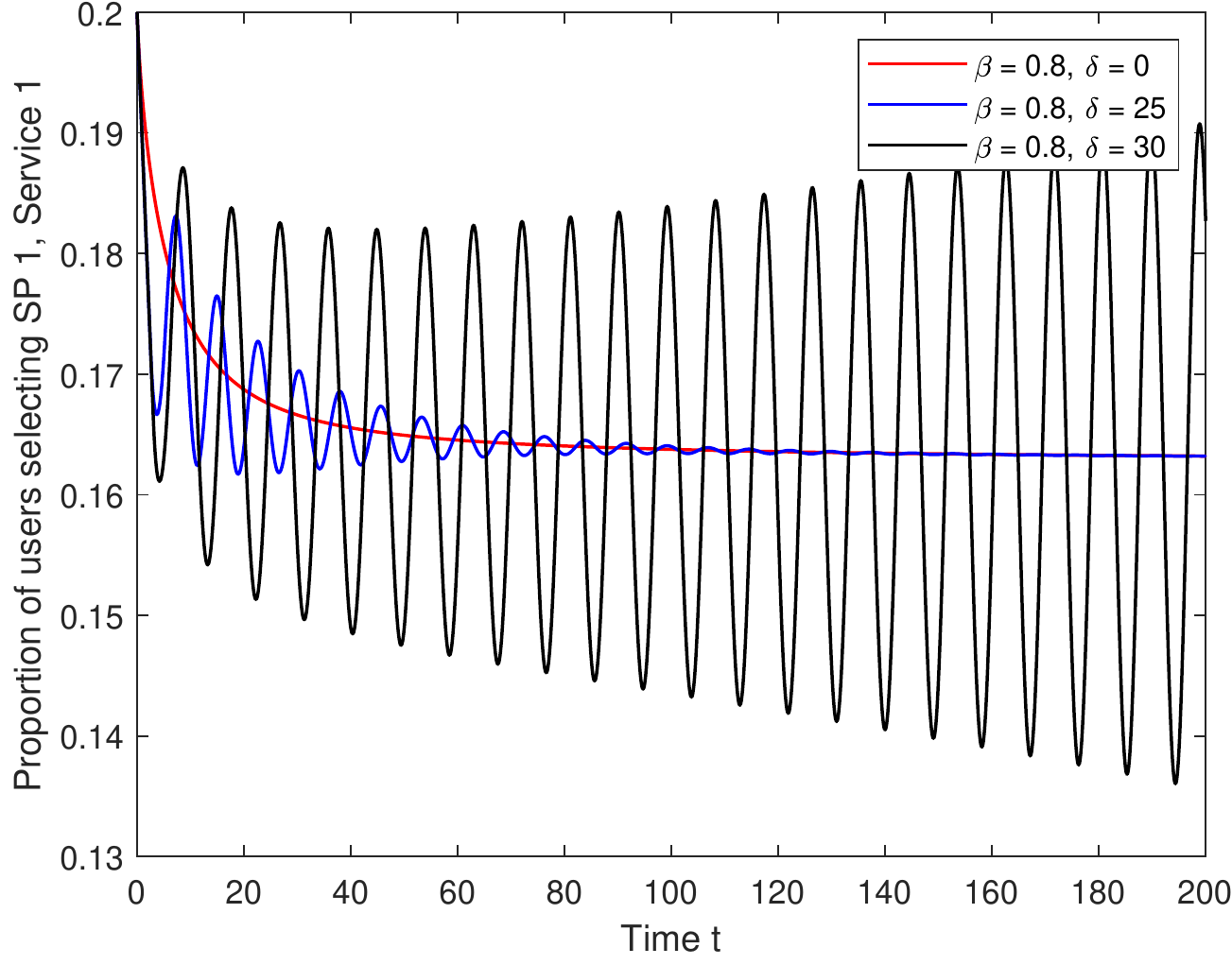}}

\caption{Proportion of the users selecting SP1 and Service 1 with different delays.}
	\label{proportion_delay}
		\vspace{-0.4cm}
\end{figure*}

Next, we discuss how the proportions of users selecting different SPs and services depend on the information delay $\delta$. For the evaluation purpose, we consider the proportion of the users selecting Service 1 provided by SP 1 as shown in Figs.~\ref{proportion_delay}(a) (b), and (c). First, we discuss the results obtained by the classical evolutionary game shown in Figs.~\ref{proportion_delay}(a). This figure shows the proportion of the users choosing Service 1 of SP 1 when the users use information for their decisions at $\delta= 0, 40$  and $80$. As seen, when $\delta > 0$, there is a fluctuating dynamics of strategy adaptation. In particular, as $\delta=40$, the game can sill converge to the equilibrium that is the same as the case when $\delta = 0$. However, as $\delta = 80$, the game cannot reach the equilibrium. These results mean that when the users use information with a small delay for their decisions, the game is still guaranteed for the convergence. In addition, we find in the figure that as $\delta=40$, the proportions of the users have more fluctuations and the game needs more time to reach the equilibrium, compared with the game where $\delta=0$. This implies that the convergence speed is slower as the users use information with larger delay. The results obtained by the fractional games with $\beta=0.8$ and $1.1$ are the same as that obtained by the classical evolutionary game. However, it seems to be that the service selection cannot reach the evolutionary equilibrium even when $\delta$ is small. For example, for fractional evolutionary game with $\beta=1.1$, as $\delta \leq 2$, the service selection cannot reach the evolutionary equilibrium. This is also a shortcoming of the fractional game in which the very outdated information may not be used for the user decisions.

 \begin{figure*}[h]
  \centering
  \subcaptionbox{}[.3\linewidth][c]{%
    \includegraphics[width=1\linewidth]{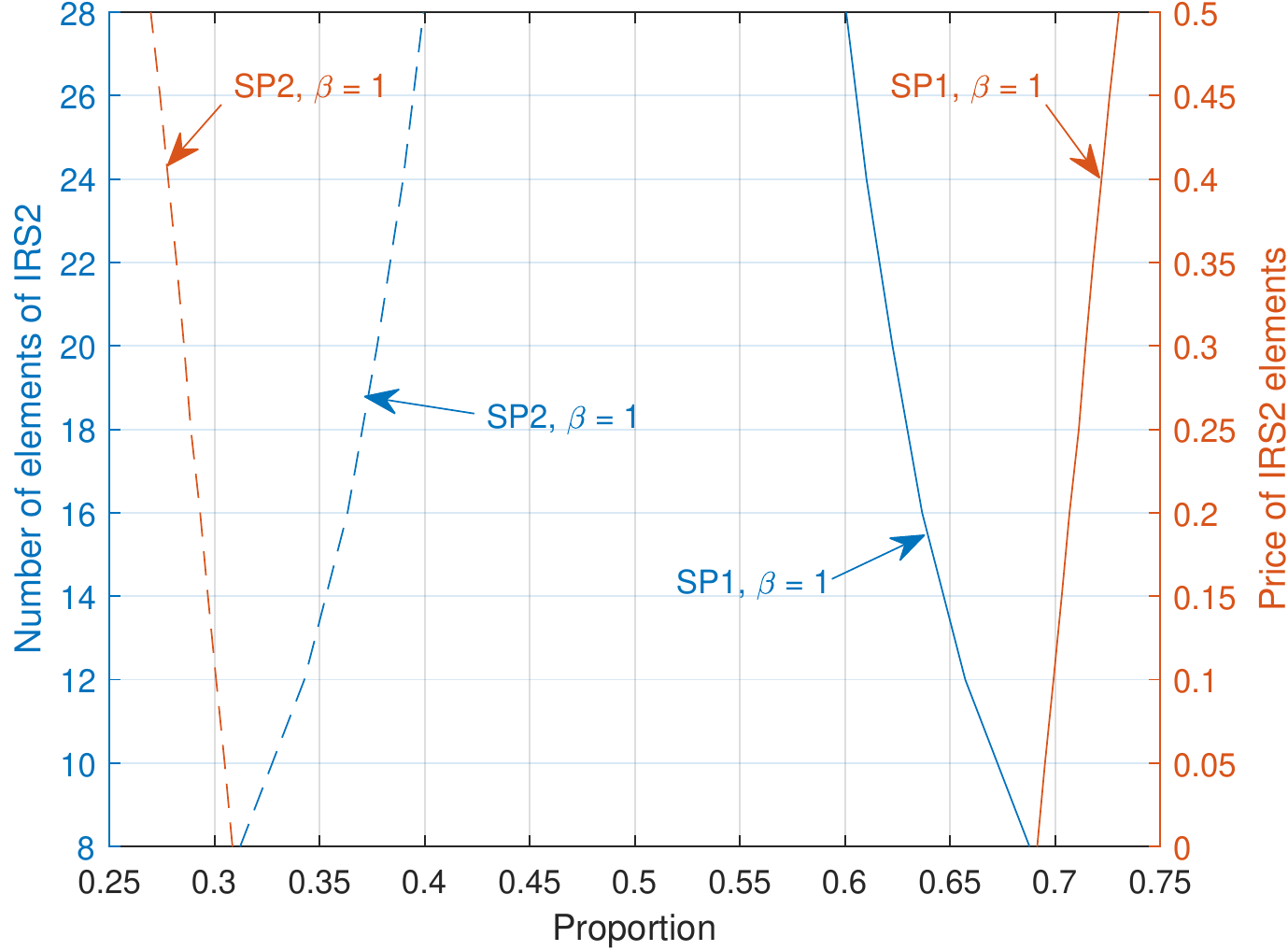}}\quad
  \subcaptionbox{}[.3\linewidth][c]{%
    \includegraphics[width=1\linewidth]{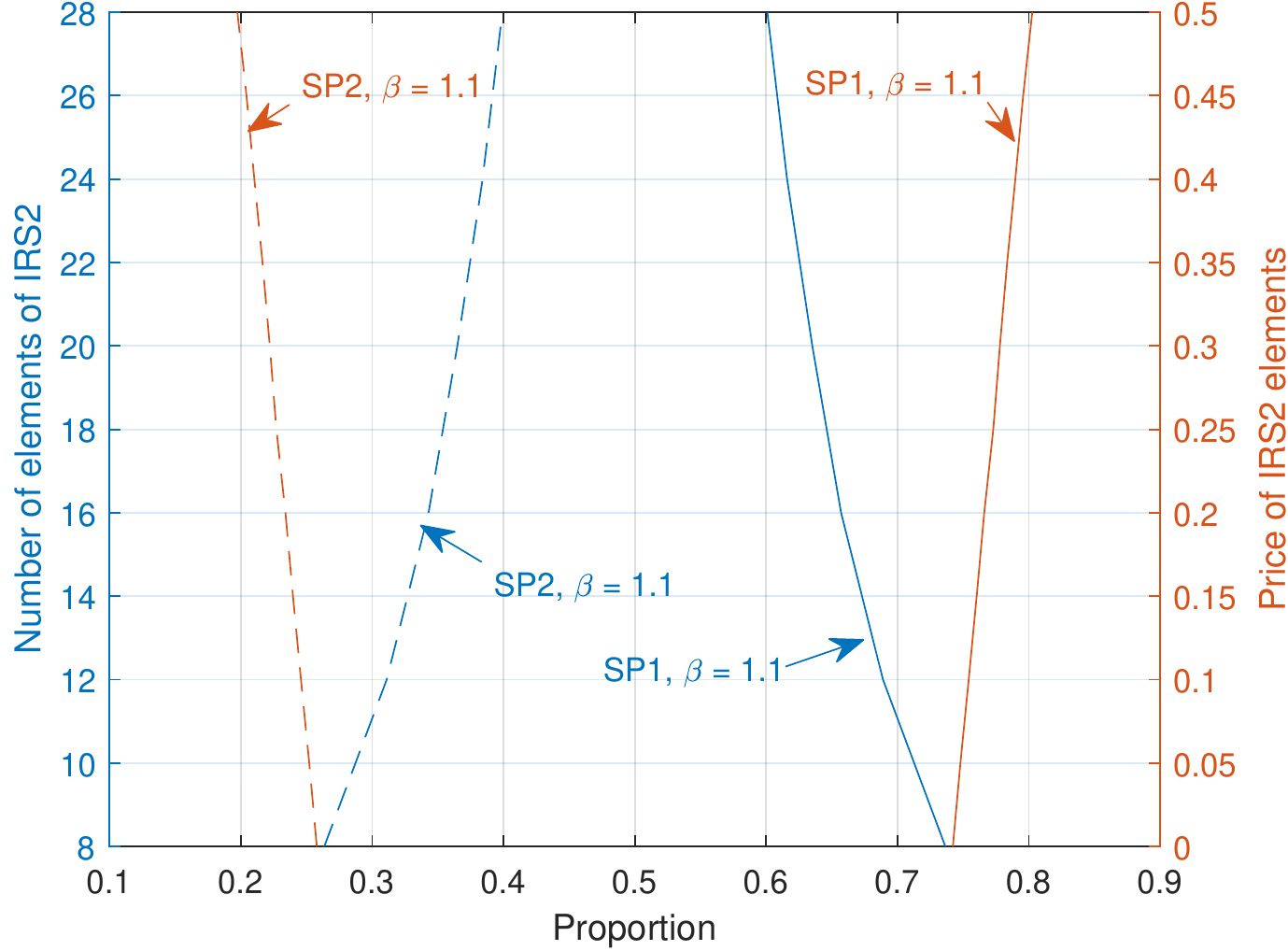}}\quad
  \subcaptionbox{}[.3\linewidth][c]{%
    \includegraphics[width=1\linewidth]{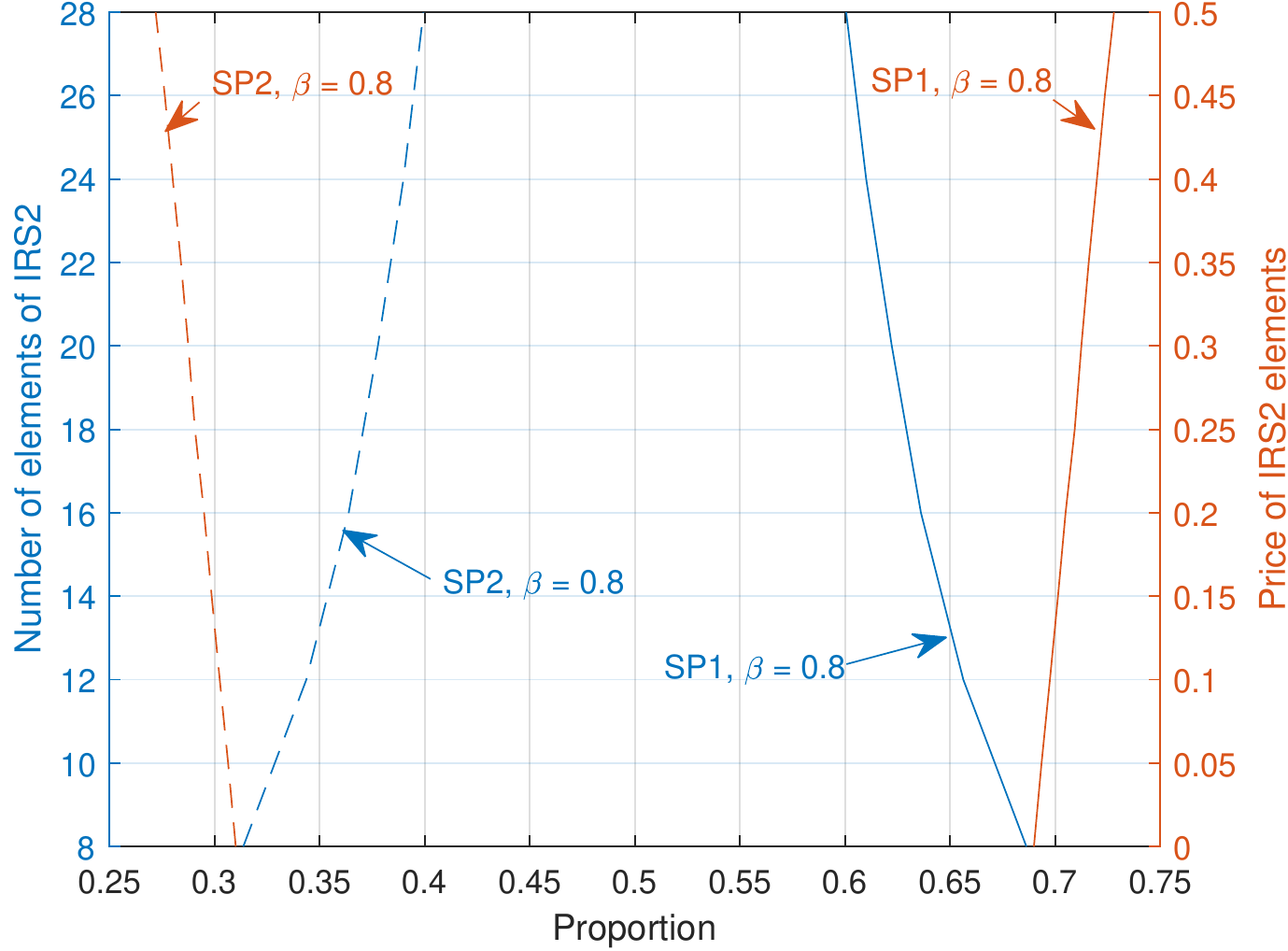}}

\caption{Proportion of the users selecting different SPs as the size of IRS and the price of IRS element provided by SP2 varies.}
	\label{proportion_price}
		\vspace{-0.4cm}
\end{figure*}

We finally discuss how the sizes of IRSs and the price of IRS 2 elements impact on the SP and service selection of the users. Figure~\ref{proportion_price}(a) shows that, for a given price, the proportion of users selecting SP 2 increases as the number of elements of IRS 2 increases. This is because of that the throughput obtained by the users selecting services of SP 2 increases.
Thus, the users are willing to select services of SP 2. This figure also shows the proportion of users selecting SP 1 increases as the price of IRS 2 elements increases. This is because of that the utility of users choose SP 2 will fall down as the price of IRS 2 elements increases. Therefore, the users are willing choose services of SP 1. The same explanations can be applied to the results obtained as the values of $\beta$ are $0.8$ and $1.1$ (Figs.~\ref{proportion_price}(b) and (c)).

\section{Conclusions}
\label{conclu}
We have proposed dynamic game frameworks for modeling the dynamic network selection of mobile users in the IRS-enabled terahertz network. First, we have adopted the classical evolutionary game in which the SP and service adaptation of the users is modeled as replicator dynamics. We have further considered the scenario in which the users use delayed information for their decision-making. In this scenario, we have analyzed the stability region of the delayed replicator dynamics. Furthermore, we have adopted the fractional evolutionary game that incorporates the memory effect to model the SP and service adaptation of the users. The proof of the existence and uniqueness of the game equilibrium has been provided. We have finally provided simulation results obtained by the proposed evolutionary game approaches. In addition, we have further discussed the selection behaviors of the users and compared the performance obtained by the proposed evolutionary games. For future work, we will study the SP and network service selections in a heterogeneous network that includes different types of relaying, i.e., IRS and active relay devices, and different communication technologies, i.e., terahertz and millimeter wave communications. 


\bibliographystyle{IEEEtran}
\bibliography{IRS_selection}{}

\begin{thebibliography}{10}
\providecommand{\url}[1]{#1}
\csname url@samestyle\endcsname
\providecommand{\newblock}{\relax}
\providecommand{\bibinfo}[2]{#2}
\providecommand{\BIBentrySTDinterwordspacing}{\spaceskip=0pt\relax}
\providecommand{\BIBentryALTinterwordstretchfactor}{4}
\providecommand{\BIBentryALTinterwordspacing}{\spaceskip=\fontdimen2\font plus
\BIBentryALTinterwordstretchfactor\fontdimen3\font minus
  \fontdimen4\font\relax}
\providecommand{\BIBforeignlanguage}[2]{{%
\expandafter\ifx\csname l@#1\endcsname\relax
\typeout{** WARNING: IEEEtran.bst: No hyphenation pattern has been}%
\typeout{** loaded for the language `#1'. Using the pattern for}%
\typeout{** the default language instead.}%
\else
\language=\csname l@#1\endcsname
\fi
#2}}
\providecommand{\BIBdecl}{\relax}
\BIBdecl

\bibitem{gao2020reflection}
Y.~Gao, C.~Yong, Z.~Xiong, J.~Zhao, Y.~Xiao, and D.~Niyato, ``Reflection
  resource management for intelligent reflecting surface aided wireless
  networks,'' \emph{arXiv preprint arXiv:2002.00331}, 2020.

\bibitem{gao2020stackelberg}
Y.~Gao, C.~Yong, Z.~Xiong, D.~Niyato, Y.~Xiao, and J.~Zhao, ``A stackelberg
  game approach to resource allocation for intelligent reflecting surface aided
  communications,'' \emph{arXiv preprint arXiv:2003.06640}, 2020.

\bibitem{hofbauer2003evolutionary}
J.~Hofbauer and K.~Sigmund, ``Evolutionary game dynamics,'' \emph{Bulletin of
  the American mathematical society}, vol.~40, no.~4, pp. 479--519, 2003.

\bibitem{gao2019dynamic}
X.~Gao, S.~Feng, D.~Niyato, P.~Wang, K.~Yang, and Y.-C. Liang, ``Dynamic access
  point and service selection in backscatter-assisted rf-powered cognitive
  networks,'' \emph{IEEE Internet of Things Journal}, vol.~6, no.~5, pp.
  8270--8283, 2019.

\bibitem{chen2019sum}
W.~Chen, X.~Ma, Z.~Li, and N.~Kuang, ``Sum-rate maximization for intelligent
  reflecting surface based terahertz communication systems,'' in \emph{2019
  IEEE/CIC International Conference on Communications Workshops in China (ICCC
  Workshops)}.\hskip 1em plus 0.5em minus 0.4em\relax IEEE, 2019, pp. 153--157.

\bibitem{ma2020joint}
X.~Ma, Z.~Chen, W.~Chen, Z.~Li, Y.~Chi, C.~Han, and S.~Li, ``Joint channel
  estimation and data rate maximization for intelligent reflecting surface
  assisted terahertz mimo communication systems,'' \emph{IEEE Access}, 2020.

\bibitem{han2012game}
Z.~Han, D.~Niyato, W.~Saad, T.~Ba{\c{s}}ar, and A.~Hj{\o}rungnes, \emph{Game
  theory in wireless and communication networks: theory, models, and
  applications}.\hskip 1em plus 0.5em minus 0.4em\relax Cambridge university
  press, 2012.

\bibitem{quijano2017role}
N.~Quijano, C.~Ocampo-Martinez, J.~Barreiro-Gomez, G.~Obando, A.~Pantoja, and
  E.~Mojica-Nava, ``The role of population games and evolutionary dynamics in
  distributed control systems: The advantages of evolutionary game theory,''
  \emph{IEEE Control Systems Magazine}, vol.~37, no.~1, pp. 70--97, 2017.

\bibitem{niyato2008dynamics}
D.~Niyato and E.~Hossain, ``Dynamics of network selection in heterogeneous
  wireless networks: An evolutionary game approach,'' \emph{IEEE transactions
  on vehicular technology}, vol.~58, no.~4, 2008.

\bibitem{forte2014handbook}
F.~Forte, R.~Mudambi, and P.~M. Navarra, \emph{A handbook of alternative
  theories of public economics}.\hskip 1em plus 0.5em minus 0.4em\relax Edward
  Elgar Publishing, 2014.

\bibitem{tarasova2018concept}
V.~V. Tarasova and V.~E. Tarasov, ``Concept of dynamic memory in economics,''
  \emph{Communications in Nonlinear Science and Numerical Simulation}, vol.~55,
  pp. 127--145, 2018.

\bibitem{tarasova2016elasticity}
------, ``Elasticity for economic processes with memory: Fractional
  differential calculus approach,'' \emph{Fractional Differential Calculus},
  vol.~6, no.~2, pp. 219--232, 2016.

\bibitem{feng2020dynamic}
S.~Feng, D.~Niyato, X.~Lu, P.~Wang, and D.~I. Kim, ``Dynamic game and pricing
  for data sponsored 5g systems with memory effect,'' \emph{IEEE Journal on
  Selected Areas in Communications}, vol.~38, no.~4, pp. 750--765, 2020.

\bibitem{wu2019intelligent}
Q.~Wu and R.~Zhang, ``Intelligent reflecting surface enhanced wireless network
  via joint active and passive beamforming,'' \emph{IEEE Transactions on
  Wireless Communications}, vol.~18, no.~11, pp. 5394--5409, 2019.

\bibitem{lin2015indoor}
C.~Lin and G.~Y. Li, ``Indoor terahertz communications: How many antenna arrays
  are needed?'' \emph{IEEE Transactions on Wireless Communications}, vol.~14,
  no.~6, pp. 3097--3107, 2015.

\bibitem{han2014multi}
C.~Han, A.~O. Bicen, and I.~F. Akyildiz, ``Multi-ray channel modeling and
  wideband characterization for wireless communications in the terahertz
  band,'' \emph{IEEE Transactions on Wireless Communications}, vol.~14, no.~5,
  pp. 2402--2412, 2014.

\bibitem{jornet2011channel}
J.~M. Jornet and I.~F. Akyildiz, ``Channel modeling and capacity analysis for
  electromagnetic wireless nanonetworks in the terahertz band,'' \emph{IEEE
  Transactions on Wireless Communications}, vol.~10, no.~10, pp. 3211--3221,
  2011.

\bibitem{picard_theorem}
D.~Gutermuth, ``Picard’s existence and uniqueness theorem,'' \emph{notes of
  Fundamental of Differential equations. https://embedded. eecs. berkeley.
  edu/eecsx44/lectures7Spring2013/Picard. pdf}.

\bibitem{ciesielski2007stefan}
K.~Ciesielski \emph{et~al.}, ``On stefan banach and some of his results,''
  \emph{Banach Journal of Mathematical Analysis}, vol.~1, no.~1, pp. 1--10,
  2007.

\bibitem{nelson2019weierstrass}
B.~Nelson, ``The weierstrass function,'' \emph{Retrieved December}, vol.~31,
  2019.

\bibitem{yu2019miso}
X.~Yu, D.~Xu, and R.~Schober, ``Miso wireless communication systems via
  intelligent reflecting surfaces,'' in \emph{IEEE International Conference on
  Communications in China}, 2019, pp. 735--740.

\bibitem{gopalsamy2013stability}
K.~Gopalsamy, \emph{Stability and oscillations in delay differential equations
  of population dynamics}.\hskip 1em plus 0.5em minus 0.4em\relax Springer
  Science \& Business Media, 2013.

\bibitem{kang1972numerical}
S.~Kang and J.~B. Cheek, \emph{Numerical solution of differential
  equations}.\hskip 1em plus 0.5em minus 0.4em\relax Waterways Experiment
  Station, 1972.

\bibitem{gao2016fast}
X.~Gao, L.~Dai, Y.~Zhang, T.~Xie, X.~Dai, and Z.~Wang, ``Fast channel tracking
  for terahertz beamspace massive mimo systems,'' \emph{IEEE Transactions on
  Vehicular Technology}, vol.~66, no.~7, pp. 5689--5696, 2016.

\end{thebibliography}


\begin{thebibliography}{10}
\providecommand{\url}[1]{#1}
\csname url@samestyle\endcsname
\providecommand{\newblock}{\relax}
\providecommand{\bibinfo}[2]{#2}
\providecommand{\BIBentrySTDinterwordspacing}{\spaceskip=0pt\relax}
\providecommand{\BIBentryALTinterwordstretchfactor}{4}
\providecommand{\BIBentryALTinterwordspacing}{\spaceskip=\fontdimen2\font plus
\BIBentryALTinterwordstretchfactor\fontdimen3\font minus
  \fontdimen4\font\relax}
\providecommand{\BIBforeignlanguage}[2]{{%
\expandafter\ifx\csname l@#1\endcsname\relax
\typeout{** WARNING: IEEEtran.bst: No hyphenation pattern has been}%
\typeout{** loaded for the language `#1'. Using the pattern for}%
\typeout{** the default language instead.}%
\else
\language=\csname l@#1\endcsname
\fi
#2}}
\providecommand{\BIBdecl}{\relax}
\BIBdecl

\bibitem{wu2019towards}
Q.~Wu and R.~Zhang, ``Towards smart and reconfigurable environment: Intelligent
  reflecting surface aided wireless network,'' \emph{IEEE Communications
  Magazine}, vol.~58, no.~1, pp. 106--112, Jan. 2019.

\bibitem{di2019smart}
M.~Di~Renzo \emph{et~al.}, ``Smart radio environments empowered by
  reconfigurable ai meta-surfaces: An idea whose time has come,'' \emph{EURASIP
  Journal on Wireless Communications and Networking}, vol. 2019, no.~1, pp.
  1--20, May 2019.

\bibitem{chen2019sum}
W.~Chen, X.~Ma, Z.~Li, and N.~Kuang, ``Sum-rate maximization for intelligent
  reflecting surface based terahertz communication systems,'' in \emph{2019
  IEEE/CIC International Conference on Communications Workshops in China (ICCC
  Workshops)}.\hskip 1em plus 0.5em minus 0.4em\relax IEEE, 2019, pp. 153--157.

\bibitem{ma2020intelligent}
X.~Ma, Z.~Chen, W.~Chen, Y.~Chi, Z.~Li, C.~Han, and Q.~Wen, ``Intelligent
  reflecting surface enhanced indoor terahertz communication systems,''
  \emph{Nano Communication Networks}, vol.~24, p. 100284, 2020.

\bibitem{ma2020joint}
X.~Ma, Z.~Chen, W.~Chen, Z.~Li, Y.~Chi, C.~Han, and S.~Li, ``Joint channel
  estimation and data rate maximization for intelligent reflecting surface
  assisted terahertz mimo communication systems,'' \emph{IEEE Access}, 2020.

\bibitem{hao2020robust}
W.~Hao, G.~Sun, M.~Zeng, Z.~Zhu, Z.~Chu, O.~A. Dobre, and P.~Xiao, ``Robust
  design for intelligent reflecting surface assisted mimo-ofdma terahertz
  communications,'' \emph{arXiv preprint arXiv:2009.05893}, 2020.

\bibitem{Ma2020archive}
X.~Ma, Z.~Chen, W.~Chen, Y.~Chi, L.~Yan, C.~Han, and S.~Li, ``Joint hardware
  design and capacity analysis for intelligent reflecting surface enabled
  terahertz mimo communications,'' \emph{arXiv preprint arXiv:2012.06993},
  2020.

\bibitem{gao2020stackelberg}
Y.~Gao, C.~Yong, Z.~Xiong, D.~Niyato, Y.~Xiao, and J.~Zhao, ``A stackelberg
  game approach to resource allocation for intelligent reflecting surface aided
  communications,'' \emph{arXiv preprint arXiv:2003.06640}, 2020.

\bibitem{gao2020resource}
------, ``Resource allocation for intelligent reflecting surface aided
  cooperative communications,'' \emph{arXiv preprint arXiv:2012.10229}, 2020.

\bibitem{hofbauer2003evolutionary}
J.~Hofbauer and K.~Sigmund, ``Evolutionary game dynamics,'' \emph{Bulletin of
  the American mathematical society}, vol.~40, no.~4, pp. 479--519, 2003.

\bibitem{han2012game}
Z.~Han, D.~Niyato, W.~Saad, T.~Ba{\c{s}}ar, and A.~Hj{\o}rungnes, \emph{Game
  theory in wireless and communication networks: theory, models, and
  applications}.\hskip 1em plus 0.5em minus 0.4em\relax Cambridge university
  press, 2012.

\bibitem{quijano2017role}
N.~Quijano, C.~Ocampo-Martinez, J.~Barreiro-Gomez, G.~Obando, A.~Pantoja, and
  E.~Mojica-Nava, ``The role of population games and evolutionary dynamics in
  distributed control systems: The advantages of evolutionary game theory,''
  \emph{IEEE Control Systems Magazine}, vol.~37, no.~1, pp. 70--97, 2017.

\bibitem{liu2018evolutionary}
X.~Liu, W.~Wang, D.~Niyato, N.~Zhao, and P.~Wang, ``Evolutionary game for
  mining pool selection in blockchain networks,'' \emph{IEEE Wireless
  Communications Letters}, vol.~7, no.~5, pp. 760--763, 2018.

\bibitem{niyato2008dynamics}
D.~Niyato and E.~Hossain, ``Dynamics of network selection in heterogeneous
  wireless networks: An evolutionary game approach,'' \emph{IEEE transactions
  on vehicular technology}, vol.~58, no.~4, 2008.

\bibitem{gao2019dynamic}
X.~Gao, S.~Feng, D.~Niyato, P.~Wang, K.~Yang, and Y.-C. Liang, ``Dynamic access
  point and service selection in backscatter-assisted rf-powered cognitive
  networks,'' \emph{IEEE Internet of Things Journal}, vol.~6, no.~5, pp.
  8270--8283, 2019.

\bibitem{gao2020reflection}
Y.~Gao, C.~Yong, Z.~Xiong, J.~Zhao, Y.~Xiao, and D.~Niyato, ``Reflection
  resource management for intelligent reflecting surface aided wireless
  networks,'' \emph{arXiv preprint arXiv:2002.00331}, 2020.

\bibitem{feng2020dynamic}
S.~Feng, D.~Niyato, X.~Lu, P.~Wang, and D.~I. Kim, ``Dynamic game and pricing
  for data sponsored 5g systems with memory effect,'' \emph{IEEE Journal on
  Selected Areas in Communications}, vol.~38, no.~4, pp. 750--765, 2020.

\bibitem{wu2019intelligent}
Q.~Wu and R.~Zhang, ``Intelligent reflecting surface enhanced wireless network
  via joint active and passive beamforming,'' \emph{IEEE Transactions on
  Wireless Communications}, vol.~18, no.~11, pp. 5394--5409, 2019.

\bibitem{zhou2020intelligent}
G.~Zhou, C.~Pan, H.~Ren, K.~Wang, and A.~Nallanathan, ``Intelligent reflecting
  surface aided multigroup multicast miso communication systems,'' \emph{IEEE
  Transactions on Signal Processing}, to appear.

\bibitem{liaskos2019joint}
C.~Liaskos, A.~Tsioliaridou, A.~Pitilakis, G.~Pirialakos, O.~Tsilipakos,
  A.~Tasolamprou, N.~Kantartzis, S.~Ioannidis, M.~Kafesaki, A.~Pitsillides
  \emph{et~al.}, ``Joint compressed sensing and manipulation of wireless
  emissions with intelligent surfaces,'' in \emph{International Conference on
  Distributed Computing in Sensor Systems (DCOSS)}, 2019, pp. 318--325.

\bibitem{lin2015indoor}
C.~Lin and G.~Y. Li, ``Indoor terahertz communications: How many antenna arrays
  are needed?'' \emph{IEEE Transactions on Wireless Communications}, vol.~14,
  no.~6, pp. 3097--3107, 2015.

\bibitem{han2014multi}
C.~Han, A.~O. Bicen, and I.~F. Akyildiz, ``Multi-ray channel modeling and
  wideband characterization for wireless communications in the terahertz
  band,'' \emph{IEEE Transactions on Wireless Communications}, vol.~14, no.~5,
  pp. 2402--2412, 2014.

\bibitem{jornet2011channel}
J.~M. Jornet and I.~F. Akyildiz, ``Channel modeling and capacity analysis for
  electromagnetic wireless nanonetworks in the terahertz band,'' \emph{IEEE
  Transactions on Wireless Communications}, vol.~10, no.~10, pp. 3211--3221,
  2011.

\bibitem{picard_theorem}
D.~Gutermuth, ``Picard’s existence and uniqueness theorem,'' \emph{notes of
  Fundamental of Differential equations. https://embedded. eecs. berkeley.
  edu/eecsx44/lectures7Spring2013/Picard. pdf}.

\bibitem{ciesielski2007stefan}
K.~Ciesielski \emph{et~al.}, ``On stefan banach and some of his results,''
  \emph{Banach Journal of Mathematical Analysis}, vol.~1, no.~1, pp. 1--10,
  2007.

\bibitem{yu2019miso}
X.~Yu, D.~Xu, and R.~Schober, ``Miso wireless communication systems via
  intelligent reflecting surfaces,'' in \emph{IEEE International Conference on
  Communications in China}, 2019, pp. 735--740.

\bibitem{gopalsamy2013stability}
K.~Gopalsamy, \emph{Stability and oscillations in delay differential equations
  of population dynamics}.\hskip 1em plus 0.5em minus 0.4em\relax Springer
  Science \& Business Media, 2013.

\bibitem{tarasova2018concept}
V.~V. Tarasova and V.~E. Tarasov, ``Concept of dynamic memory in economics,''
  \emph{Communications in Nonlinear Science and Numerical Simulation}, vol.~55,
  pp. 127--145, 2018.

\bibitem{tarasova2016generalization}
V.~Tarasova and V.~Tarasov, ``A generalization of the concepts of the
  accelerator and multiplier to take into account of memory effects in
  macroeconomics,'' \emph{J. Econ. Entrep}, vol.~10, pp. 1121--1129, 2016.

\bibitem{tarasova2017logistic}
V.~V. Tarasova and V.~E. Tarasov, ``Logistic map with memory from economic
  model,'' \emph{Chaos, Solitons \& Fractals}, vol.~95, pp. 84--91, 2017.

\bibitem{kang1972numerical}
S.~Kang and J.~B. Cheek, \emph{Numerical solution of differential
  equations}.\hskip 1em plus 0.5em minus 0.4em\relax Waterways Experiment
  Station, 1972.

\bibitem{gao2016fast}
X.~Gao, L.~Dai, Y.~Zhang, T.~Xie, X.~Dai, and Z.~Wang, ``Fast channel tracking
  for terahertz beamspace massive mimo systems,'' \emph{IEEE Transactions on
  Vehicular Technology}, vol.~66, no.~7, pp. 5689--5696, 2016.

\end{thebibliography}



\end{document}